\documentclass[reprint,amsfonts, amssymb, amsmath,  showkeys, nofootinbib,pra, superscriptaddress, twocolumn,longbibliography]{revtex4-2}
\usepackage{float}
\makeatletter
\let\newfloat\newfloat@ltx
\makeatother
\usepackage[english]{babel}
\usepackage[utf8]{inputenc}
\usepackage{graphics}
\usepackage{selinput}

\usepackage{braket}
\usepackage{amsthm}
\usepackage{mathtools}
\usepackage{physics}
\usepackage{xcolor}
\usepackage{graphicx}
\usepackage[left=16mm,right=16mm,top=35mm,columnsep=15pt]{geometry} 
\usepackage{adjustbox}
\usepackage{placeins}
\usepackage[T1]{fontenc}
\usepackage{lipsum}
\usepackage{csquotes}
\usepackage{tikz}

\usepackage{tcolorbox}

\newtcolorbox[auto counter]{pabox}[2][]{fonttitle=\bfseries,
title=Example~\thetcbcounter: #2,#1,colframe=gray}

\usepackage{braket}
\usepackage{algorithm}
\usepackage[noend]{algpseudocode}
\def\Cbb{\mathbb{C}}

\def\HC{\mathcal{H}}

\def\ad{^{\dagger}}

\usepackage[makeroom]{cancel}
\usepackage[toc,page]{appendix}
\usepackage[colorlinks=true,citecolor=blue,linkcolor=magenta,linktocpage=true]{hyperref}
\newcommand{\dya}[1]{\ket{#1}\!\bra{#1}}

\newcommand{\Rbb}{\mathbb{R}}

\newcommand{\Zbb}{\mathbb{Z}}

\newcommand{\DC}{\mathcal{D}}

\newcommand{\OC}{\mathcal{O}}

\newcommand{\RC}{\mathcal{R}}
\newcommand{\SC}{\mathcal{S}}
\newcommand{\TC}{\mathcal{T}}

\newcommand{\WC}{\mathcal{W}}

\newcommand{\YC}{\mathcal{Y}}

\newcommand{\Span}{{\rm span}}

\renewcommand{\geq}{\geqslant}
\renewcommand{\leq}{\leqslant}

\newcommand{\SWAP}{\mathrm{SWAP}}

\renewcommand{\vec}[1]{\boldsymbol{#1}}  % Bold vectors instead of arrow vectors

\newcommand*{\id}{\openone}
\def\idty{\mathds{1}}

\newcommand{\bs}{\textsf{BS}}

\newcommand{\thv}{\vec{\theta}}

\def\be{\begin{equation}}
\def\ee{\end{equation}}
\def\bs{\begin{split}}
\def\e{\end{split}}
\def\ba{\begin{eqnarray}}
\def\bea{\begin{eqnarray}}

\def\tea{\end{eqnarray}}
\def\ea{\end{eqnarray}}
\def\eea{\end{eqnarray}}

\def\R{\mathds{R}}

\def\liea{\mathfrak{k}}
\def\liea{\mathfrak{g}}
\def\su{\mathfrak{s}\mathfrak{u}}

\def\liea{\mathfrak{k}}
\def\liea{\mathfrak{g}}

\def\su{\mathfrak{su}}

\newtheorem{theorem}{Theorem}

\newtheorem{corollary}{Corollary}
\newtheorem{observation}{Observation}

\newtheorem{proposition}{Proposition}

\newtheorem{definition}{Definition}
\usepackage{amssymb}
\usepackage{dsfont}

\def\be{\begin{equation}}
\def\te{\end{equation}}
\def\ee{\end{equation}}
\def\ba{\begin{eqnarray}}
\def\bea{\begin{eqnarray}}

\def\tea{\end{eqnarray}}
\def\ea{\end{eqnarray}}
\def\eea{\end{eqnarray}}

\newcommand{\paran}[1]{\left( #1 \right)}
\def\idty{\mathds{1}}

\definecolor{myblue}{RGB}{0,163,243}
\definecolor{myred}{RGB}{255,100,100}
\definecolor{mygreen}{RGB}{0,153,0}
\definecolor{mypurple1}{RGB}{46, 156, 202}
\definecolor{mypurple2}{RGB}{190,90,175}
\definecolor{fundamental}{RGB}{55, 110, 111}
\definecolor{tensor}{RGB}{161, 195, 209}

\newtcolorbox[use counter from=pabox]{red_boxed_example}[2][]{colback=myred!5!white,colframe=myred!75!black,fonttitle=\bfseries,floatplacement=h!t,float,title=Example~\thetcbcounter: #2,#1}
\newtcolorbox[use counter from=pabox]{green_boxed_example}[2][]{colback=mygreen!5!white,colframe=mygreen!75!black,fonttitle=\bfseries,floatplacement=h!t,float,title=Example~\thetcbcounter: #2,#1}

\newtcolorbox[use counter from=pabox]{purple_boxed_example1}[2][]{colback=mypurple1!5!white,colframe=mypurple1!75!black,fonttitle=\bfseries,floatplacement=h!t,float,
title=Example~\thetcbcounter: #2,#1}
\newtcolorbox[use counter from=pabox]{purple_boxed_example2}[2][]{colback=mypurple2!5!white,colframe=mypurple2!75!black,fonttitle=\bfseries,floatplacement=h!t,float,title=Example~\thetcbcounter: #2,#1}

\newtcolorbox[use counter from=pabox]{purple_boxed_example}[2][]{%
colback=mypurple2!5!white,colframe=mypurple2!75!black,fonttitle=\bfseries,floatplacement=h!t,float,
title=Example~\thetcbcounter: #2,#1}

\newtcolorbox[use counter from=pabox]{blue_boxed_example}[2][]{%
colback=myblue!5!white,colframe=myblue!75!black,fonttitle=\bfseries,floatplacement=h!t,float,
title=Example~\thetcbcounter: #2,#1}

\newtcolorbox[use counter from=pabox]{fundamental_boxed_example}[2][]{%
colback=fundamental!5!white,colframe=fundamental!75!black,fonttitle=\bfseries,floatplacement=h!t,float,
title=Example~\thetcbcounter: #2,#1}

\newtcolorbox[use counter from=pabox]{tensor_boxed_example}[2][]{%
colback=tensor!5!white,colframe=tensor!75!black,fonttitle=\bfseries,floatplacement=h!t,float,
title=Example~\thetcbcounter: #2,#1}

\newcommand{\umatxone}{\boxed{\begin{matrix} U_1 & & 0 \\ 
& \ddots &  \\
0 & & U_1 \end{matrix}}}
\newcommand{\umatxk}{\boxed{\begin{matrix} U_k & & 0 \\ 
& \ddots &  \\
0 & & U_k \end{matrix}}}
\newcommand{\umatxK}{\boxed{\begin{matrix} U_K & & 0 \\ 
& \ddots &  \\
0 & & U_K \end{matrix}}}

\newcommand{\bmatxone}{\boxed{\begin{matrix}  & &  \\ 
& \mathcal{B}_{m_1}\otimes \idty_{\dim(U_1)} &  \\
 & &  \end{matrix}}}
\newcommand{\bmatxk}{\boxed{\begin{matrix}  & &  \\ 
& \mathcal{B}_{m_k}\otimes \idty_{\dim(U_k)} &  \\
 & &  \end{matrix}}}
\newcommand{\bmatxK}{\boxed{\begin{matrix}  & &  \\ 
& \mathcal{B}_{m_K}\otimes \idty_{\dim(U_K)} &  \\
 & &  \end{matrix}}}

\usepackage[normalem]{ulem}

\begin{document}

\title{Representation Theory for Geometric Quantum Machine Learning}

\author{Michael Ragone}
\affiliation{Theoretical Division, Los Alamos National Laboratory, Los Alamos, New Mexico 87545, USA}
\affiliation{Department of Mathematics, University of California Davis, Davis, California 95616, USA}

\author{Paolo Braccia}
\affiliation{Theoretical Division, Los Alamos National Laboratory, Los Alamos, New Mexico 87545, USA}
\affiliation{Dipartimento di Fisica e Astronomia, Universit\`{a} di Firenze, Sesto Fiorentino (FI), 50019 , Italy}

\author{Quynh T. Nguyen}
\affiliation{Theoretical Division, Los Alamos National Laboratory, Los Alamos, New Mexico 87545, USA}
\affiliation{Harvard Quantum Initiative, Harvard University,
Cambridge, Massachusetts 02138, USA}

\author{Louis Schatzki}
\affiliation{Information Sciences, Los Alamos National Laboratory, Los Alamos, New Mexico 87545, USA}
\affiliation{Department of Electrical and Computer Engineering, University of Illinois at Urbana-Champaign, Urbana, Illinois 61801, USA}

\author{Patrick J. Coles}
\affiliation{Theoretical Division, Los Alamos National Laboratory, Los Alamos, New Mexico 87545, USA}

\author{Fr\'{e}d\'{e}ric Sauvage}
\affiliation{Theoretical Division, Los Alamos National Laboratory, Los Alamos, New Mexico 87545, USA}

\author{Mart\'{i}n Larocca}
\affiliation{Theoretical Division, Los Alamos National Laboratory, Los Alamos, New Mexico 87545, USA}
\affiliation{Center for Nonlinear Studies, Los Alamos National Laboratory, Los Alamos, New Mexico 87545, USA}

\author{M. Cerezo}
\affiliation{Information Sciences, Los Alamos National Laboratory, Los Alamos, New Mexico 87545, USA}

\begin{abstract}
Recent advances in classical machine learning have shown that creating models with inductive biases  encoding the symmetries of a problem can greatly improve performance. Importation of these ideas, combined with an existing rich body of work at the nexus of quantum theory and symmetry, has given rise to the field of Geometric Quantum Machine Learning (GQML). Following the success of its classical counterpart, it is reasonable to expect that GQML will play a crucial role in developing problem-specific and quantum-aware models capable of achieving a computational advantage. Despite the simplicity of the main idea of GQML -- create architectures respecting the symmetries of the data -- its practical implementation requires a significant amount of knowledge of group representation theory. We present an introduction to representation theory tools from the optics of quantum learning, driven by key examples involving discrete and continuous groups. These examples are sewn together by an exposition outlining the formal capture of GQML symmetries via ``label invariance under the action of a group representation'', a brief (but rigorous) tour through finite and compact Lie group representation theory, a reexamination of ubiquitous tools like Haar integration and twirling, and an overview of some successful strategies for detecting symmetries.  
\end{abstract}

\maketitle
{
  \hypersetup{linkcolor=blue}
  \tableofcontents
}

\section{Introduction}

Quantum Machine Learning (QML) has recently emerged as one of the most promising candidates to make practical use of quantum computers~\cite{biamonte2017quantum,cerezo2020variationalreview,schuld2021machine,cerezo2022challenges}. 
By leveraging quantum systems to process information, QML presents the ultimate framework for data analysis~\cite{cerezo2022challenges}. In the near-term~\cite{preskill2018quantum}, QML aims to achieve a quantum advantage, i.e.,  to solve certain tasks exponentially faster than any classical supercomputer. In the fault-tolerant era, QML will be the most natural choice to learn from the data created by large-scale quantum devices.

Despite some promising results~\cite{huang2021provably,abbas2020power,caro2021generalization,havlivcek2019supervised,liu2021rigorous,huang2021quantum,sharma2020reformulation}, the field of QML is in its infancy. One of the main challenges for QML is developing quantum-aware and problem-specific models, as it has been shown that simply translating classical architectures into quantum ones, or using problem-agnostic architectures, can lead to serious issues that hinder the performance of the QML model~\cite{mcclean2018barren,cerezo2020cost,sharma2020trainability,patti2020entanglement,marrero2020entanglement,holmes2021connecting,thanasilp2021subtleties,larocca2021diagnosing,arrasmith2021equivalence}.

\begin{figure*}[th]
    \centering
    \includegraphics[width=1\linewidth]{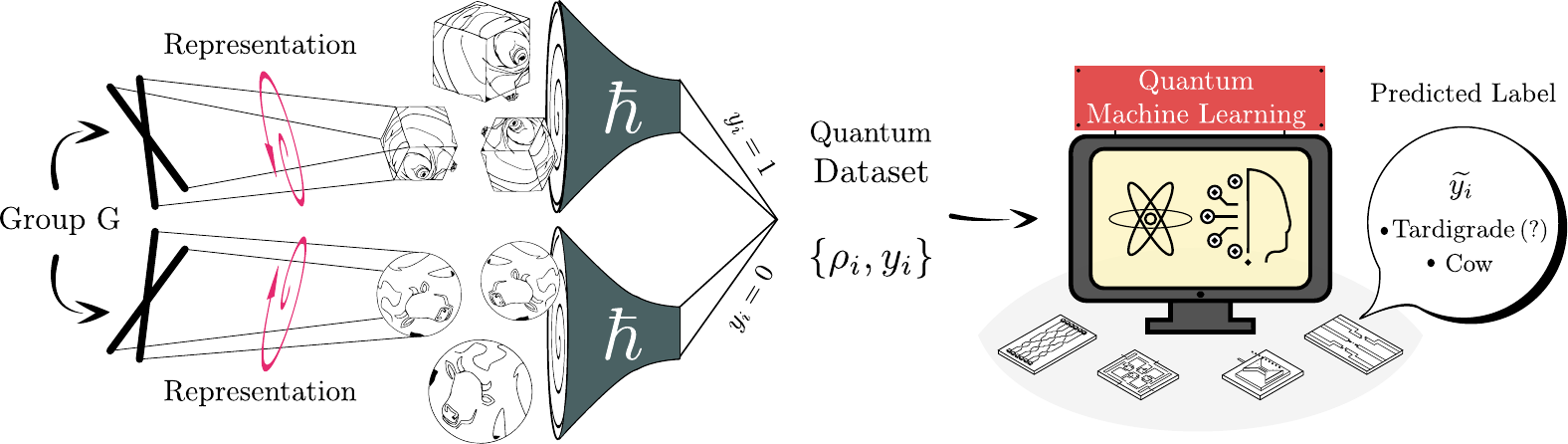}
    \caption{\textbf{Symmetries in QML.} Consider the QML task of classifying between spherical cows and cubic tardigrades~\cite{lee2021entanglement,vedral2021microscopic} whose information has been encoded into quantum states~\cite{havlivcek2019supervised}. Note that a rotated spherical cow is still a spherical cow, and as such the QML model should be able to accurately classify it regardless of how the cow was rotated. In this case a rotation in three-dimensions corresponds to a symmetry of the dataset, as the group of rotations do not change the labels.    }
    \label{fig:fig1}
\end{figure*}

Recently, and inspired by the tremendous success of classical geometric deep learning~\cite{bronstein2021geometric}, there have been several efforts to create QML models with strong inductive biases that respect the underlying structure and symmetries of the data over which they act~\cite{larocca2022group,skolik2022equivariant,meyer2022exploiting,glick2021covariant,zheng2021speeding,sauvage2022building,mernyei2022equivariant}. The goal of inductive biases is to restrict the space of functions explored by the QML model by imposing task-specific knowledge or assumptions. It is expected that models tailored to a given task will have better performance, both in training and generalization, than those without inductive biases~\cite{bronstein2021geometric,larocca2022group,skolik2022equivariant,meyer2022exploiting,glick2021covariant,zheng2021speeding,sauvage2022building, mernyei2022equivariant,schatzki2022theoretical}.  These efforts have led to the inception of the field of \textit{Geometric Quantum Machine Learning} (GQML). We remark that the scope of GQML is quite broad and is relevant to quantum deep learning, quantum kernels, quantum generative modeling, and other related topics.

A rich source of inductive biases in GQML arises from the analysis of symmetries in datasets. Symmetries are mathematically captured by group theory and representation theory, and so GQML researchers will need to be armed with their fundamentals to attack symmetric tasks. Group theory studies abstract algebraic objects called groups: sets equipped with a binary operation for combining elements satisfying certain properties. Representation theory, on the other hand, studies groups (and other algebraic structures) by representing them as collections of linear transformations acting upon a vector space. The two are intimately related, and a common motif is: \textit{groups encode abstract symmetries, and representations describe the concrete actions of these symmetries}. Their joint impact within physics has been deep and pervasive. Noether's theorem rests fundamentally upon continuous symmetries to extract conserved physical quantities~\cite{noether1918invariante}. The analysis of central potentials, notably the hydrogen atom, uses symmetries to perform a block-diagonalization which turns an intimidating 3D problem into a morally 1D problem~\cite{sakurai1995modern}. The development of quantum mechanics as a field was profoundly shaped by groups and their representations, and they continue to be central tools across quantum information theory~\cite{ritter2005quantum,bartlett2007reference,zanardi2000stabilizing,nielsen2000quantum,childs2010quantum,hayashi2017group}, the study of phase transitions and critical phenomena in condensed matter systems~\cite{dresselhaus2008application,onuki2002phase,pelissetto2002critical}, quantum field theory~\cite{feynman76theory,chriss1997representation,frohlich2006quantum}, and even the work towards understanding the exotic behavior of black holes~\cite{ashtekar2000quantum,domagala2004black,engle2010black,engle2010blackhole}. 

In the program of GQML, group and representation theory are crucial to manipulate the symmetries underlying the data, and to understand the interplay between these symmetries and the quantum learning process. For instance, representation theoretical tools can be used  to create \textit{equivariant quantum neural networks and measurement operators}, as well as to understand how \textit{different representations} of the same symmetry group can access \textit{different types of information} in a quantum state. 

While group theory is usually covered and studied in many undergraduate classes, the same cannot be said about representation theory, which is usually left as a topic for more advanced and specialized mathematical courses. Moreover, the literature for representation theory can sometimes be hard to access for non-experts as it is written from a purely mathematical and algebraic perspective. The previous has motivated us to present the basic mathematical tools for group and representation theory through the lens of QML. Our hope is to popularize the use of representation theory as a fundamental ingredient to both near-term and fault-tolerant QML model design. That being said, we note that this article is not aimed at being a comprehensive review of representation theory. Hence, we also point readers to  Refs.~\cite{hall2013lie,simon1996representations,fulton1991representation,serre1977linear} for more detailed presentations of representation theory. 

This work is aimed at readers with a (small but non-zero) background in QML. The topics covered are as follows. Section \ref{sec:quantum machine learning} briefly introduces us to QML and to the importance of identifying symmetries in a dataset. Section \ref{sec:symmetries and groups in QML} unravels a general framework to formally describe symmetries in QML, which Section \ref{sec:examples of discrete and continuous symmetries in QML} then puts to work through several concrete examples of discrete and continuous symmetries. Section \ref{sec:abstractifying physical symmetries to groups} sheds light on the link between symmetries, groups, and representations. Section \ref{sec:representation theory for discrete and continuous groups} dives deep into the world of finite and Lie group representation theory. This is the meatiest chunk of the paper, and contains important definitions, foundational theorems, and illustrations of guiding themes within representation theory. Section \ref{sec:some representation theory-rich constructions in QML} ties representation theory to key constructions and methods widely used within QML (such as Haar integrals and twirling), and Section \ref{sec:symmetries in the wild} finishes with a recapitulation of strategies to detect symmetries in a given problem. We emphasize that this text is written with a rich bank of examples as the guiding stars, and we encourage the reader to constantly return to them when lost in the seas of abstraction. Group representation examples are presented throughout the article in colored boxes (the same color indicates that we are studying the same symmetry group and the same representation), and QML-relevant examples appear across several figures. Finally, the end of the article is capped off by a sort of field guide with boxes containing important information and facts about several common symmetry groups.

\section{Quantum machine learning}\label{sec:quantum machine learning}

In this work we focus on the supervised QML problem of classifying labeled quantum data. However, we remark that the tools presented here can broadly be used in more general QML scenarios such as unsupervised learning, reinforcement learning, and generative modeling. For our purposes, we assume that one is given repeated access to a dataset of the form $\SC=\{\rho_i,y_i\}_{i=1}^N$, where $\rho_i$ are $n$-qubit states belonging to a data domain $\RC$ in a $d$-dimensional Hilbert space $\HC$ (with $d=2^n$), while $y_i$ are real-valued labels in some label domain $\YC$. We further assume that the data instances in $\SC$ are drawn i.i.d. from a distribution defined over $\RC\times\YC$, such that the label $y_i$ associated to the state $\rho_i$ is assigned according to some (unknown and potentially probabilistic) function $f : \RC \rightarrow \YC$. That is, $f(\rho_i) = y_i$. We note that we will make no particular assumption regarding to how the data states $\rho_i$ are created. That is, $\rho_i$ could be the result of embedding classical data into quantum states (QML for classical data~\cite{havlivcek2019supervised}), or they could be obtained from some physical quantum mechanical process (QML for quantum data~\cite{schatzki2021entangled}).

The goal is to train a  model $h_{\vec{\theta}}: \RC \rightarrow \YC$, with $\thv$ being trainable parameters,  to produce labels that match those of $f$ with high probability over the training set (small training error), but also over new and previously unseen cases (small generalization error). QML models come in many forms and flavors, leveraging both the power of  quantum hardware  (e.g., to compute some classically intractable expectation value over the input data) as well as that of classical computers (e.g., use an optimizer to train the parameters $\vec{\theta}$, or analyze the quantum measurement outcomes with some classical neural network). In any case, the success of the QML model hinges on several factors, but perhaps the most important one is how the model is defined, i.e., what are the inductive biases encoded in the model. As previously mentioned, the goal of GQML is to embed information about the symmetries of the data in $\SC$ into the model. Here we will not concern ourselves with how to actually create QML models encoding these symmetries (we refer the reader to Refs.~\cite{larocca2022group,skolik2022equivariant,meyer2022exploiting,glick2021covariant,zheng2021speeding,sauvage2022building,mernyei2022equivariant} for that), but rather we will present the tools to theoretically understand and handle the symmetries themselves. 

Throughout this work we will present several QML tasks and their respective symmetries. For the sake of simplicity, we will consider the case when the parameterized model $h_{\vec{\theta}}$ is simply given by taking $k$ copies of the input states from the dataset $\SC$, sending them through a parameterized channel (usually called a \textit{quantum neural network}), and making a measurement at the output. That is, we will focus on QML models of the form
\begin{equation}\label{eq:QML-model}
    h_{\vec{\theta}}(\rho_i)=\Tr[\WC_{\vec{\theta}}(\rho_i^{\otimes k}) M_i]\,,
\end{equation}
where $\WC_{\vec{\theta}}:B(\HC^{\otimes k })\rightarrow B(\HC^{\otimes k'})$  is a trainable parameterized quantum channel   (usually a unitary channel), and $M_i$ is a -- potentially data-dependent -- Hermitian measurement operator. Here, $B(\HC^{\otimes k })$ denotes the space of bounded linear operators on $\HC^{\otimes k })$.

\section{Symmetries and groups in QML}\label{sec:symmetries and groups in QML}

As shown in Fig.~\ref{fig:fig1}, the main goal of this article is to study symmetries in QML within the framework of representation theory. As such let us first define \textit{what a symmetry is}. At the highest level, a \textit{symmetry describes some property of the data in $\RC$ or of the underlying function $f$, which is left unchanged under some transformation}. We will consider that such transformation refers to a \textit{unitary} evolution applied to the quantum state, i.e., to a map $\rho\rightarrow U\rho U\ad$ for some $U$. As we will see below, in many cases this type of unitary transformation suffices to encompass a wide range of scenarios of interest.\footnote{We will encounter and work with symmetry representations that are ostensibly not unitary. However, Theorem \ref{thm:Weyl's unitary trick} will show that a wide class of representations are equivalent to unitary ones. Further, Wigner's theorem~\cite{wigner2012group} guarantees that all symmetry transformations of quantum states preserving inner products are either unitary or antiunitary, and often antiunitary transformations are ``unitary and complex conjugation''.} Let us now consider the following proposition.
\begin{proposition}~\label{prop:symm-group}
Let $G$ be the set of all unitary symmetry transformations, such that for any $U\in G$, the map  $\rho\rightarrow U\rho U\ad$ leaves some property of $\rho$ unchanged. Then, $G$ forms a group.
\end{proposition}

We note that given any two unitaries $U$ and $V$ in $G$, the unitary $V \cdot U$ obtained by multiplying $V$ and $U$ is also a symmetry transformation. This follows from the fact that concatenating two property-preserving transformations $\rho\rightarrow U\rho U\ad\rightarrow V \cdot U\rho U\ad \cdot V\ad$ constitutes in itself a property-preserving transformation. Since $G$ is a group, it satisfies the group axioms: associativity, existence of identity, and existence of inverse. 

\textit{Associativity}. Given any $U$, $V$ and $W$ in $G$, then $(W\cdot V)\cdot U= W\cdot (V\cdot U)$.

\textit{Identity}. There exists an element in $G$, corresponding to the $d\times d$ identity matrix $\id$, such that $\id\cdot U= U\cdot \id =U$.

\textit{Inverses}. For each $U$ in $G$, there exists an element $U\ad$ in $G$ such that $U \cdot U\ad=U\ad \cdot U =\id$, where $\id$ is the identity matrix, and $U\ad$ is the inverse (conjugate transpose) of $U$.

Since  symmetries are ubiquitous in physics,  here we will focus on those that are relevant for QML, i.e., those that preserve the labels of the data. As such we introduce the following definition.
\begin{definition}[Label invariance]\label{def:label-inv}
The action of the group $G$ is said to leave the data labels $y_i$ invariant, if 
\begin{equation}
    f(U\rho_i U\ad)=f(\rho_i)=y_i\,,
\end{equation}
for all $\rho_i$ with label $y_i$ and for all  $U\in G$.
\end{definition}

Here we make several important remarks. 
First, we note that if the states themselves are invariant, i.e., if $U \rho_i U^\dagger = \rho_i$, then label invariance is immediately satisfied. 
But this need not be the case: there are many interesting examples wherein the states are not invariant but their labels are: i.e., $U\rho_i U^\dagger \neq \rho_i$, but $f(U\rho_i U^\dagger) = f(\rho_i)$ (we will present some examples below!). 
This is a wider class of symmetries than the frequently encountered state symmetries in physics, and it captures a notion more QML-ish in spirit: \textit{the essential characteristic of a data point is its label, so label symmetries are the correct symmetries to track}. Second, just as a square has both reflection and rotation symmetry, data and their labels commonly support several different symmetry groups $G_i$. Further, endowing fixed data with different labels will often change the symmetry groups at play. This could lead, for instance, to a scenario where the data in different classes is associated with different symmetry groups. In all cases, it is up to the wisdom of the QML practitioner to both discover the symmetries in a given task and decide which of them to leverage. As usual in science, good examples will cultivate this wisdom, and we have thus tried to make this tutorial as example-driven as possible.

Recently, there have been several proposals within the nascent field of GQML to create QML models  $h_{\vec{\theta}}$ that respect the symmetries and label invariances of the problem at hand, as this can leads to models with less data-requirements, simpler training landscapes, and better generalization (see Refs.~\cite{bronstein2021geometric,larocca2022group,skolik2022equivariant,meyer2022exploiting,glick2021covariant,zheng2021speeding,sauvage2022building,mernyei2022equivariant,astrakhantsev2022algorithmic,astrakhantsev2022algorithmic,schatzki2022theoretical}).

While there are many ways to ensure that a QML model $h_{\vec{\theta}}$ as in Eq.~\eqref{eq:QML-model} is invariant under $G$, here we invoke a strategy requiring the following two conditions: \textit{equivariance} under $G$ of the parameterized quantum channel
\begin{equation}\label{eq:equiv}
    \WC_{\vec{\theta}}(U^{\otimes k}\rho_i^{\otimes k}(U\ad)^{\otimes k})= U^{\otimes k'}\WC_{\vec{\theta}}((\rho_i)^{\otimes k})(U\ad)^{\otimes k'}\!\!,\,\,\,  \forall U\in G\,,
\end{equation}
and equivariance of the measurement operator
\begin{equation}\label{eq:inv}
    [M,U^{\otimes k'}]= 0\,,  \quad \forall U\in G\,.
\end{equation}
One can readily verify that if Eqs.~\eqref{eq:equiv} and~\eqref{eq:inv} are satisfied, then  the model produces predicted labels that are \textit{invariant} under the action of $G$:
\begin{align}
    h_{\vec{\theta}}(U \rho_iU\ad)&=\Tr[\WC_{\vec{\theta}}((U\rho_iU\ad)^{\otimes k}) M_i]\nonumber\\
    &=\Tr[U^{\otimes k'}\WC_{\vec{\theta}}((\rho_i)^{\otimes k})(U\ad)^{\otimes k'} M_i]\nonumber\\
    &=\Tr[\WC_{\vec{\theta}}((\rho_i)^{\otimes k})(U\ad)^{\otimes k'} M_i U^{\otimes k'}]\nonumber\\
    &=\Tr[\WC_{\vec{\theta}}((\rho_i)^{\otimes k}) M_i ]\nonumber\\
    &=h_{\vec{\theta}}(\rho_i)\,, \quad \forall U\in G\,.
\end{align}
Conceptually, we can think of equivariant quantum neural networks as ``\textit{passing}'' the action of the symmetry from their input, to their output, while equivariant measurements lead to models that ``\textit{absorb}'' the action of the symmetry. We refer the reader to the existing GQML literature of Refs.~\cite{bronstein2021geometric,larocca2022group,skolik2022equivariant,meyer2022exploiting,glick2021covariant,zheng2021speeding,sauvage2022building,mernyei2022equivariant,astrakhantsev2022algorithmic,astrakhantsev2022algorithmic,schatzki2022theoretical} for additional details on the importance of equivariance and invariance, and their crucial role on improving the performance of QML models.

\section{Examples of discrete and continuous symmetries in QML}\label{sec:examples of discrete and continuous symmetries in QML}

Discrete groups, as the name implies, have underlying discrete sets, often finite sets or the set of integers $\mathbb{Z}$.
For instance, consider a QML task of where we want to classify one-dimensional real-valued data. Namely, as shown in Fig.~\ref{fig:fig2}(a), we want to distinguish blue circles (with label $y_i=0$) from orange triangles (with label $y_i=1$). Note that here the labels of the data are invariant under the transformation $x\rightarrow -x$. To translate this into a quantum problem, we encode the data in a single qubit state (see Fig.~\ref{fig:fig2}(a)), with a data-dependent rotation about the $Y$ axis of the Bloch sphere acting on the $\ket{+}$ state. One can verify that now the labels are invariant, for example, under a bit-flip transformation, corresponding to the group
\begin{equation}\label{eq:par_group}
    G_{\rm{bflip}}=\{\id,X\}\,,
\end{equation}
with $X$ the Pauli-$x$ operator. Clearly this group is discrete and contains two elements.

\begin{figure}
    \centering
    \includegraphics[width=0.9\columnwidth]{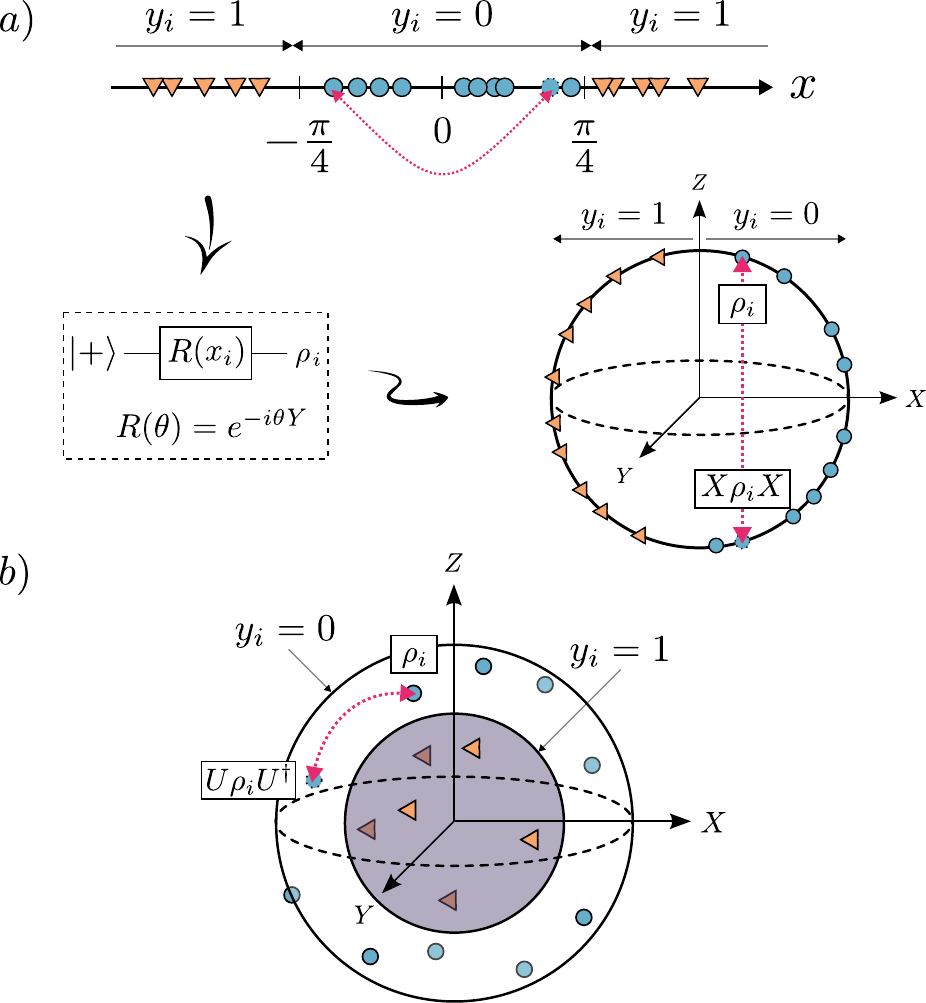}
    \caption{\textbf{Examples of discrete and continuous symmetry groups.} a) In this QML task the goal is to classify real valued data $x$. The data $x_i$ with label $y_i=0$ (blue circles) is sampled from the interval $[-\frac{\pi}{4},\frac{\pi}{4}]$, while the data $x_i$ with label $y_i=1$ (orange triangles) is sampled from $(-\frac{\pi}{2},-\frac{\pi}{4}]\cup[\frac{\pi}{4},\frac{\pi}{2})$. As indicated by the pink dashed arrow, the symmetry operation $x\rightarrow -x$ preserves the label. We encode this classical data in a quantum state by initializing a single qubit to the state $\ket{+}$ and performing a rotation about the $y$-axis with an angle $x_i$. As schematically depicted, the quantum states $\rho_i$ in the dataset are pure states living on the surface of the Bloch sphere. Here, the symmetry group that preserves the labels of the quantum states is $G_{\rm{bflip}}=\{\id,X\}$ of Eq.~\eqref{eq:par_group} (see pink dashed arrow over the Bloch sphere). b) In this QML task the goal is to classify single-qubit pure states from single-qubit mixed states. The data  $\rho_i$ with label $y_i=0$ (blue circles) correspond to pure states living on the surface of the Bloch sphere,  while the data $x_i$ with label $y_i=1$ (orange triangles) are mixed states living in a shell inside of the Bloch sphere. Since the purity is a spectral property, it gets preserved by the action of any unitary. As such, the symmetry group preserving the labels is $
    G_{\rm{uni}}=\{U\in SU(2)\}$ of Eq.~\eqref{eq:uni_group} (see pink dashed arrow over the Bloch sphere). Note that here the data is quantum mechanical in nature, as it does not correspond to classical data encoded in quantum states.     }
    \label{fig:fig2}
\end{figure}

Continuous groups, on the other hand, are also manifolds, which means that they are locally homeomorphic to an Euclidean space, and as such we can parameterize regions of the group by tuples of real numbers called ``coordinates''. 
For example, consider  the binary QML classification task of Fig.~\eqref{fig:fig2}(b), where we want to classify single-qubit pure states (blue circles on the surface of the Bloch sphere) with label $y_i=0$ from single-qubit mixed states (orange triangles in a shell inside of the Bloch sphere) with label $y_i=1$. It is easy to see that now the labels are invariant under the action of any unitary. This corresponds to the group
\begin{equation}\label{eq:uni_group}
    G_{\rm{uni}}=\{U\in SU(2)\}\,,
\end{equation}
where $SU(2)$ denotes the special unitary group of degree $2$ (all $2\times 2$ unitary matrices with determinant $1$). Here, $G_{\rm{uni}}$ is continuous as it is the set of all (infinitely many) unitaries acting on one qubit. More specifically,  we note that any element in $G_{\rm{uni}}$ can be expressed as $U=c_0\id+(c_1X+c_2Y+c_3Z)$, where $c_0\in \R$ and $\sum_{i=0}^3c_i^2=1$ are the real coordinates parametrizing the manifold of $2\times 2$ unitaries. \footnote{Geometry note: technically, we are abusing the term ``coordinates'', because we are thinking of $(c_0,c_1,c_2,c_3)$ as parametrizing $\R^4$ and identifying $U(2)$ as a submanifold.}

Continuous groups necessarily have uncountably many elements, in stark contrast to the often finite discrete groups. But since continuous groups are also manifolds, they contain additional structure lacked by discrete groups: We can construct smooth paths in the group, just as we can construct paths on surfaces. We can then take derivatives along paths in the group, just as we can take derivatives of paths on surfaces.
From here, analogously to solving ordinary differential equations, the exponential map allows us to ``integrate'' and lift from derivatives over time to paths in the group. 
Specifically, the structure that stores the information of ``directional derivatives of continuous group paths'' is called the \textit{Lie algebra} $\mathfrak{g}$ associated to the continuous group, which we henceforth refer to as a \textit{Lie group}. 
Exponentiation of every element $X$ of the Lie algebra leads to an element of the Lie group $e^{X}=g$. That is, $\mathfrak{g}=\{X\in\mathbb{C}^{d\times d} \,|\, e^X \in G \}$. The correspondences between Lie algebras and Lie groups will be fleshed out in Section~\ref{sec:there and back again: lie groups and algebras}, but we will first spend some time investigating the connection between symmetries  and groups.

%--------------------------------------------------
\section{Abstractifying physical symmetries to groups}\label{sec:abstractifying physical symmetries to groups}
Once the symmetries of the data have been identified, it is extremely useful to ``\textit{abstractify}'' them: \textit{connect a physical symmetry group with some familiar abstract mathematical group}. To bridge this language with the rest of the article, recall the earlier motif: \textit{groups} encode abstract symmetries, and \textit{representations} describe concrete symmetries. The main utility of this abstractification procedure is that groups as mathematical objects have been thoroughly studied since the early 19th century, and a wealth of information is readily available for scores of them.  Moreover, in the eyes of physics, the list of abstract groups is surprisingly short, thanks in large part to classification programs for finite groups and semisimple Lie groups and nature's seeming preferential treatment of these groups---this means that identification is direct in many cases. It is worth highlighting, however, that this procedure is highly heuristic and fully general approaches do not exist: one physical symmetry group (representation) can be identified with several ambient abstract groups. But there is commonly a ``simple'' choice of group to make, and this often suffices for the goals of abstractification.

\begin{figure}
    \centering
    \includegraphics[width=0.9\columnwidth]{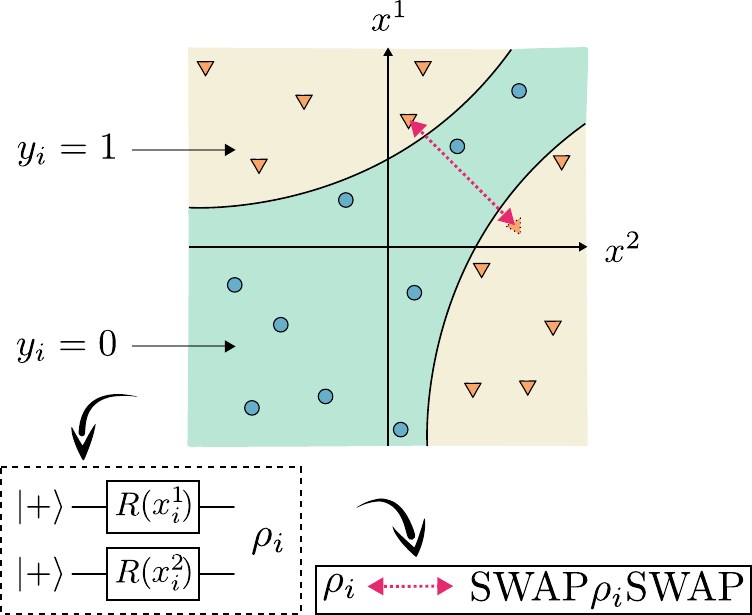}
    \caption{\textbf{Example of another discrete symmetry group.} In this QML task the goal is to classify real valued data $x=(x^1,x^2)$ living in a two-dimensional plane.  The data $x_i$ with label $y_i=0$ (blue circles) is sampled from the green region of the plane, while the data $x_i$ with label $y_i=1$ (orange triangles) is sampled from the yellow region of the plane. As indicated by the pink dashed arrow, the symmetry operation $(x^1,x^2)\rightarrow (x^2,x^1)$ preserves the label. We encode this classical data in  a two-qubit quantum state by initializing  the qubits to the state $\ket{+}\otimes \ket{+}$ and performing a rotation about the $y$-axis with a angle $x_i^1$ for the first qubit, and $x_i^2$ for the first qubit. Now, the symmetry group preserving the labels of the quantum states is $G_{\SWAP}=\{\id,\SWAP\}$ of Eq.~\eqref{eq:swap-group}.  }
    \label{fig:fig3}
\end{figure}

As an example, consider the aforementioned problem where the labels are invariant under a bit-flip transformation, i.e., under the group  $G_{\rm{bflip}}$ of Eq.~\eqref{eq:par_group}. To abstractify a small discrete group, it can be  useful to construct the group multiplication table, or Cayley table. Cayley tables show the multiplication of all group elements, which constitutes a type of fingerprint for a group. For instance, for the bit-flip symmetry group the Cayley table is as follows:
\begin{equation}
    \noindent\begin{tabular}{c | c c }
     & $\id$ & $X$   \\
    \cline{1-3}
    $\id$ & $\id$ & $X$  \\
    $X$ & $X$ & $\id$  
\end{tabular}\,.
\end{equation}
The only group consisting of 2 elements obeying this Cayley table is the cyclic group $\mathbb{Z}_2$, so we can identify this as the group underlying bit-flip symmetry. Interestingly, we can also consider the QML task of Fig.~\ref{fig:fig3} where we want to classify real-valued two-dimensional vectors $x=(x^1,x^2)$. Here the symmetry that preserves the labels is $(x^1,x^2)\rightarrow (x^2,x^1)$. We can verify that after encoding the data in a quantum state (as described in Fig.~\ref{fig:fig3}), the labels are invariant under the swapping of qubits.  Denoting as $\SWAP$ the operator that swaps two qubits (i.e., $\SWAP\ket{ij}=\ket{ji}$, for all $i,j$), the symmetry group is
\begin{equation}\label{eq:swap-group}
    G_{\SWAP}=\{\id,\SWAP\}\,,
\end{equation}
which is again discrete and also contains two elements. The Cayley table is found to be:
\begin{equation}
    \noindent\begin{tabular}{c | c c }
     & $\id$ & $\SWAP$   \\
    \cline{1-3}
    $\id$ & $\id$ & $\SWAP$  \\
    $\SWAP$ & $\SWAP$ & $\id$  
\end{tabular}\,.
\end{equation}
Notably, from this table we can also identify the qubit-reflection symmetry group with the $\mathbb{Z}_2$ group. 
Of course, as the size of the physical symmetry group grows, Cayley tables become intractable and more group-theoretic tools are required to probe candidate groups---in practice, one often has in mind some ``usual suspect'' groups (see Box~\ref{box:3}) and can proceed by process of elimination.

For continuous Lie groups there is a different simple trick one can do to abstractify the group. Here, one needs to compute and identify the Lie algebra associated with the group. Once the Lie algebra is found, it can be used to match the symmetry group with some familiar mathematical group. As an example, let us first consider the group $G_{\rm{uni}}$ of Eq.~\eqref{eq:uni_group}. This case is straightforward as we can use the fact that any single qubit unitary can be obtained as $U=e^{-i \phi_3 Y} e^{-i \phi_2 X}e^{-i \phi_1 Y}$ (for some set of parameters $\{\phi_1,\phi_2,\phi_3\}$) to recognize the generators of $U$ as the Pauli matrices $Y$ and $X$. Once we have identified these operators we need to find \textit{the algebra they form}. That means, we need to start computing their commutation relationships. It is straightforward to see that $[X,Y]=2i Z$, and more generally $[\sigma_i,\sigma_j]=2i \epsilon_{ijk}\sigma_k$, where $\epsilon_{ijk}$ is the Levi-Civita tensor and $\sigma_{1}=X$, $\sigma_{2}=Y$ and $\sigma_{3}=Z$. Thus, the Lie algebra $\mathfrak{g}$ associated to $G_{\rm{uni}}$ corresponds to the unitary Lie algebra $\mathfrak{su}(2)$ satisfied by the Pauli operators. Thus, we can identify $G_{\rm{uni}}$ to the unitary Lie group $SU(2)$. 

Next, let us analyze a slightly less straightforward example. Consider the QML task  from Fig.~\ref{fig:fig4} where we want to classify two-qubit ferromagnetic from antiferromagnetic states. The continuous symmetry group is now composed of a tensor product of local unitaries, i.e., 
\begin{equation}\label{eq:loc_group}
    G_{\rm{prod}}=\{U^{\otimes 2}\,|\, U\in SU(2)\}\,.
\end{equation}
Using again the fact that any local single-qubit unitary  can be expressed as $U=e^{-i \phi_3 Y} e^{-i \phi_2 X}e^{-i \phi_1 Y}$, one can verify that the unitaries $\bigotimes_{i=1}^2 V$ are obtained by exponentiation of the operators $\Sigma_1=\sum_{j=1}^2 X_j$, and $\Sigma_2=\sum_{j=1}^2 Y_j$.  Notably, here we find that their algebra is $[\Sigma_1,\Sigma_2]=2i\Sigma_3$, with $\Sigma_3=\sum_{j=1}^2 Z_j$, and more generally we see that $[\Sigma_i,\Sigma_j]=2i \epsilon_{ijk}\Sigma_k$ for any $i,j,k=1,2,3$. Hence, the generators of $U$ satisfy the same algebraic commutation relationship as the single-qubit Pauli operators. From the previous we can identify the Lie algebra as the unitary algebra $\mathfrak{su}(2)$, and the Lie group as  $SU(2)$. (Note that the previous result also allows us to identify the group $G=\{U^{\otimes k}\,|\, U\in SU(2)\}$ with $SU(2)$, for any value of $k$.)

\begin{figure}
    \centering
    \includegraphics[width=0.9\columnwidth]{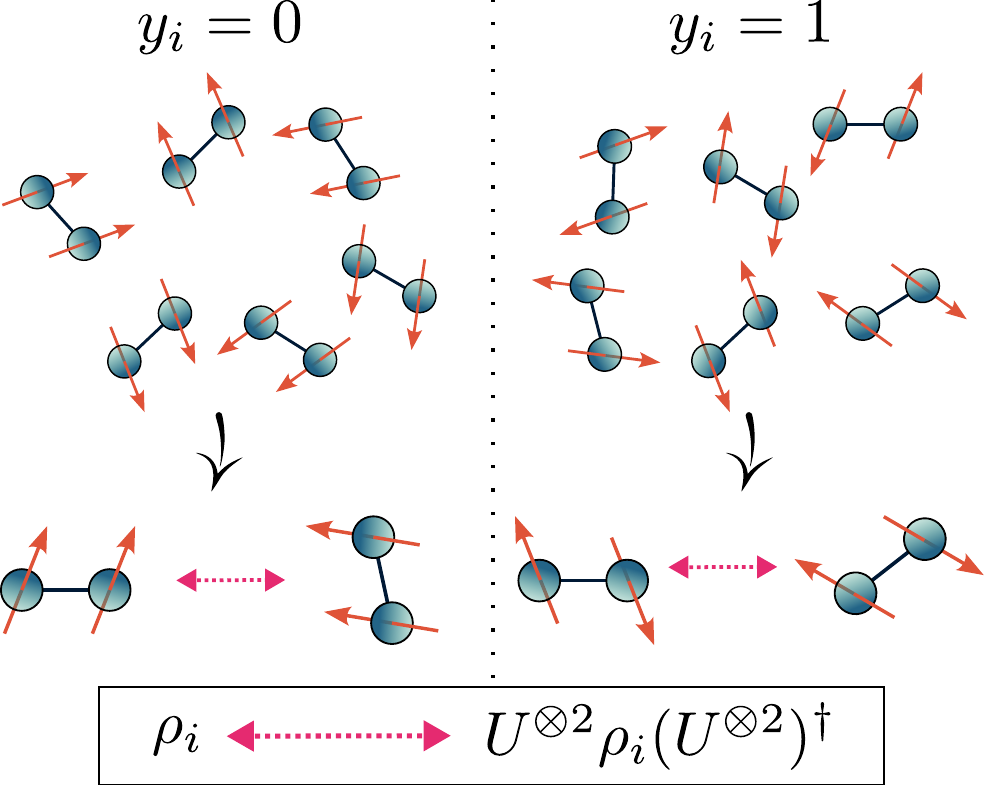}
    \caption{\textbf{Example of another continuous symmetry group.} In this QML task the goal is to distinguish ferromagnetic from antiferromagnetic two-qubit states.   The data $\rho_i$ with label $y_i=0$ is given by states where the single-qubit reduced states $\rho_A$ and $\rho_B$ are aligned (i.e., $\Tr[\rho_A\sigma]=\Tr[\rho_B\sigma]$ for any $\sigma=X,Y,Z$). The data $\rho_i$ with label $y_i=1$ is given by states where the single-qubit reduced states $\rho_A$ and $\rho_B$ are anti-aligned (i.e., $\Tr[\rho_A\sigma]=-\Tr[\rho_B\sigma]$). Here the symmetry group preserving the labels is $G_{\rm{prod}}=\{U^{\otimes 2}\,|\, U\in SU(2)\}$ of Eq.~\eqref{eq:loc_group}.  Note that here the data is quantum mechanical in nature, as it does not correspond to classical data encoded in quantum states.   }
    \label{fig:fig4}
\end{figure}

\section{Representation theory for discrete and continuous groups}\label{sec:representation theory for discrete and continuous groups}

In the previous sections we have taken a constructive approach where we started with a given dataset $\SC=\{\rho_i,y_i\}_{i=1}^N$, identified the symmetries that leave the labels invariant according to Definition~\ref{def:label-inv}, and found their abstract counterparts. This led us to prove that both the group $G_{\rm{bflip}}$ of the bit-flip symmetry transformation and the group $G_{\rm{SWAP}}$ of two qubit permutations have the abstract structure of $\mathbb{Z}_2$. Hence, at this point, we are ready to say that the  bit-flip and the qubit-swapping groups are both \textit{representations} of $\mathbb{Z}_2$. Similarly, we can say that $G_{\rm{prod}}$ from Eq.~\eqref{eq:loc_group} is a \textit{representation} of $SU(2)$.

Representation theory studies how groups can act on vector spaces through linear transformations.
In this section and those to come, we formally flesh out the mathematical basics of representation theory and Lie theory, providing a rigorous footing for the examples earlier in the text and in Box~\ref{box:1} and Box~\ref{box:2}. 
This is a standard topic with many excellent texts, so we opt instead here to instruct through examples rather than developing the theory completely.
Especially important examples will appear repeatedly throughout the text, so we have added these in color-coded boxes to aid the reader.
In writing this, we primarily reference Hall~\cite{hall2013lie}, periodically deferring to Serre~\cite{serre1977linear}, Fulton and Harris~\cite{fulton1991representation}, and Nachtergaele and Sims~\cite{nachtergaele2016quantum}. 
We warn that most quantum scientists who come in close contact with Lie theory find themselves inevitably enamored by the subject, and we encourage them to use these wonderful texts to explore further. 

In the following sections we will present some key mathematical definitions that will be used throughout this article. 

\subsection{Groups: discrete and continuous}

First, we define a group.
\begin{definition} \label{def:group definition}
A \emph{group} is a set $G$ with a binary operation $\cdot:G\times G\to G$ obeying the following axioms:
\begin{enumerate}
    \item \textit{Associativity:} For all $g,h,k\in G$, $(g\cdot h) \cdot k = g\cdot (h\cdot k)$.
    \item \textit{Identity element:} There exists an identity element $1\in G$ such that for every $g\in G$, $1\cdot g = g = g \cdot 1$.
    \item \textit{Inverse element:} For all $g\in G$, there exists a $g^{-1}$ such that $g\cdot g^{-1} = 1 = g^{-1}\cdot g$.
\end{enumerate} 
\end{definition}
We will commonly write $g\cdot h =: gh$. 
If for all $g,h\in G$ we have that $g\cdot h = h \cdot g$, we call $G$ \emph{abelian}, and if this does not hold, we call $G$ \emph{nonabelian}. 

At this point, we refer the reader to Box~\ref{box:3} places at the end of this work, where we have listed some commonly appearing discrete groups and their key properties. Here we also note that there is a stark difference in theory between finite and infinite discrete groups: most of the representation theory we will develop will only apply to finite discrete groups. That being said, Definition~\ref{def:group definition} also includes continuous groups as well: for instance  the \textit{general linear group} $GL(n, \mathbb{C})$ (resp. $GL(n, \mathbb{R})$) of invertible $n\times n$ complex (resp. real) matrices forms a group under the matrix multiplication operation. Here, the identity matrix $\idty$ is the identity element. 

With the previous, we now formally define what a \textit{continuous} matrix group is.
Recall that a \textit{subgroup} $H$ of a group $G$ is a subset of $G$ which is also a group.
\begin{definition} 
A \emph{matrix Lie group} $G$ is a closed subgroup of $GL(d,\mathbb{C})$ (or $GL(d,\mathbb{R}))$, where by closed we mean that if $A_m\in G$ is a sequence of matrices with $\lim_{m\to\infty}A_m = A\in GL(d, \mathbb{C})$, then $A\in G$.
\end{definition}

Matrix Lie groups have the key feature that they also form \textit{smooth manifolds}, or hypersurfaces.
In practice, this means that we can parameterize matrix Lie groups using coordinates, and we can take derivatives along paths in the group just as one can compute tangent vectors along paths embedded in surfaces.
The smooth manifold structure turns the potentially unwieldy problem of understanding groups with uncountably many elements into a tractable one: we get to use powerful tools from not just algebra, but analysis, geometry, and topology as well. It is worth mentioning that in greater generality, a Lie group is a group $G$ that is also a manifold, where the group multiplication and inversion are continuous operations. 
Indeed, there are Lie groups that are not matrix Lie groups, but most of the physically relevant examples like the unitary groups $U(n)$, orthogonal groups $O(n)$, or special linear groups $SL(n)$ form matrix Lie groups, so we restrict our attention to this more concrete class.\footnote{In fact, via a corollary of the Peter-Weyl theorem~\cite{folland2016course}, every compact Lie group is a matrix Lie group.}
We refer the reader to  Box~\ref{box:4} at the end of this work where we present some important examples of matrix Lie groups and their key properties.

Now, finite discrete groups are ``well behaved'' in comparison to their infinite discrete cousins, in the sense that we can say many things about their representation theory: for instance, we know their representations are completely reducible (Theorem~\ref{thm:complete reducibility of unitary representations}), can always be converted into unitary representations (Theorem~\ref{thm:Weyl's unitary trick}), and we know how many irreducible representations they have (Definition~\ref{def:reducible}). While we do not define these terms now, they will be presented below, so do not worry if these do not make sense now. The analogous condition to yield ``well behaved'' matrix Lie groups is a sort of topological proxy for finiteness: \textit{compactness}. 
Compact matrix Lie groups will have the nicest representation theory, and luckily for quantum researchers, unitary groups are compact.
Below we present a special version of the definition of compactness given by the Heine-Borel theorem.
\begin{definition}\label{def:compactness}
A matrix Lie group $G$ is called \emph{compact} if it is a closed and bounded subset of the vector space of the $d\times d$ matrices  $M_d(\mathbb{C})$ (or $M_d(\mathbb{R})$).
\end{definition}

It is useful to note that any putative Lie groups defined by equations and continuous operations (like multiplication or taking adjoints) can be readily shown to be closed: for instance, if we consider a sequence $U_m$ of unitary matrices converging to a matrix $U$, then since taking adjoints $U\mapsto U\ad$ and multiplication $U\mapsto U^\dagger \cdot U$ are continuous operations, we can take the limit of the equation $U_m^\dagger U_m = \idty$ as $m\to\infty$ to see that $U^\dagger U = \idty$, meaning the unitary group is closed. 
Boundedness can be readily checked from standard linear algebra knowledge: for instance, since the operator norm $\norm{U} = 1$, the unitary group is bounded.
Thus, the unitary group $U(d)$ (and consequently the special unitary group $SU(d)$) is compact. 
To contrast, the special linear group $SL(d; \mathbb{C})$, which consists of invertible $d\times d$ matrices of determinant 1, is not compact because it is not bounded. To see why, consider the sequence of diagonal matrices $A_n = \text{diag}(n, 1/n, 1, 1,\dots ,1)$. 
The sequence of norms $\norm{A_n}$ can be made arbitrarily large, and so $SL(d; \mathbb{C})$ is not compact.

Before moving onto representations, we need one more definition to clarify the maps which preserve the Lie structure. 

\begin{definition}
A \emph{Lie group homomorphism} between $G$ and $H$ is a smooth group homomorphism $R:G\to H$, meaning we can take directional derivatives of $R$ and we have $R(g_1 g_2) = R(g_1)R(g_2)$ for all $g_1,g_2\in G$.
\end{definition}
Note that a group homomorphism in general only consists of the condition $R(g_1 g_2) = R(g_1)R(g_2)$. 
But for finite groups, we can think of them as discrete subgroups of a larger matrix group, and from this viewpoint group homomorphisms become Lie group homomorphisms (where the smoothness condition is vacuously true).

\subsection{Representations of groups}

Now that is all well and good, but while these abstract groups govern physical symmetries, the actual incarnation of symmetry in quantum systems happens via \textit{representations} of these groups: these are the actions of a group on vector spaces through linear transformations. That is, while the abstract symmetry group  for the dataset in Fig.~\ref{fig:fig2}(a) and Fig.~\ref{fig:fig3} is $\mathbb{Z}_2$, the action of the elements of the group arises by the representations $G_{\rm{uni}}=\{U\in U(2)\}$ of Eq.~\eqref{eq:uni_group} and $G_{\SWAP}=\{\id,\SWAP\}$ of Eq.~\eqref{eq:swap-group}, acting on $\mathbb{C}^{2}$ and $(\mathbb{C}^{2})^{\otimes 2}$ respectively.

In Proposition~\ref{prop:symm-group}, we saw that the set of all unitary symmetry transformations form a group.
Indeed, in this language, given a vector space $V$, the group of unitary transformations is the subgroup of $GL(V)$ formed by the image $R(G)$ of the representation. 
Sometimes this representation yields obvious fingerprints of the abstract group $G$ (such as the previous $\mathbb{Z}_2$ example), and sometimes it is more subtle. 
In any case, there is no substitute for good examples: watch out for the colored boxes for some driving examples of representations, as they will serve as guides  through the theorems and structure results to follow. 
For the entirety of this paper, we will only consider representations on real $\mathbb{R}$ or complex $\mathbb{C}$ vector spaces $V$. That being said, we can finally define what a representation is:
\begin{definition}\label{def:representation}
A \emph{representation} of a group $G$ on a vector space $V$ is a (Lie) group homomorphism $R:G\to GL(V)$, the general linear group. The \textit{dimension} of a representation $R$ is defined to be $\dim(R) = \dim(V)$.
\end{definition} 
Commonly, we will write $R(g)=:R_g$ for notation and abbreviate ``representation'' to ``rep''.
As an unfortunate feature of the subject, the word ``representation'' can equivalently refer to the group homomorphism $R$, the vector space upon which it acts $V$, or the image subgroup $R(G)\subseteq GL(V)$.
Once one gets used to this, it is not as bad as it sounds: in practice, one often thinks of a representation as being the shared data of the vector space $V$ and the linear action of $G$ on that vector space.\footnote{Mathematicians often formalize this ``shared data'' perspective by saying a representation is a vector space $V$ with a $G$-module structure. This definition is equivalent to the one presented here, but provides alternative perspective.}
We will try to be clear as to which term we mean throughout this text. 

\begin{figure}
    \centering
    \includegraphics[width=0.8\columnwidth]{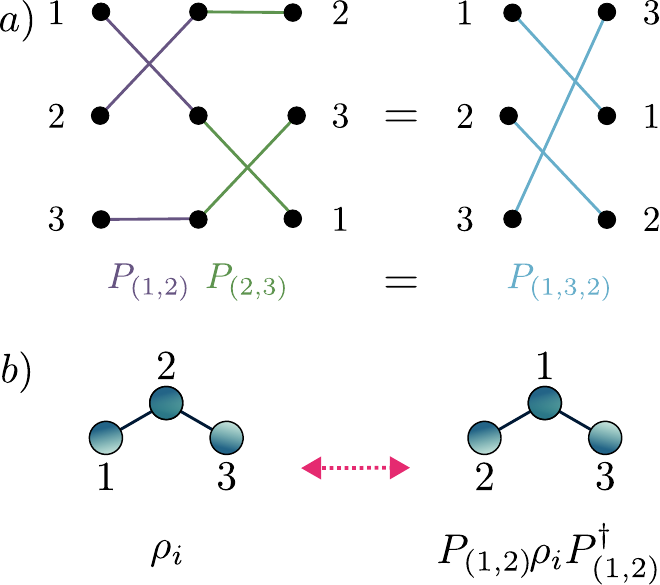}
    \caption{\textbf{The symmetric group $S_3$ and one of its representations.} a) The elements of the symmetric group $S_n$ correspond to the permutations of a set of size $n$. Here we depict the case for $S_3$ where its elements are denoted as $P_{\pi}$. For instance taking $\pi = (12)$ (using cycle notation), the element $P_{\pi}$ maps $123$ to $213$. Here we can see that a generating set for $S_3$ is $\{(1,2),(2,3),(1,3)\}$ as any permutation can be obtained via transpositions of two elements. b) Consider a QML task where the states $\rho_i$ can be thought of as representing an $n$-qubit quantum system whose interaction topology follows that of a graph~\cite{verdon2019quantumgraph,larocca2022group}. Since the way one labels the vertices and assigns them to qubits is completely arbitrary, the problem should be invariant under the action of $S_n$. By conjugating the quantum states $\rho_i$ with elements $P_\sigma\in S_n$ one obtains a new quantum state $P_\sigma \rho_i P_\sigma\ad$ whose underlying graph vertices are permuted according to $P_\sigma$. } 
    \label{fig:fig6}
\end{figure}

Once we fix an abstract group $G$, a natural question spawns from this definition: ``\textit{What do its representations look like?}''. Note that this is the inverse path that we have previously taken, where we start from a representation of a group and we go towards its abstractification. Definition~\ref{def:representation} tells us that we need a vector space $V$ and a collection of matrices $R(G)\subseteq GL(V)$ which obey the same ``group laws'' as in $G$ (this is the homomorphism piece). A given group may have many different representations, some obvious, and others less so. Let us dive into examples.

The first example of a representation is the trivial one, literally. The \textit{trivial representation} of a group $G$ on any vector space $V$, is given by $R_g = \idty\in GL(V)$ for all $g\in G$. Many times, however, representations are a bit more complicated, and they might have underlying structure to them. For instance, sometimes representations can have invariant subspaces, leading to the concepts of subrepresentations and irreducible representations. While these will be formalized below, we note that given a representation  $R:G\to GL(V)$, a subrepresentation is a subspace $W\subseteq V$ that is invariant under the action of the representation. That is,  given any $w\in W$, we have that $R_g \cdot w \in W$ for all $g\in G$. Getting ahead of ourselves, we will say that a  representation is said to be \textit{irreducible}  if it contains no smaller subrepresentations. 

Let us consider more some representations by restricting our attention to a particularly useful family of groups.

\subsubsection{Representations of the symmetric group}

\begin{purple_boxed_example}[label={box:Permutation_rep_of_S3_on_3_qubits}]{Permutation rep of $S_3$ on 3 qubits}
Consider the Hilbert space for 3 qubits $V = (\mathbb{C}^2)^{\otimes 3}$ and consider the symmetric group $S_3$, the group of all possible permutations of three elements.
Then one representation $R:S_3\to GL(V)$ is given on the computational basis by
\[
    P_\pi \cdot \ket{i_1 \; i_2 \;i_3} = \ket{i_{\pi^{-1}(1)} \; i_{\pi^{-1}(2)} \; i_{\pi^{-1}(2)}}
\] For instance, if we express $\pi = (12)$ using cycle notation, then $P_{(12)}\cdot \ket{i_1 \; i_2 \; i_3} = \ket{i_2 \; i_1\; i_3}$.
In other words, $P_{(12)} = \text{SWAP}_{1,2}$.
So, permutation representations of the symmetric group allow us to permute qubits (and more generally, tensor indices).
\end{purple_boxed_example}

The symmetric group $G = S_n$ is defined as the group whose elements are all bijections $\pi$ from a finite set to itself (see also Box~\ref{box:3}).  Moreover, as shown in Fig.~\ref{fig:fig6}(a), once we index the set with numbers $\{1,2,\dots, n\}$, $S_n$ corresponds to the permutations that can be performed on a set of size $n$. Take $S_3$. 
Now, in the search for representations, a natural vector space to probe is the Hilbert space of three qubits, since we can label a basis of this space with three indices. In Example~\ref{box:Permutation_rep_of_S3_on_3_qubits} we explicitly define a representation of $S_3$ acting on this basis by permuting indices as  $P_\pi \cdot \ket{i_1 \; i_2 \;i_3} = \ket{i_{\pi^{-1}(1)} \; i_{\pi^{-1}(2)} \; i_{\pi^{-1}(2)}}$\footnote{One may wonder why the representation acts by the inverse permutation $\pi^{-1}$ on indices instead of $\pi$. Observe that if we chose the naive definition and wrote $P_\pi \cdot \ket{i_1 i_2} = \ket{i_{\pi(1)} i_{\pi(2)}}$, then associativity no longer holds: $P_{\pi_1}\cdot (P_{\pi_2}\cdot\ket{v})\neq (P_{\pi_1}\cdot P_{\pi_2})\cdot \ket{v}$.
Something similar happens when defining representations on group functions $f:G\to \Cbb$: the ``natural'' representation $R$ of $G$ on the vector space of functions is given by precomposition: define $R_g\cdot f(x) = f(g^{-1}x)$ for all $x\in G$.}. 
The simplicity of this representation belies its depth: it will be a frequent offender, playing a central role in a variety of representations on tensor spaces. It is also useful for QML tasks dealing with quantum states defined over a graph (see Fig.~\ref{fig:fig6}(b)).

In Example~\ref{box:dihedral-rep-S3}, we present an alternative representation for $S_3$. This example serves two purposes: firstly, it illustrates that a fixed group $G$ can, and will, have representations which look nothing alike; secondly, it shows how to define a representation by working on generators of a group.
\begin{purple_boxed_example}[label={box:dihedral-rep-S3}]{Dihedral rep of $S_3$ on 1 qubit}
We will construct a representation of $S_3$ on a single qubit $V = \Cbb^2$ using some geometric insight: $S_3$ is also $D_3$, the symmetry group of the triangle. 
From this perspective, we can generate the group using a rotation $(123)$ and a reflection $(12)$, and we will define a representation which acts simultaneously on two triangles given in each copy of $\Cbb$ as the cubic roots of unity $\{1,\omega, \omega^{-1}\}$ where $\omega := e^{2\pi i /3}$. 
Define these two maps in the computational basis
\[
    R_{(123)} = \begin{pmatrix} \omega & 0 \\ 0 & \omega^{-1}
    \end{pmatrix}, \quad     R_{(12)} = \begin{pmatrix} 0 & 1 \\ 1 & 0
    \end{pmatrix}.
\] One way to think about these maps is the following: $R_{(123)}$ rotates two triangles in opposite directions by the angle $\omega$, by
\begin{align*}
    R_{(123)} R_{(123)} \ket{0} &= R_{(123)}( \omega \ket{0}) = \omega^2\ket{0} \\
     R_{(123)} R_{(123)} \ket{1} &= R_{(123)}( \omega^{-1} \ket{1}) = \omega^{-2}\ket{1} ,\\   
\end{align*} and $R_{(12)}$ is the reflection which swaps the two states (and thus the two triangles)
\[
    R_{(12)} \ket{0} = \ket{1}.
\]
\end{purple_boxed_example}
In particular, recall that \textit{a set of group elements generates} $G$, written $\langle g_1,g_2,\dots, g_n\rangle = G$, if every element $h\in G$ can be expressed as a string of generators $g_{i_1} \cdot g_{i_2} \cdot \dots \cdot g_{i_k} = h$. Of course, there are often many different generating sets for a given $G$, and particular choices are made with the intent of analysis: for example, the symmetric group $S_n$ is generated by all transpositions of two elements $(i,j)$, $1\leq i,j \leq n$. The attentive QML researcher's ears should perk up at this second statement: this means that for index permutation representations like Example~\ref{box:Permutation_rep_of_S3_on_3_qubits}, all of the representatives of $S_n$ can be expressed as products of $\text{SWAP}_{i,j}$ gates.

\subsubsection{Representations of continuous groups}

Let us now accrue a smattering of representations of continuous groups.
Continuous groups arise very naturally within QML: the unitary groups $U(d)$ are prototypical examples of Lie groups, and quantum computing is ultimately all about unitaries. In what follows we will focus and give special attention to the special unitary group $SU(2)$, and the QML task of classifying single-qubit states according to their purity (see Fig.~\ref{fig:fig2}(b)). Our choice of considering $SU(2)$ is further motivated  both for its privileged position in quantum physics as the symmetry group underlying spins and for its remarkably well behaved representation theory. 
\begin{fundamental_boxed_example}[label={box:fundamental rep of group SU(2)}]{$SU(2)$ and its fundamental rep (spin-1/2)}
Recall that the special unitary group is defined as
\[
SU(2) = \{U\in GL(\Cbb^2): \; \; U^\dagger U = \idty \text{ and } \det(U) = 1\}.
\] Take 1 qubit $V=\Cbb^2$.
The fundamental representation $U: G\to GL(V)$ is given by $U_g = g$, i.e., the matrix in $G$ is the same as its representative.
While we are here, it is good to recall that any matrix $U_g\in SU(2)$ can be expressed in terms of Pauli matrices $X,Y,Z$ by
\[
    U_g = c_0 \idty +  i(c_1 X + c_2 Y + c_3 Z) ,
\] where $(c_0,c_1,c_2,c_3)$ is a real vector with Euclidean norm $\abs{(c_0,c_1,c_2,c_3)} = 1$.
Note that since this defines an invertible and continuous map between the 3-sphere $S^3\subseteq \mathbb{R}^4$ and $SU(2)$, the two spaces are homeomorphic.
As special cases of this decomposition, we have the rotations around Cartesian axes of angle $\theta$ of the Bloch sphere:
\begin{align*}
    R_X(\theta) &= e^{-i\theta X/2} = \cos(\theta/2) \idty - i\sin(\theta/2) X \\ 
    R_Y(\theta) &= e^{-i\theta Y/2} = \cos(\theta/2) \idty - i\sin(\theta/2) Y \\
    R_Z(\theta) &= e^{-i\theta Z/2} = \cos(\theta/2) \idty - i\sin(\theta/2) Z \\
\end{align*}
\end{fundamental_boxed_example}
Again, just like with discrete groups, we always have the trivial representation of $G$ on any vector space $V$, given by $R_g = \idty\in GL(V)$ for all $g\in G$.

All matrix Lie groups also have a so-called \textit{fundamental representation} or \textit{defining representation}\footnote{There yet another unfortunate clash of terminology here--many physicists commonly use these terms interchangeably, but mathematicians would say that the definition here is strictly the defining representation, which for classical matrix Lie groups is one of the fundamental representations. Matrix Lie groups commonly have more fundamental representations than the defining rep (e.g. the defining rep $\mathbf{3}$ and its dual $\bar{\mathbf{3}}$ for $SU(3)$), but in the case of $SU(2)$, the defining representation is the only fundamental representation. } (Example~\ref{box:fundamental rep of group SU(2)}), which means that the matrices in the group and the matrices in the representation coincide (up to isomorphism, which we will learn morally means ``up to change of basis'' in Definition \ref{def:equivariance}).
Try to note this important distinction: even though the matrices between the group $G$ and its representatives $\{R_g : g\in G\}\subseteq GL(V)$ are identical, we think of the abstract group and its representatives as conceptually distinct. 
For instance, spin representations of angular momentum are, at their core, (irreducible) representations of the Lie group $SU(2)$.
In the case of spin-$1/2$, this is exactly the fundamental representation of $SU(2)$ on $V = \Cbb^2$ where both the matrices and their representatives are of the form $e^{\frac{-i}{\hbar} \vec{\theta} \cdot \vec{\sigma}}$ where $\vec{\theta} = (\theta_x,\theta_y,\theta_z)$ is a vector of real angles and $\vec{\sigma} = (X,Y,Z)$ is a vector whose entries are the Pauli matrices. But in higher spin, e.g., spin-$1$, the group $SU(2)$ is fixed but the representation $V = \Cbb^3$ is not, and so the $3\times 3$ representative ``spin matrices'' take on the new form $S_X,S_Y,S_Z$.
More light will be shed on the spin representations and angular momentum once we get to Lie algebras: a pervading theme of Lie theory is that it is often easier to work on the Lie algebra (e.g., the spin matrices $X,Y,Z$ with commutation relations) than on the Lie group (here, the rotation matrices $e^{-i \vec{\theta} \cdot \vec{\sigma}}$).

Here we note that the fact that the parameter $\theta$ is divided by a factor of two in Example~\ref{box:fundamental rep of group SU(2)} is closely related to an important topological fact: $SU(2)$ is the double cover of the group of rotations in three dimensions $SO(3)$. A key example of this ``covering'' idea is intimately familiar to the QML scientist: the map $\phi\mapsto e^{i 2\pi \phi}$ maps $\phi\in [0,1]$ and all $\phi+k$, $k\in \Zbb$ to the same rotation. In this sense, the real line is a cover of the unit circle. The ``double'' in double cover simply means that above each rotation $R\in SO(3)$, there are two corresponding rotations $\tilde{R}_1, \tilde{R}_2\in SU(2)$ (see Hall \cite{hall2013lie} for a proper definition and more detailed discussion).
Further, $SU(2)$ is simply connected, meaning every loop in $SU(2)$ can be continuously contracted to a single point without leaving $SU(2)$, while $SO(3)$ is not simply connected.
As a sneak peek down the road, this seemingly innocuous fact will play a huge role in the representation theory of these groups: while every Lie group representation gives rise to a Lie algebra representation (Theorem~\ref{thm:lie group reps yield lie alg reps}), the converse only holds locally for ``small angles'' unless the group is simply connected (Theorem~\ref{thm:lie alg reps lift to simple lie group reps}).
And since representation theory for Lie algebras is essentially linear algebraic in nature, we often wish to recover Lie group representations from Lie algebra representations, so paying attention to this simply connected condition will be crucial. 

The following pair of examples, the adjoint representation in Example \ref{box:adjoint rep of group SU(2)} and the tensor representation in Example \ref{box:tensor rep of group SU(2)}, may very well be the most important examples in this entire article. 
They are rife with representation theoretic structure, and will provide a deeper understanding of how representations and symmetries play a key role in  QML. As such, we  will be returning to these examples several times throughout our journey.
Furthermore, while they may look superficially distinct: these representations are in fact equivalent, and we will unravel the connection when we define equivalence of representations in Definition \ref{def:equivariance}.

\begin{blue_boxed_example}[label={box:adjoint rep of group SU(2)}]{$A\mapsto UAU^\dagger$ and the adjoint representation of $SU(2)$}
Let $V = M_2(\Cbb)$ denote the set of $2\times 2$ complex matrices. 
We note that the linear superoperator given by conjugation $A\mapsto U_g A U_g^\dagger$, with $U_g = g\in SU(2)$, is in $GL(V)$.
Indeed, the map $g\mapsto U_g(\cdot)U_g^\dagger$ defines a representation of $SU(2)$ known as the adjoint representation. 
Since the set $\{X,Y,Z, \idty\}\subseteq V$ forms an orthonormal basis of $V$ with respect to the Hilbert-Schmidt inner product $\langle A,B\rangle_{HS} = \Tr \frac{1}{2}[A^\dagger B]$, we can represent any $A\in V$ as 
\[
    A = c_0\idty + c_1X + c_2Y + c_3Z ,
\] where now $(c_0,c_1,c_2,c_3)\in \Cbb^4$.
Note that the trace is invariant with respect to this representation: i.e., for all $g\in SU(2)$, $\Tr[A] = \Tr[U_g A U_g^\dagger]$.
Further observe that every nonzero operator in $\text{span}\{\idty\}$ has nonzero trace, while every operator in $\text{span}\{X,Y,Z\}$ has trace zero. 
Combining these two facts, we realize that $U_g(\cdot)U_g^\dagger$ cannot map between these two linear subspaces. In other words, these subspaces are \textit{invariant subspaces} of the representation $U$ (see Definition~\ref{def:invariant subspace}).
In particular, this means that the set of representatives $\{U_g(\cdot)U_g^\dagger: g\in SU(2)\}$ can be simultaneously block-diagonalized in the basis $\{X, Y, Z, \idty\}$ of $M_2(\Cbb)$:
\begin{equation*}
    U_g(\cdot)U_g^\dagger=\begin{pmatrix} \hspace{.25mm}
    \fbox{ $\begin{array}{c c c}
    \,\,\,\,\,\,\,\, & & \\
    & & \\
    & & \\
    \end{array}$} \hspace{-.75mm}&  \begin{array}{c}
     0 \\
    0 \\
    0
    \end{array}    \\
    \begin{array}{c c c}
    0 & 0 & 0 
    \end{array} &\fbox{ $\begin{array}{c}
    \end{array}$}
    \end{pmatrix}\,,
\end{equation*}
Indeed, the invariant subspace $\text{span}\{\idty\}$ is the trivial one, which can be seen by observing that $U_g\idty U_g^\dagger = \idty$, and the invariant subspace $\text{span}\{X,Y,Z\} = \su(2)$ is the \textit{adjoint representation} of $SU(2)$ (see Definition~\ref{def:adjoint rep}).
In other words, $V = \su(2)\oplus \Cbb\idty$ as representations (see Definition~\ref{def:complete reducibility}).
\end{blue_boxed_example} 

First, let us consider the \textit{adjoint representation} of $SU(2)$ in Example~\ref{box:adjoint rep of group SU(2)}. As seen therein, the adjoint map admits a block diagonal representation in the basis $\{X, Y, Z, \idty\}$ of $M_2(\Cbb)$. This implies that the adjoint representation has two subrepresentations (e.g., invariant subspaces), one three-dimensional, and one one-dimensional. As such, we can say that the adjoint representation is not irreducible.

\begin{tensor_boxed_example}[label={box:tensor rep of group SU(2)}]{Tensor rep of $SU(2)$ on 2 qubits}
Take two qubits $V = (\Cbb^2)^{\otimes 2}$.
We can define the tensor representation $U^{\otimes 2}$ (Definition~\ref{def:tensor rep}) of $SU(2)$ on this space by taking two copies of the fundamental representation, i.e.,  $g\mapsto U_g\otimes U_g$. 
Here, one can readily see that all representatives commute with the SWAP operator $[U_g\otimes U_g, \text{SWAP}] = 0$. Moreover, it is also worth recalling that the SWAP operator acts as $\idty$ on the symmetric subspace spanned by $\{\ket{11}, \ket{01}+\ket{10}, \ket{00}\}$ and as $-\idty$ on the antisymmetric subspace spanned by $\{\ket{10}-\ket{01}\}$. 
Leveraging Proposition~\ref{prop:symmetries, Hamiltonians, and eigenstates}, this means that the representation $U^{\otimes 2}$ is block diagonalized by the symmetric-antisymmetric decomposition of $V$: i.e., in the basis $\{\ket{11}, \ket{01}+\ket{10}, \ket{00}, \ket{10}-\ket{01}\}$, every representative $U_g\otimes U_g$ can be expressed as
\begin{equation*}
    U_g\otimes U_g =\begin{pmatrix} \hspace{.25mm}
    \fbox{ $\begin{array}{c c c}
    \,\,\,\,\,\,\,\, & & \\
    & & \\
    & & \\
    \end{array}$} \hspace{-.75mm}&  \begin{array}{c}
     0 \\
    0 \\
    0
    \end{array}    \\
    \begin{array}{c c c}
    0 & 0 & 0 
    \end{array} &\fbox{ $\begin{array}{c}
    \end{array}$}
    \end{pmatrix}\,.
\end{equation*}
In other words, using the notation $\text{Sym}^2(\Cbb^2)$ for the symmetric subspace and $\text{Alt}^2(\Cbb^2)$ for the antisymmetric subspace, we have $V = \text{Sym}^2(\Cbb^2)\oplus \text{Alt}^2(\Cbb^2)$ as representations. 
While this is not yet obvious, it will be revealed in Example~\ref{box:equivalence of adjoint and tensor reps} that $\text{Sym}^2(\Cbb^2)$ is the adjoint representation $\su(2)$, and $\text{Alt}^2(\Cbb^2)$ is the trivial representation $\Cbb\idty$.
\end{tensor_boxed_example}

Next, let us consider a different representation for $SU(2)$  acting on two qubits that is obtained by taking two copies of the fundamental representation, i.e.,  $g\mapsto U_g\otimes U_g$. As shown in Example~\ref{box:tensor rep of group SU(2)}, the tensor representation of $SU(2)$ also takes a block diagonal form  $\text{Sym}^2(\Cbb^2)\oplus \text{Alt}^2(\Cbb^2)$ with $\text{Sym}^2(\Cbb^2)$ the symmetric subspace and $\text{Alt}^2(\Cbb^2)$ the antisymmetric subspace. Here it is fundamental to remark that while the block structure of the adjoint representation and the tensor representations \textit{look} the same, they are expressed in completely different bases. The adjoint representation is diagonalized in the Pauli basis $\{X,Y,Z,\idty\}$, while the tensor one is diagonalized in the symmetric and antisymmetric basis $\{\ket{11}, \ket{01}+\ket{10}, \ket{00}, \ket{10}-\ket{01}\}$. Notably, the similarity between these two representation predicts a fact that we will prove later: these two representations are equivalent.

At this point we find it instructive to revisit the task of classifying single-qubit states according to their purity. As we will see, the results we  obtained from representation theory can lead to some extremely powerful insights into how symmetries play a key role in QML (for a more formal treatment see~\cite{larocca2022group}).  First, we recall that, as shown in Fig.~\ref{fig:fig2}(b), the goal is to classify single-qubit pure states from single-qubit mixed states. Here we know that the labels are  invariant under the action of any unitary, meaning that $G=\{U\in SU(2)\}$ (we can then trivially generalize for $G=\{U\in U(2)\}$). 

As mentioned previously, we want to build QML models respecting the symmetry of the dataset. First, let us consider a QML model that is a special case of Eq.~\eqref{eq:QML-model} where we only act on a single copy of each state in the dataset ($k=1$), and where the parametrized channel is the identity (here we seek to find an optimal measurement operator). These type of experiments are commonly known as conventional experiments (ones with no quantum memory)~\cite{huang2021quantum,larocca2022group} (see Fig.~\ref{fig:fig7}(a)). Thus, we have a QML model of the form \begin{equation}\label{eq:conv}
    h(\rho_i)=\Tr[\rho_i M]\,,
\end{equation}
where $M$ is a single-qubit Hermitian measurement operator. Then, the principles of GQML indicate that if $h$ is to respect the symmetries of the dataset, one needs to  pick an invariant measurement operator (see Eq.~\eqref{eq:inv}), i.e., an operator $M$ such that  
\begin{equation}\label{eq:inv2}
    [M,U_g]= 0\,,  \quad \forall U_g\in SU(2)\,.
\end{equation}
Since we know that we are working with the fundamental representation of $SU(2)$, we can use Example~\ref{box:fundamental rep of group SU(2)}, and more specifically the fact that a unitary can be expressed as  $ U_g = c_0 \idty +  i(c_1 X + c_2 Y + c_3 Z)$ with $\abs{(c_0,c_1,c_2,c_3)} = 1$, to see that the \textit{only} operators that commutes with any unitary in $SU(2)$ are those proportional to the identity. Combining this result with Eqs.~\eqref{eq:conv} and ~\eqref{eq:inv2} shows that the only GQML model that is invariant under $SU(2)$ is $h(\rho_i)=\Tr[\rho_i  (c\idty)]=c\Tr[\rho_i]=c$ which is clearly incapable of classifying the data. The previous can be further understood from the fact that the fundamental representation of $SU(2)$ admits no (non-trivial) block diagonal representation, i.e., the fundamental representation of $SU(2)$ is \textit{irreducible}. As shown in Fig.~\ref{fig:fig7}(a), this absence of a block diagonal structure implies that the only possible measurement operator is the identity. This idea will return at a higher level when we learn Schur's lemma (Theorem~\ref{lemma:Schur's lemma}), which states that the only maps from an irrep to itself that commute with every unitary are scalar multiples of the identity $c\idty$.

\begin{figure}
    \centering
    \includegraphics[width=1\columnwidth]{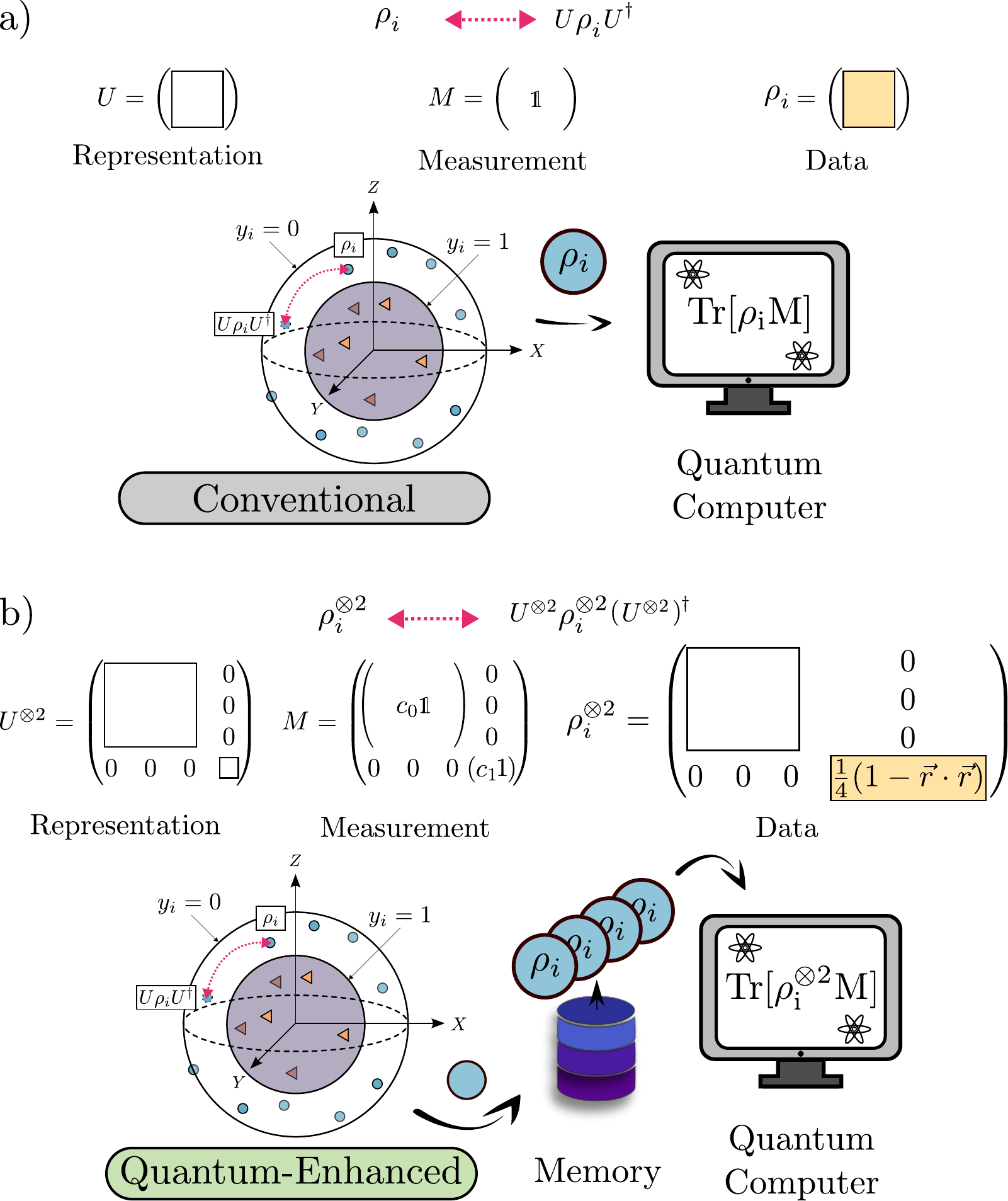}
    \caption{\textbf{Fundamental and tensor representation of $SU(2)$ and the task of classifying single-qubit states according to their purity.} a) In a conventional experiment we want to classify the data by computing $h(\rho_i)=\Tr[\rho_i M]$. Here, we are working with the fundamental representation of  $G=SU(2)$, which admits no non-trivial block diagonal structure (is irreducible). Using Eq.~\eqref{eq:inv2} we find that the only possible measurement operator is $M\propto\idty$, in which case the model  cannot classify the data.  b) In a quantum-enhanced experiment we want to classify the data by computing $h(\rho_i)=\Tr[\rho_i^{\otimes 2} M]$. Here, we are working with the tensor representation of $G=SU(2)$, which admits a block diagonal structure (is reducible) in the symmetric-antisymmetric basis (see Example~\ref{box:tensor rep of group SU(2)}). Now the operators $M$ satisfying  Eq.~\eqref{eq:inv22} are those that take  the form  $M=c_0\idty_3\bigoplus c_1 1$. This will be revisited in greater generality when we reach the commutant structure theorem \ref{thm:commutant structure} and Schur-Weyl Duality \ref{sec:Schur-Weyl Duality}.  Notably, by writing $\rho_i=\frac{1}{2}(\idty+\vec{r}\cdot\vec{\sigma})$, we can see that in the same  symmetric-antisymmetric basis  $\rho_i^{\otimes 2}$ has a component in the antisymmetric space given by $\frac{1}{4}(1-\vec{r}\cdot\vec{r})$.  }
    \label{fig:fig7}
\end{figure}

Next, let us consider a QML model as in Eq.~\eqref{eq:QML-model} where we are allowed to act on two copies of the data states. These are known as  quantum-enhanced experiments (one with quantum memories)~\cite{huang2021quantum,larocca2022group} (see Fig.~\ref{fig:fig7}(b)). That is,
\begin{equation}\label{eq:conv2}
    h(\rho_i)=\Tr[\rho_i^{\otimes 2} M]\,,
\end{equation}
where $M$ is a Hermitian operator on two qubits. Now, the representation of the symmetry group is given by $U_g\otimes U_g$, and thus  Eq.~\eqref{eq:inv} indicates that we need to pick a measurement operator such that 
\begin{equation}\label{eq:inv22}
    [M,U_g^{\otimes 2}]= 0\,,  \quad \forall U_g\in SU(2)\,.
\end{equation}
Take a look at the block diagonal structure in the symmetric-antisymmetric basis of $U_g^{\otimes 2}$ in Example~\ref{box:tensor rep of group SU(2)}. We will later learn in Example \ref{box:reducibility of tensor rep of SU(2)} that this representation cannot be further block diagonalized and the symmetric and antisymmetric spaces are irreducible---for now, take this for granted. Schur's lemma will then imply that $M$ must be of the form $M=c_0\idty_3\bigoplus c_1 1$, where $\idty_3$ denotes the $3\times 3$ identity. This can be better visualized from the block structure of $U^{\otimes 2}$ and of $M$ in Fig.~\ref{fig:fig7}(b). For instance, the choice $c_0=c_1=1$ leads to $M=\idty$, whereas $c_0=1$ and $c_1=-1$ leads to $M=\SWAP$, both of which an be trivially verified to commute with $U_g^{\otimes 2}$.
In particular, writing $\rho_i=\frac{1}{2}(\idty+\vec{r}\cdot\vec{\sigma})$, and expanding $\rho_i^{\otimes 2}$ in the symmetric-antisymmetric basis shows that the component in the antisymmetric space is $\frac{1}{4}(1-\vec{r}\cdot\vec{r})$. Hence, we can readily see that the special choice $c_0=0$ and $c_1=1$ leads to the outcome of the GQML model in Eq.~\eqref{eq:conv} to be $h(\rho_i)=\frac{1}{4}(1-\vec{r}\cdot\vec{r})$. Clearly, this GQML model can indeed classify the data in the purity dataset.

The previous example is extremely rich and has quite a few points worth highlighting. First, we note that a QML model acting on $k$-copies of $\rho_i$ induces a $k$-th order tensor representation of the symmetry group. For instance in the purity dataset where $G=SU(2)$, a QML model with a single copy of $\rho_i$ leads to the fundamental representation which has no non-trivial block diagonal structure (the  representation is irreducible), while a model working with $k=2$ copies leads to the tensor representation which admits a non-trivial block diagonal structure (the representation is \textit{reducible}). \textit{Understanding the block diagonal structure of the representation is crucial as its presence (or absence)  imposes restrictions on the  operators $M$ we can use in a GQML model} (see Fig.~\ref{fig:fig7}, were a block diagonal $U_g^{\otimes 2}$ implies a block diagonal $M$ with identities of different sizes). Finally, since we  compute the expectation value of $M$ over $\rho_i$, \textit{the block diagonal structure of $M$ dictates what is the information that we can access from the state}. For a conventional experiment, we saw that the only accessible information was the norm of $\rho_i$, which is not useful to solve the classification task. On the other hand, in a quantum-enhanced experiment we found that  $\rho_i^{\otimes 2}$ also admits a block diagonal structure, and that $M$ can precisely access the block that encodes the information about the purity of the quantum state. This will be a running theme, the \textit{ostinato} if you will, for the rest of this work and in GQML: \textit{symmetries are all about block diagonalization.}

\subsection{There and back again: Lie groups and algebras} \label{sec:there and back again: lie groups and algebras}
In this section, we unravel the mathematical correspondences between a Lie group $G$ and its Lie algebra $\liea$.
While many symmetries arise as Lie group symmetries, Lie groups can be somewhat unwieldy objects.
Manifolds are complicated, and their often highly nonlinear structure requires sophisticated analysis to probe.
But Lie algebras, which describe ``derivatives of paths in the Lie group'', are ultimately vector spaces and thus vulnerable to the mighty hammers of linear algebra.
A natural program arises from this observation:
\begin{center}
    \textit{When we have a problem with Lie group symmetry, pass to the Lie algebra, analyze it, and return to the Lie group.}
\end{center}
(This idea will be key in our examples below, so keep it in mind!)
The key to passing between the two is the \textit{exponential map}, which is a local diffeomorphism (a smooth map with smooth inverse) between $\liea$ and $G$ nearby $\idty\in G$ (Theorem \ref{thm:exp is a local diffeo}).
But again, symmetries are ultimately representations of the group $G$, not the group itself.
To pass between representations of $G$ and of $\liea$, we have two major theorems: Theorem \ref{thm:lie group reps yield lie alg reps} guarantees that every Lie group representation induces a Lie algebra representation, and Theorem \ref{thm:lie alg reps lift to simple lie group reps} provides a partial converse wherein Lie algebra representations locally lift to Lie group representations.
Our three crucial representations of $SU(2)$---the fundamental representation Example~\ref{box:fundamental rep of group SU(2)}, the adjoint representation Example~\ref{box:adjoint rep of group SU(2)}, and the tensor representation Example~\ref{box:tensor rep of group SU(2)}---will continue serving us as we work through the theory and should paint a vivid picture of this correspondence.

Formally, a \textit{Lie algebra} is a vector space $\liea$ over a field $\mathbb{F} \in \{\Cbb,\Rbb\}$ (for us usually over $\Cbb$) with a \textit{Lie bracket} $[\cdot,\cdot]: \liea\times \liea \to \liea$, which satisfies the following axioms holding for all $X_1,X_2,X_2\in \liea$ and $a,b\in \mathbb{F}$,
\begin{enumerate}
    \item \textit{Antisymmetry:} $[X_2,X_2] = -[X_2,X_1]$.
    \item \textit{Bilinearity:} $[aX_1+bX_2,X_3] = a[X_1,X_3] + b[X_2,X_3]$.
    \item \textit{Jacobi Identity:} $[[X_1,X_2],X_3] + [[X_2,X_3],X_1] + [[X_2,X_1],X_2] = 0$.
\end{enumerate} Note that the second condition is called bilinearity because linearity in one component implies linearity in the other by antisymmetry.
Crucially, Lie algebras are not in general associative, meaning we often have $[[X_1,X_2],X_3] \neq [X_1,[X_2,X_2]]$.
The Jacobi Identity acts as a sort of weaker form of associativity in its stead.
We say two elements $X_1,X_2\in \liea$ \textit{commute} if $[X_1,X_2]=0$. 
If every pair of elements in $\liea$ commutes, then we call $\liea$ a \textit{commutative} Lie algebra.
Note that any associative algebra $\mathcal{A}$, like matrices $M_n(\Cbb)$ or other $C^*$ algebras, give rise to a natural Lie algebra when we take $[\cdot,\cdot]$ to be the commutator.
Physicists often leverage this and describe Lie algebras by simply listing generators and commutation relations (via structure constants): for instance, the Lie algebra $\mathfrak{so}(3)$ can be expressed by writing angular momentum operators $L_X,L_Y,L_Z$ and their commutation relations $[L_\ell,L_m] = i\hbar \sum_{n} \varepsilon_{\ell m n} L_n$.

If $G$ is a matrix Lie group\footnote{More generally, when $G$ is a Lie group, the Lie algebra is $\liea = T_\idty G$ with Lie bracket given by the vector field commutator of the vectors $X,Y$ pushed forward by left-multiplication: $[X,Y] = [LX,LY]_\idty$ for all $X,Y\in \liea$.}, then the Lie algebra corresponding to $G$ is $\liea = T_\idty(G)$, the tangent space to the identity $\idty\in G$, together with the matrix commutator (see Fig.~\ref{fig:fig5}). 
Equivalently, $\liea$ is the set of all matrices $X$ such that $e^{tX}$ is in $G$ for all real numbers $t$.
Paths $e^{tX}$ are called \textit{one-parameter subgroups}.
This means that we can recover any $X\in \liea$ by starting with a one-parameter subgroup $e^{tX}\subseteq G$ and taking derivatives at the identity $\idty = e^0$: 
\begin{equation}
    X = \frac{d}{dt} (e^{tX})\Big\vert_{t=0} .
\end{equation} This is what we mean when we say the Lie algebra is ``the vector space of directional derivatives at $\idty\in G$'', which is just a rephrasing of $\liea = T_\idty(G)$ (see Examples \ref{box:fundamental rep of group SU(2)} and \ref{box:fundamental rep of alg su(2)} to see this in action).

\begin{figure} 
    \centering
    \includegraphics[width=0.7\columnwidth]{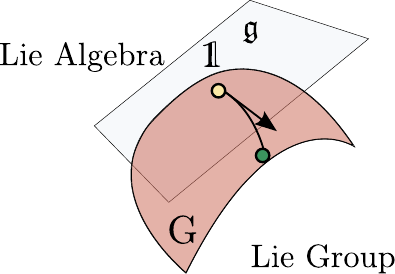}
    \caption{\textbf{Lie group and Lie algebra.}
    The Lie algebra $\mathfrak{g}$ is defined as the tangent plane at identity of the manifold (Lie group) $G$. 
    So, paths in $G$ can be differentiated at $\idty$ to give directional derivatives in $\mathfrak{g}$, and paths in $\liea$ can be exponentiated to give paths in $G$, just as in solutions to linear ordinary differential equations.
    For instance if for an element of the algebra $X\in G$ we define a path $g=e^{\theta X}\in G$, then $\frac{d}{d\theta} g(\theta) \vert_{\theta = 0} = X e^{\theta X} \vert_{\theta = 0} =X$.
    }
    \label{fig:fig5}
\end{figure}

Let us present a different way of working with tangent space using a computational trick: any Lie group element $g\in G$ nearby $\idty$ can be Taylor expanded $g = \idty + \varepsilon X + O(\varepsilon^2)$ where $X\in \liea$. 
Then matrix equations defining the Lie group $G$ induce equations defining the Lie algebra $\liea$: for instance, if $G$ is a unitary group, $g^\dagger g = \idty$ and so
\begin{equation}\label{eqn:computational trick for Lie groups to Lie algebras}
    (\idty + \varepsilon X + \OC(\varepsilon^2))^\dagger(\idty + \varepsilon X + \OC(\varepsilon^2)) = \idty.
\end{equation} Equating $\varepsilon$ terms, we see that $X^\dagger = -X$, meaning the Lie algebra $\mathfrak{u}(d)$ of a unitary group $U(d)$ consists of skew-Hermitian matrices with the standard matrix commutator $[\cdot,\cdot]$.
The same method can be applied to find that when $G = O(d; \R)$, then $\liea = \{X \in M_d(\R) : X^T + X = 0\}$, or when $G = GL(d; \Cbb)$, then $\liea = \mathfrak{gl}(d;\Cbb) := M_d(\Cbb)$.
Including the condition ``special'' means that $\det(g) = 1$, and since $\det(e^{tX}) = e^{\Tr[tX]}$, we have the condition $\Tr[X] = 0$ on the Lie algebra.
For instance, one finds $\su(d) = \{X\in \mathfrak{u}(d): \Tr[X] = 0\}$.

The tangent space definition, depicted pictorially in Fig.~\ref{fig:fig5} tells us quite a bit geometrically.
For one, it immediately tells us that $G$ and $\liea$ have the same dimension as manifolds (note that any finite dimensional vector space is automatically a manifold, where its dimension as a manifold is its dimension as a real vector space).
It also tells us that we can think of the Lie algebra $\liea$ as the vector space of directional derivatives arising from paths in the Lie group $G$.
Crucially, this means that when we have symmetries arising as representations of $G$, we can take derivatives and work on representations of the Lie algebra. 
In practice it is often easier to work on the Lie algebra: \textit{while $G$ is generically not a vector space, $\liea$ is, and so the tools of linear algebra are at our disposal}.
We need then to understand how to properly move between $G$ and $\liea$ and their representations.
The series of following definitions, propositions, and theorems provides a bare-bones toolkit for this program which we will later use to construct equivariant maps as in Eq~\eqref{eq:equiv}. 
Again, the beating heart of our treatment here is through the boxed examples and the figures: follow along with those to get a feel for how the toolkit works.

The exponential map is central to moving between $G$ and $\liea$. 
Recall a \textit{diffeomorphism} is a smooth map with smooth inverse.
The ``local'' condition here is crucial: either injectivity or surjectivity may fail globally.
\begin{theorem}[Exponential map is a local diffeomorphism]\label{thm:exp is a local diffeo}
The exponential map $e^{(\cdot)}: \liea \to G$ is locally a diffeomorphism nearby $0\in \liea$, meaning there exists open regions $U\subseteq \liea$ with $0\in U$ and $V\subseteq G$ with $\idty \in V$ such that $e^{(\cdot)}:U\to V$ is a diffeomorphism.
\end{theorem}
Notice that by continuity, since the vector space $\liea$ is connected, the exponential map can only map into the connected component of $G$ containing $\idty$. 
This immediately knocks out the possibility of $\exp(\cdot)$ being onto $G$ when $G$ is disconnected.
A key example of this type are the orthogonal groups $O(d)$, which have 2 disconnected pieces (which can be seen by realizing that $\det:O(d)\to \{-1,1\}$ is a continuous map and so the preimages of $-1$ and $1$ must be disconnected).
But even when $G$ is connected, we are not guaranteed that the exponential map is surjective.
So---when is it surjective?
In general, that is a nuanced question, but the proposition below gives one useful instance where this is so.
\begin{proposition}\label{prop:exp is surjective when G compact}
If $G$ is a compact connected Lie group, then the exponential map is surjective.
\end{proposition}

In practice however, this is not too terrible a constraint: when $G$ is connected, we can still write Lie group elements as products of exponentials of Lie algebra elements, thanks to the next proposition.
And if $G$ consists of finitely many disconnected components, one can easily work on the connected component and then move to other components by acting via a finite group. 
The prototype here is $O(d)$, which consists of $SO(d)$ and $O(d)-SO(d)$. 
We can use Lie theory to work on $SO(d)$, then extend our analysis to $O(d)-SO(d)$ by multiplying by any map of determinant $-1$.

Let us now consider the following proposition:
\begin{proposition}\label{prop:lie group elements as products of exps of lie algebra elements}
Let $G$ be a connected Lie group and let $g\in G$.
Then we can find $X_1,X_2,\dots,X_m\in \liea$ such that
\[
    g = e^{X_1}e^{X_2}\dots e^{X_m}.
\]
\end{proposition} Note that while in Hall \cite{hall2013lie} the result is only stated for matrix Lie groups, a large result in Phillips \cite{phillips1994how} (which Phillips attributes to Djokovi\'c) in fact ensures this holds for any connected Lie group, and $m$ can be chosen to be $m\leq 3$. An explicit example of this appears when one uses Euler angles to decompose any single qubit unitary into three rotations $U = e^{-i\phi_1 X}e^{-i\phi_2 Y} e^{-i\phi_3 Z}$.

Now, for the quantum practioner, we need to know how the correspondence furnished by the exponential map passes through to representations---remember, physical symmetries come from representations of groups, not the groups themselves.
First, a definition.
Notice the similarity to Lie group representation: Lie group representations are homomorphisms into $GL(V)$, the invertible matrices from $V\to V$, while Lie algebra representations are homomorphisms into $\mathfrak{gl}(V)$, the matrices from $V\to V$ with the commutator $[\cdot,\cdot]$ as the Lie bracket. 
\begin{definition}\label{def:lie representation}
Let $\liea$ be a Lie algebra and $V$ be a finite dimensional vector space. 
A \emph{representation} $r$ of $\liea$ acting on $V$ is a map $r:\liea \to \mathfrak{gl}(V)$ that is a Lie algebra homomorphism, a linear map satisfying 
\[
    r\paran{[X,Y]} = [r(X),r(Y)], \qquad \text{for all }X,Y\in \liea .
\] The dimension of the representation $r$ is defined by $\dim(r)=\dim(V)$.
\end{definition} Every example in this section consists of Lie algebra representations. 
All three of these arise from their corresponding Lie group representations: the key, perhaps unsurprisingly, is the exponential map combined with the tangent space definition of a Lie algebra.

We now present a series of theorems which capture the relationship between Lie group and Lie algebra representations.
First, every matrix Lie group representation gives rise to a Lie algebra representation.
\begin{theorem}[Lie group reps induce Lie algebra reps]\label{thm:lie group reps yield lie alg reps}
Let $G$ be a matrix Lie group with Lie algebra $\liea$.
If $R$ is a representation of $G$ on $V$, then there exists a unique representation $r$ of $\liea$ on $V$ given by
\[
    r(X) = \frac{d}{dt}\paran{R(e^{tX})}\Big\vert_{t=0}, \qquad \text{for all } X\in \liea.
\] We call $r$ the representation of $\liea$ induced by $R$.
\end{theorem}

This next pair of theorems form a partial converse to the previous theorem.\footnote{Lie's third theorem assures us that every finite dimensional Lie algebra is the Lie algebra of a Lie group (not necessarily a matrix Lie group, but often this is the case), so this handles most cases.} 
\begin{theorem}[Lie algebra reps lift to simple Lie group representations]\label{thm:lie alg reps lift to simple lie group reps}
Let $G$ be a simply connected matrix Lie group, and let $r$ be a representation of the corresponding Lie algebra on $V$. 
Then there is a unique representation $R$ of $G$ with the property
\[
    R(e^X) = e^{r(X)} \quad \text{for all } X\in \liea.
\]
\end{theorem} 
\begin{corollary}[Lie algebra reps locally lift to Lie group reps]\label{cor:lie alg reps locally induce Lie group reps}
Let $G$ be a matrix Lie group, and let $r$ be a representation of the corresponding Lie algebra on $V$. 
Then using Theorem \ref{thm:exp is a local diffeo}, we can always locally define a representation $R$ on $G$ by the mapping 
\[
    R(g) = e^{r(X)} \text{ defined for all } g=e^X \text{ nearby } \idty.
\] Here, by ``nearby'' we mean ``wherever the exponential map is a diffeomorphism''. 
Indeed, in this region, all $g$ can be written as $g=e^X$.
\end{corollary}

Put together, Theorems \ref{thm:lie group reps yield lie alg reps} and \ref{thm:lie alg reps lift to simple lie group reps} mean that if $G$ is simply connected, there is a one-to-one correspondence between their representations. 
If we relax the simply connected assumption, the power of this theorem weakens, but not too terribly much: the corollary guarantees the existence of a locally defined representation.

An important example of this distinction between the cases for simply connected and non simply connected groups arises naturally in physics: the group of 3D rotations $SO(3)$ is not simply connected, but the spin group $SU(2)$ is simply connected.
They have the same Lie algebra $\mathfrak{so}(3)\cong \su(2)$, and we can take as a basis the orbital angular momentum operators $L_x, L_y, L_z$ (of course, the Pauli operators $X,Y,Z$ also form a basis for this Lie algebra).
It can be shown that this Lie algebra has, up to isomorphism, exactly one irreducible representation each (see Definition \ref{def:reducible}) on $\Cbb^2$ (spin-1/2) and on $\Cbb^3$ (spin-1).\footnote{More generally, it has exactly one irreducible representation on each $\Cbb^d$ where $d$ is any natural number, and they are called spin $s$ $d=2s+1$ representations.} 
Theorem \ref{thm:lie alg reps lift to simple lie group reps} guarantees that both of these representations of $\su(2)$ lift to representations on $SU(2)$, called spin representations. 
But only the spin-1 representation lifts to a representation of $SO(3)$---the spin-1/2 representation does not yield a representation of $SO(3)$. 
This justifies the previously mysterious (at least to us) undergraduate physics statement: ``orbital angular momentum can only have integer quantum numbers, but spin angular momentum can have half-integer quantum numbers'', where quantum numbers really just mean labels of irreducible representations.\footnote{For $SU(2)$, dimension is a sufficient label, in the sense that each dimension has a unique irreducible representation. 
This is not the case for other Lie groups, or even $SU(d)$ for $d>2$.}

Now, returning to our focus: many interesting symmetry groups in quantum machine learning are not simply connected! 
How do we deal with these?
There are two general approaches:
\begin{enumerate}
    \item Use Corollary \ref{cor:lie alg reps locally induce Lie group reps} to work nearby $\idty$.
    For instance, if the only physical symmetries that are realized in a model are perturbations of $\idty$, as may occur in some noisy classification tasks, we do not need the full representation and a local one will suffice.
    Indeed, this is a common situation, especially in physically motivated tasks.
    \item There is a topological construction called the ``universal cover'' $\tilde{G}$ of a Lie group $G$. 
    The universal cover is simply connected by definition, and the quotient group $\tilde{G}/G$ is always a discrete group.
    In practice, this group is often finite and small: in the spin example, we have $\tilde{G}/G = SU(2)/SO(3)\cong \Zbb_2$. 
    The good news about the universal cover is that every representation of $G$ comes from a representation of $\tilde{G}$.
    This means that the only thing that could go wrong is that a Lie algebra representation of $\liea$ may not lift to $G$ due to some obstruction in the discrete quotient group. 
    This just means checking some condition on the quotient group, which for groups like $SO(d)$, only means checking a condition on $\mathbb{Z}_2$---this is exactly what is done to differentiate half-integer spin representations like spin-1/2 from integer spin representations like spin-1.
\end{enumerate}

\begin{fundamental_boxed_example}[label={box:fundamental rep of alg su(2)}]{$\su(2)$ and its fundamental rep (spin-1/2)}
To move from a representation of a Lie group $G$ to a representation of the Lie algebra $\liea$, we take derivatives of paths and evaluate at the identity. 
Let us demonstrate this with the fundamental representation on $V=\Cbb^2$ from earlier in Example \ref{box:fundamental rep of group SU(2)}.
We know that since $SU(2)$ is homeomorphic to the 3-sphere, it is a 3 dimensional real manifold, and so its Lie algebra $\su(2)$ is a 3 dimensional real vector space. 
It thus suffices to find 3 linearly independent tangent vectors, which we can do by taking derivatives of parameterized paths and evaluating at the identity $\idty$.
Since we have explicit rotation paths which are $\idty$ when $\theta = 0$, let us use those:
\begin{align*}
    \frac{d}{d\theta}R_X(\theta)\Big|_{\theta = 0} &= \frac{d}{d\theta} e^{-i\theta X/2} \Big|_{\theta = 0} = \frac{-i}{2} X \\ 
    \frac{d}{d\theta}R_Y(\theta)\Big|_{\theta = 0} &= \frac{d}{d\theta} e^{-i\theta Y/2}\Big|_{\theta = 0} = \frac{-i}{2} Y\\
    \frac{d}{d\theta}R_Z(\theta)\Big|_{\theta = 0} &= \frac{d}{d\theta} e^{-i\theta Z/2}\Big|_{\theta = 0} = \frac{-i}{2} Z\,.
\end{align*} But we of course know Paulis are linearly independent! 
Thus, up to a pesky factor of $-i$, the Pauli spin matrices form a (real) basis for the Lie algebra $\su(2)$, and we now know that the fundamental representation $r: \su(2)\to \mathfrak{gl}(V)$ is given simply by $r_X = \frac{-i}{2}X$, $r_Y = \frac{-i}{2}Y$, and $r_Z = \frac{-i}{2}Z$ (we can drop the factors of $-i/2$ because the representation $V$ is complex and so we can rescale).
So just like the case for $U:G\to GL(V)$, the matrix in $\liea$ is the same as its representative.

Thanks to Theorem \ref{thm:lie alg reps lift to simple lie group reps}, we now realize that this representation on the Lie algebra basis $X,Y,Z$ lifts to the fundamental representation of $SU(2)$ on $\Cbb^2$. This justifies calling both the $SU(2)$ and $\su(2)$ representations here the ``spin-1/2'' representation.
\end{fundamental_boxed_example}

\begin{blue_boxed_example}[label={box:adjoint rep of alg su(2)}]{$A\mapsto {[X, A ]}$ and the adjoint rep of $\su(2)$}
Return to the $G = SU(2)$ representation $V = M_2(\Cbb)$ given by $A\mapsto U_g A U_g^\dagger$, which decomposes as $V = \su(2) \oplus \Cbb\idty$.
Let us again use differentiation to see how this passes to the Lie algebra $\liea = \su(2)$: any element in $SU(2)$ nearby $\idty$ can be reached by a smooth parameterized path $U_g(\theta) = e^{i \theta W}$ with $W\in \su(2)$, noticing that $U_g(0) = \idty$.
Then 
\begin{align*}
\frac{d}{d\theta}\paran{U_g(\theta) A U_g^\dagger(\theta)} \Big|_{\theta = 0} &= \frac{d}{d\theta}\paran{e^{i\theta H} A e^{-i\theta H}} \Big|_{\theta = 0} \\
&= i\Big(He^{i\theta H}A e^{-i\theta H} \\
&\qquad - e^{i\theta H}AH e^{-i\theta H} \Big)\Big|_{\theta = 0} \\
&= i[H,A].
\end{align*} So, the Lie group representation $g\mapsto U_g(\cdot)U_g^\dagger \in GL(V)$ induces the Lie algebra representation $W\mapsto [H,\cdot]\in \mathfrak{gl}(V)$.
Indeed, in agreement with Theorem \ref{thm:lie group reps yield lie alg reps}, this representation similarly respects the decomposition of $V$ into invariant subspaces: observe that for all $H\in \su(2)$, $[H,\idty] = 0$ trivially, and $[H,\su(2)]\subseteq \su(2)$ because $\su(2)$ is a Lie algebra and thus closed under commutators.
So $V = \su(2)\oplus \Cbb\idty$ as Lie algebra representations. 
\end{blue_boxed_example}

Now, let us see this collection of theorems in action to see how to concretely pass from representations of the Lie group $G$ to the Lie algebra $\liea$ in Examples~\ref{box:fundamental rep of alg su(2)}, ~\ref{box:adjoint rep of alg su(2)}, and ~\ref{box:tensor rep of alg su(2)}. Note that there is an unfortunate distinction between physicists' and mathematicians' conventions: physicists generally write the exponential map as $e^{i \theta W}$, while mathematicians write $e^{\theta W}$.
This means that we need to be slightly careful with our statements about Lie algebra generators: physicists will say the Paulis, which are Hermitian, generate $\su(2)$, whereas the equivalent statement for mathematicians is that $\su(2) = \text{span}_\R(iX,iY,iZ)$.
There are two comments to be made here:
\begin{itemize}
    \item When we work on the Lie groups and Lie algebras themselves, the distinction between these conventions is crucial.
    We recommend using the ``Taylor expansion trick'' (Equation \ref{eqn:computational trick for Lie groups to Lie algebras}) from the beginning of this section to avoid getting confused.
    \item In practice for quantum machine learning, we are almost always working on \textit{complex representations}: since we are allowing complex scalars on the representation side, the representations of a Lie algebra $\liea$ are the same as for the complexification of that Lie algebra $\su(2)\otimes \Cbb$, which just means we allow complex linear combinations of our basis. 
    Here, the distinctions between the physicist's convention and the mathematician's is no longer problematic.
    Do note that this comes at a price if we need other information beyond complex representations of the Lie algebra. For instance, the Lie groups $SU(2)$ and $SL(2)$ are very different (e.g. the former is compact while the latter is not); but their complexified Lie algebras are identical $\su(2)\otimes \Cbb = \mathfrak{sl}_2\otimes \Cbb$. 
\end{itemize}

All the previous algebra theory to connect Lie groups to Lie algebras is nice and good. But how can we use this for QML? We have already stated that dealing with Lie groups can be cumbersome, whereas dealing with Lie algebras is more manageable. In particular we recall the mantra: \textit{When we have a problem with Lie group symmetry, pass to the Lie algebra, analyze it, and return to the Lie group.} Lets put this to practice for the task of constructing an equivariant quantum neural network. Consider a QML model as in Eq.~\eqref{eq:QML-model} of the form 
\begin{align}\label{eq:QML-model-qnn}
    h_{\vec{\theta}}(\rho_i)=\Tr[\WC_{\vec{\theta}}(\rho_i) M_i]=\Tr[W(\vec{\theta})\rho_i W\ad(\vec{\theta}) M_i]\,,
\end{align}
where $\WC_{\vec{\theta}}:B(\HC)\rightarrow B(\HC)$ with $\WC_{\vec{\theta}}(\rho_i)=W(\vec{\theta})\rho_i W\ad(\vec{\theta})$   a parameterized unitary quantum neural network. Our goal here is to construct a $\WC_{\vec{\theta}}$ respecting the symmetries in the dataset. Recall from Eq.~\eqref{eq:equiv} that $\WC_{\vec{\theta}}$ will be equivariant if
\begin{equation}
    W(\vec{\theta})U_g\rho_iU_g\ad W\ad(\vec{\theta})= U_gW(\vec{\theta})\rho_i W\ad(\vec{\theta})U_g\ad\,,\label{eq:equiv-qnn}
\end{equation}
for all $g$ in $G$, and for all $\vec{\theta}$. That is, we need
\begin{align}\label{eq:equiv2}
    [W(\vec{\theta}),U_g]=0\,, \quad \forall \vec{\theta}\,,\,\,\forall g\in G.
\end{align}

Solving Eq.~\eqref{eq:equiv2} means determining \textit{all} the unitaries $W(\vec{\theta})$ that commute with \textit{all} the representations $U_g$ of the elements of the symmetry group $G$. Quite the daunting task! However, we can use  the trick of passing to the Lie algebra, solving there, and going back to the group. Explicitly, let us consider the case where $W(\vec{\theta})$ is composed of a single ``layer'', which is a fancy way of saying that $W(\vec{\theta})=e^{-i\theta H }$, for some Hermitian operator $H$, and for some trainable parameter $\theta\in\mathbb{R}$.  In particular, since~\eqref{eq:equiv2} must hold for all $\theta$, it must hold for infinitesimal parameters. We can again use the Taylor expansion trick as in Eq. ~\eqref{eqn:computational trick for Lie groups to Lie algebras} to expand around $\theta=0$. So $e^{-i\theta H }=\id+\theta H+\OC(\theta^2)$ in Eq.~\eqref{eq:equiv2} yields
\begin{align}\label{eq:equiv3}
    [W(\vec{\theta}),U_g]&=\theta [H,U_g]+\OC(\theta^2)\,,
\end{align}
which is zero (to first order) if $[H,U_g]=0$. That is,  the quantum neural network $W(\vec{\theta})$ will be equivariant if its generator $H$ commutes with all the representations of the group elements. In fact, one can check that
\begin{equation}\label{eq:equiv-generator}
    [H,U_g]=0\,, \quad \forall g\in G\,,
\end{equation}
is enough to guarantee that all remaining higher orders terms will also commute with $U_g$~\cite{meyer2022exploiting}. Not surprisingly, it is easier to solve  $[H,U_g]=0$ at the algebra level than to solve  $[W(\vec{\theta}),U_g]$ at the group level. Here we further remark that if G is a Lie group, then we can map to the Lie algebra and find all operators such that  $[H,X]=0$ for all $X\in \mathfrak{g}$. This exemplifies  how working at the algebra level of the layer is
easier. Note that these results can be generalized to a quantum neural network with multiple layers. Here, $W(\vec{\theta})=\prod_l e^{-i\theta_l H_l }$, where $W(\vec{\theta})$ will be equivariant if each layer is equivariant. Thus, we require $[H_l,U_g]=0$, $\forall H_l$ and $\forall g\in G$. 

As an example, consider the QML task of Fig.~\eqref{fig:fig3} of classifying real valued data $x=(x^1,x^2)$ living in a two-dimensional plane. Here, the data is encoded into two qubit states, and the symmetry group is $\mathbb{Z}_2$ whose representation acting on two qubits is $G_{\SWAP}=\{\id,\SWAP\}$ of Eq.~\eqref{eq:swap-group}. From the previous result we have that a quantum neural  network will be equivariant if its generators commutes with the identity operator (always trivially true), and with the SWAP operator. The reader can verify that any generator $H$ in $\Span(\{X\otimes \id + \id\otimes X,Y\otimes \id + \id\otimes Y ,Z\otimes \id + \id\otimes Z ,X\otimes X ,Y\otimes Y ,Z\otimes Z \})$, satisfies $[H,\SWAP]=0$ (below we will explicitly show how this set can be found, we also refer the readers to Refs.~\cite{larocca2022group,meyer2022exploiting,sauvage2022building,nguyen2022atheory}). In Fig.~\ref{fig:fig8} we show an example of a quantum neural network that is constructed by exponentiation of these generators, and thus, that is equivariant to $G_{\SWAP}$.

\begin{figure}
    \centering
    \includegraphics[width=1\columnwidth]{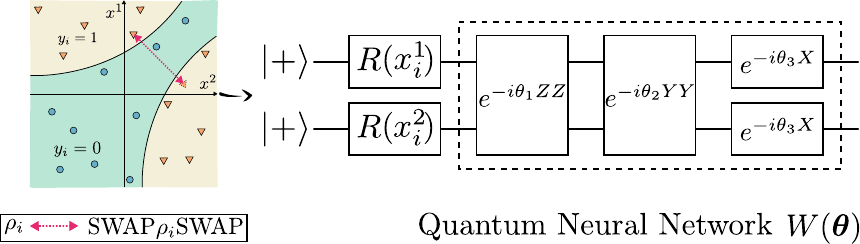}
    \caption{\textbf{Example of equivariant quantum neural network.} Here we consider the task in Fig.~\ref{fig:fig3} of classifying real valued data $x=(x^1,x^2)$ living in a two-dimensional plane. We use tools from the Lie group / Lie algebra correspondence to build equivariant quantum neural networks to the symmetry group of the task, $G_{\SWAP}=\{\id,\SWAP\}$.}
    \label{fig:fig8}
\end{figure}

\begin{tensor_boxed_example}[label={box:tensor rep of alg su(2)}]{Tensor rep of $\su(2)$ on 2 qubits}
Return to the tensor representation of $SU(2)$ on two qubits $V = (\Cbb^2)^{\otimes 2}$, which we recall decomposes as $V = \text{Sym}^2(\Cbb^2)\oplus \text{Alt}^2(\Cbb^2)$.
Just as in the previous example, we study the induced representation on the Lie algebra by writing $U_g(\theta) = e^{i\theta H}$ with $H\in \su(2)$.
Then using Leibniz rule for tensor products,
\begin{align*}
    \frac{d}{d\theta}\paran{U_g(\theta)\otimes U_g(\theta)}\Big|_{\theta = 0} &= \frac{d}{d\theta}\paran{e^{i\theta H}\otimes e^{i\theta H}}\Big|_{\theta = 0} \\
    &= i\Big(H e^{i\theta H}\otimes e^{i\theta H} \\
    &\qquad + e^{i\theta H}\otimes H e^{i\theta H} \Big)\Big|_{\theta = 0} \\
    &= i(H\otimes \idty + \idty \otimes H ).
\end{align*} So, the Lie group representation $g\mapsto U_g\otimes U_g^\dagger \in GL(V)$ induces the Lie algebra representation $H\mapsto H \otimes \idty + \idty \otimes H \in \mathfrak{gl}(V)$.
Again, as demanded by Theorem \ref{thm:lie group reps yield lie alg reps}, this representation respects the decomposition of $V$ into invariant subspaces: observe that for all $H\in \su(2)$, $H\otimes \idty + \idty\otimes H$ maps antisymmetric vectors to antisymmetric vectors and maps symmetric vectors to symmetric vectors, so $V = \text{Sym}^2(\Cbb^2)\oplus \text{Alt}^2(\Cbb^2)$ as Lie algebra representations.
\end{tensor_boxed_example}

\subsection{Representation Theory: the fundamentals}
Our working examples the entire time have been representations of a few key groups.
It is high time we listed key definitions to provide formal structure to the pictures they paint.
First, let us tackle a conceptually different aspect of representation theory: \textit{Given a representation of a symmetry group. How can I build a new representation?} Answering this question can be both conceptually rich, but also practically useful. From a theoretical stand-point there is an ``upward direction'' where we can take  simple group representations and build more complex ones. More interestingly, however, is the ``downward direction'' where we can try to understand very complex representations of a group in terms of simpler ones. From a practical perspective, we will see that change of representations are extremely useful in certain types of quantum neural networks~\cite{cong2019quantum,pesah2020absence}.

\subsubsection{Constructing new representations, changing representation}

Note that in what follows we will construct new representations from two given representations $R_1,R_2$. The generalization to $R_1,\dots, R_m$ is straightforward.

We start by noting that up to this point we have corresponding definitions at the level of the Lie group $G$ and the Lie algebra $\liea$. Really, we should have Theorem \ref{thm:lie group reps yield lie alg reps} and Corollary \ref{cor:lie alg reps locally induce Lie group reps} at the front of our mind for each and every one of these: the definition at the group level induces the algebra definition, and the definition at the algebra level induces the group definition. 
It is straightforward (albeit not pedagogically wise) to show this equivalence using these theorems, often using some version of the Taylor series expansion computational trick (Eq.~\eqref{eqn:computational trick for Lie groups to Lie algebras}).

First, let us recall that every Lie group $G$ has a particularly important representation associated to it.\footnote{Lie group structure is highly determined by this representation, and it plays a central role in the classification of so-called semisimple Lie groups.}
\begin{definition}\label{def:adjoint rep}
Let $G$ be a Lie group with associated Lie algebra $\liea$. 
The \emph{adjoint representation} $\text{Ad}: G\to GL(\liea)$ is the map $g\mapsto \text{Ad}_g$ defined by
\[
    Ad_g(X) = gXg^{-1}, \qquad \text{for all } X\in \liea.
\] Using Theorem \ref{thm:lie group reps yield lie alg reps}, this induces the \emph{adjoint representation} of $\liea$, $\text{ad}: \liea \to L(\liea)$ is the map $X\mapsto \text{ad}_X = [X,\cdot]$ given by
\[
    \text{ad}_X(Y) = [X,Y], \qquad \text{for all } Y\in \liea.
\]  
\end{definition} 
Since quantum science primarily works with complex representations, it is common to write, as we do in this article, the adjoint representation as $V=\liea$ while meaning the complexified Lie algebra $V=\liea\otimes \Cbb$.

Next, we let us highlight the remarkable fact that  that given two representations of a group, their direct sum is also a representation. 
\begin{definition}\label{def:direct sum}
Let $R_1,R_2$ be representations of a Lie group $G$ acting on vector space $V_1,V_2$.
The \emph{direct sum} $R_1\oplus R_2$ is a representation of $G$ acting on $V_1\oplus V_2$ defined by
\[
    \paran{R_1\oplus R_2(g)}(v_1,v_2) = \paran{R_1(g) v_1, R_2(g) v_2}, \quad \text{for all }g\in G.
\] 
Likewise, if $r_1,r_2$ are representations of a Lie algebra $\liea$ on $V_1,V_2$, the \emph{direct sum} $r_1\oplus r_2$ is a representation of $\liea$ acting on $V_1\oplus V_2$ defined by
\[
    \paran{r_1\oplus r_2(X)}(v_1,v_2) = \paran{r_1(X) v_1, r_2(X) v_2}, \quad \text{for all }X\in \liea.
\]
\end{definition}
It can be useful to think of the direct sum in \textit{block diagonal form}: the representation $R_1\oplus R_2$ acting on $V_1\oplus V_2$ can be written as
\[
    (R_1\oplus R_2)(g) = \begin{pmatrix} R_1(g) & 0 \\ 0 & R_2(g)
    \end{pmatrix}\,, \quad \text{for all }g\in G. 
\] 
Of course, a similar form holds for $r_1\oplus r_2$. Notice that a direct sum decomposition appeared in both of our key examples in Examples~\ref{box:adjoint rep of group SU(2)},~\ref{box:tensor rep of group SU(2)},~\ref{box:adjoint rep of alg su(2)},~\ref{box:tensor rep of alg su(2)}.
Block diagonalization will be of the utmost importance throughout this paper. We finally note that
one can also take the direct sum of the same representation, i.e., $R_1\oplus R_1$, in which case we say that $R_1$ has multiplicity of two, and we write 
\[
    (R_1\oplus R_1)(g) = \id_2\otimes R_1(g)\,, \quad \text{for all }g\in G. 
\]

For instance, let us consider $G$ to be $SU(2)$. Then, take $R_1(g)$ be the representation $\text{Sym}(\Cbb^2)$ from Example~\ref{box:tensor rep of group SU(2)} (often called the spin-1 representation), and $R_2(g)$ be the representation $\text{Alt}^2(\Cbb^2)$ from the same example (often called the spin-0 representation, or singlet). Their direct sum is of the form
\begin{equation}\label{eq:su2-block-sum}
    (R_1 \oplus R_1)(g) =\begin{pmatrix} \hspace{.25mm}
    \fbox{ $\begin{array}{c c c}
    \,\,\,\,\,\,\,\, & & \\
    & & \\
    & & \\
    \end{array}$} \hspace{-.75mm}&  \begin{array}{c}
     0 \\
    0 \\
    0
    \end{array}    \\
    \begin{array}{c c c}
    0 & 0 & 0 
    \end{array} &\fbox{ $\begin{array}{c}
    \end{array}$}
    \end{pmatrix}\,.
\end{equation}

We have also already seen through  Example~\ref{box:tensor rep of group SU(2)} that the tensor product of two representations is also a representation. Moreover, from Example~\ref{box:tensor rep of alg su(2)} we also saw how Leibniz's rule for differentiation induces the following form of tensor representations of a Lie algebra $\liea$.
\begin{definition}\label{def:tensor rep}
Let $R_1,R_2$ be representations of a Lie group $G$ acting on vector spaces $V_1,V_2$.
The \emph{tensor representation} $R_1\otimes R_2$ of $G$ on $V_1\otimes V_2$ is
\[
    (R_1\otimes R_2)(g) = R_1(g) \otimes R_2(g) .
\] If $r_1,r_2$ are representations of a Lie algebra $\liea$ acting on $V_1,V_2$, then the \emph{tensor representation} $r_1\otimes r_2$ of $\liea$ on $V_1\otimes V_2$ is
\[
    (r_1\otimes r_2)(X) = r_1(X) \otimes \idty + \idty \otimes r_2(X) .
\]
\end{definition}

Here, let us consider again the case of $G=SU(2)$. Recalling from Example~\ref{box:tensor rep of group SU(2)} that in the symmetric-antisymmetric basis  we have
\begin{equation}\label{eq:su2-block-prod}
    U_g\otimes U_g =\begin{pmatrix} \hspace{.25mm}
    \fbox{ $\begin{array}{c c c}
    \,\,\,\,\,\,\,\, & & \\
    & & \\
    & & \\
    \end{array}$} \hspace{-.75mm}&  \begin{array}{c}
     0 \\
    0 \\
    0
    \end{array}    \\
    \begin{array}{c c c}
    0 & 0 & 0 
    \end{array} &\fbox{ $\begin{array}{c}
    \end{array}$}
    \end{pmatrix}\,.
\end{equation}
The fact that there exists a basis where $ U_g\otimes U_g$ in Eq.~\eqref{eq:su2-block-prod} has the exact block diagonal form as that in Eq.~\eqref{eq:su2-block-sum} indicates that these two representations are isomorphic (Definition~\ref{def:equivariance}). More importantly, it unravels a fundamental aspect of representation theory: \textit{given a representation (such as the tensor one), it is extremely useful to express it as a direct sum of smaller representations.} For our previous example, we now know that the tensor representation of $SU(2)$ can be decomposed as a direct sum of two ``smaller'' representations, the spin-$1$ and the spin-$1/2$. We will return to this in Example~\ref{box:equivalence of adjoint and tensor reps}.

The previous examples show how to build larger representations from smaller ones (though direct sums or products). There is another key construction which allows us to build new representations from old, but this one keeps the dimension the same, a commonly useful feature for change of representation tasks.

Recall that the dual space of a finite dimensional vector space $V$, where we think of $V$ as spanned by column vectors, is the space of row vectors $V^*$, i.e. linear functionals $\{V\to \mathbb{F}\}$ with $\mathbb{F} \in \{\R,\Cbb\}$.
\begin{definition}\label{def:dual representation}
Let $R$ be a representation of $G$ acting on a finite-dimensional vector space $V$.
Then the \emph{dual representation} $R^*$ is the representation of $G$ acting on $V^*$ given by
\[
    R^*(g) := [R(g^{-1})]^T\,, \qquad \text{for all } g\in G,
\] where $T$ denotes the transpose. 
If $r$ is a representation of $\liea$ on $V$, then the \emph{dual representation} $r^*$ of $\liea$ on $V^*$ is given by
\[
    r^*(X) := -r(X)^T\,, \qquad \text{for all } g\in G\,.
\]
\end{definition} Note that this is also called the \textit{contragredient} representation.

\begin{figure*}[t]
    \centering
    \includegraphics[width=.8\linewidth]{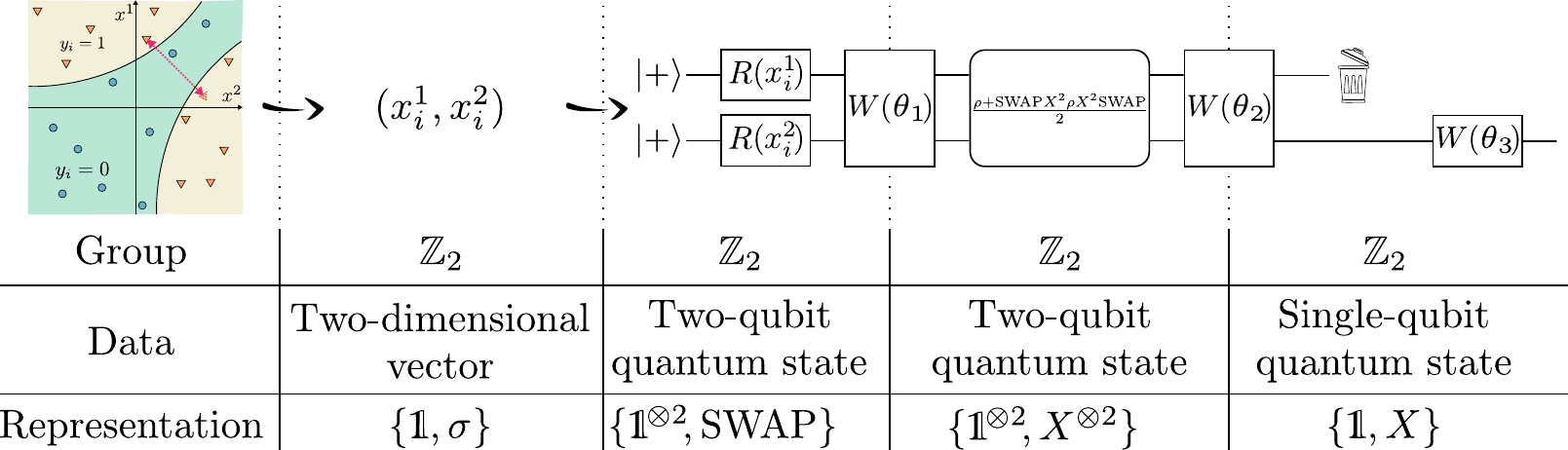}
    \caption{\textbf{Change of representations in QML.} We consider the task in Fig.~\ref{fig:fig3} of classifying real valued data $(x^1,x^2)\in\mathbb{R}^2$. Here, the data labels are symmetric under the action of $\mathbb{Z}_2$, whose representation is $G_1=\{\id,\sigma\}$, where $\sigma\cdot(x^1,x^2)=(x^2,x^1)$. The first step is to embed the data into the state of two qubits via some data-encoding scheme that effectively changes the representation of $\mathbb{Z}_2$ to  $G_2=\{\id,\SWAP\}$. Next, one can process the information via some  $G_2$-equivariant quantum neural network $W(\vec{\theta}_1)$ such as that in Fig.~\ref{fig:fig8}. Further, as we show below we can find a map $\Phi_1:(\Cbb^2)^{\otimes 2}\rightarrow (\Cbb^2)^{\otimes 2}$, given by $\Phi_1(\rho)=\frac{1}{2}(\rho + \SWAP X^2\rho X^2 \SWAP)$ that changes the representation of the group to $G_3=\{\id,X^{\otimes 2}\}$ such that we can process the information with a  $G_3$-equivariant quantum neural network $W(\vec{\theta}_2)$. Next, we show a \textit{pooling} map  $\Phi_2:(\Cbb^2)^{\otimes 2}\rightarrow \Cbb^2$ (given by the partial trace) which maps two-qubit states to single-qubit states (note that the partial trace can also output mixed quantum states), and which also changes the representation of the symmetry group to $G_4=\{\id,X\}$. Finally, we can process the information encoded in the single qubit via a $G_4$-equivariant quantum neural network $W(\vec{\theta}_3)$.
}
    \label{fig:fig9}
\end{figure*}

The dual representation and tensor representation constructions immediately yield natural representations of $G$ on linear maps $\Phi: V\to W$ (which includes operators and quantum channels) and more generally tensors with mixed covariance and contravariance.
This occurs because we can think of the space of linear maps $\{\Phi:V\to W\}$ as the tensor product $V^*\otimes W$, where a basis for this space is given by, e.g., $\ket{w}\bra{\overline{v}}$ where we have $\ket{v}\in V, \ket{w}\in W$ and $\overline{v}$ denotes the complex conjugate of $v$.
In this case, given representations $R$, $S$ on $V$ ,$W$, respectively, the natural action on $\Phi$ is given by $R^*\otimes S$ for any $g\in G$ is
\[
    \paran{(R^*\otimes S)(g) \cdot \Phi }(v) = S_g \Phi (R_g^{-1}v)\,, \quad \text{where } v\in V,
\] which, when $R$ and $S$ are unitary representations (Definition \ref{def:unitary representations}) as is usually the case in QML, can be rewritten as
\begin{equation}
    \paran{(R^*\otimes S)(g) \cdot \Phi}(v) = S_g \Phi (R_g^\dagger v).
\end{equation} The kicker of this form is that studying how representations act on linear maps (and thus quantum channels!) boils down to a special case of tensor representations.\footnote{It should be noted that for some special cases, the dual representation $R^*$ is isomorphic to the original representation $R$ (this is the case for e.g. $SU(2)$), but in general this is not true, for instance any $SU(d)$ with $d\geq 3$. 
It is however true that $R$ is irreducible iff $R^*$ is, and that $(R^*)^* \cong R$.}
For instance, the vector space of equivariant maps $\Phi$ (for unitary reps $S,R$), is given by 
\begin{equation}
    \Phi = S_g \circ \Phi \circ R_g^\dagger\,,
\end{equation}
which is the collection of trivial 1D invariant subspaces of $V^*\otimes W$, since saying $\Phi$ is equivariant is the same as saying $(R^*\otimes S)(g)\cdot \Phi = \Phi$ (see Definition \ref{def:invariant subspace}).\footnote{The trivial representation always breaks into 1D trivial representations---this is almost immediate from the definition.}
In many cases, since tensor representations of common groups are well understood (re: Clebsch-Gordan decomposition), this shift in perspective can be both theoretically and computationally useful.

We refer the reader to Fig.~\ref{fig:fig9} where we showcase how different representations can play a role in a QML classification task. Notably, this figure depicts several key ingredients in QML tasks such as \textit{embedding layers} (mapping classical data to quantum data), \textit{quantum neural network layers} (to process the information in the quantum states), and \textit{data-pooling layers} (where qubits in the system are discarded with the hope of reducing the feature space dimension while preserving the relevant data features). We note that the field of change of representation in GQML is still  in its infancy (see~\cite{nguyen2022atheory} for some theory).

\subsubsection{Structure definitions and basic theory}
We now have a rich supply of contextualized examples: time to study their structure.

To begin, we need an appropriate definition of a subrepresentation of a representation.
\begin{definition}\label{def:invariant subspace}
Let $R:G\to GL(V)$ be a group representation. 
An subspace $W\subseteq V$ is a $G$-\textit{invariant} linear subspace of $V$, meaning for all $w\in W$, 
\[
    R_g \cdot w \in W \qquad \text{for all }g\in G.
\]

Similarly, if $r:\liea \to \mathfrak{gl}(V)$ a Lie algebra representation, then an \emph{invariant subspace} $W\subseteq V$ is a $\liea$-invariant linear subspace of $V$, meaning for all $w\in W$,
\[
    r_X \cdot w \in W \qquad \text{for all }X\in \liea.
\]
\end{definition} 

Of course, the trivial subspaces $\{1\}$ and $V$ are always invariant subspaces of $V$. 
But what about nontrivial ones?

\begin{definition}\label{def:reducible}
A representation $R$ (resp. $r$) of a Lie group $G$ (resp. Lie algebra $\liea$) over $V$ is called \emph{reducible} if there exists a nontrivial invariant subspace $W\subseteq V$. If the only invariant subspaces are the trivial subspaces $\{1\}$ and $V$, then $V$ is called \emph{irreducible}.
\end{definition}
Invariant subspaces are closely connected to simultaneous block diagonalization and direct sums (Definition \ref{def:direct sum}): this connection will be fleshed out by Definition \ref{def:complete reducibility}.
Notice that we have already been exposed to this notion via example: for the $SU(2)$ representation $g\mapsto U_g(\cdot)U_g^\dagger$ on $M_2(\Cbb)$ in Example~\ref{box:adjoint rep of group SU(2)}, we found that the spaces $\text{span}\{\idty\}$ and $\text{span}\{X,Y,Z\}$ are invariant subspaces; and for the $SU(2)$ representation $g\mapsto U_g\otimes U_g$ on $(\Cbb^2)^{\otimes 2}$ in Example~\ref{box:tensor rep of group SU(2)}, the symmetric and antisymmetric spaces are invariant subspaces.
To precisely pin down the link to simultaneous block diagonalization, we need the notion of total reducibility.
\begin{definition}\label{def:complete reducibility}
Let $R$ (resp. $r$) be a representation of $G$ (resp. $\liea$) on a vector space $V$. 
Then $V$ is called \emph{completely reducible} if there exists a direct sum decomposition of $V$ into subspaces $W_1,\dots, W_k$ 
\[
    V = W_1 \oplus W_2 \oplus \dots \oplus W_k 
\] where each $W_j$ is an $R$ (resp. $r$) invariant subspace such that the restriction $R_j := R\vert_{W_j}$ is irreducible.
\end{definition} So, when $R$ (or $r$) is a completely reducible representation on $V$, there exists a basis of $V$ such that we have simultaneous block diagonalization:
\[
    R(g) = \begin{pmatrix}
    R_1(g) & 0 & 0 & 0 \\
    0 & R_2(g) & 0 & 0 \\
    0 & 0 & \ddots & 0 \\
    0 & 0 & \dots & R_k(g)
    \end{pmatrix}\,, \qquad \text{for all } g\in G.
\] 
This definition is key, because it establishes that any completely reducible representation can be constructed out of the ``building blocks'' of irreducible representations, much like how we use prime numbers to construct all other natural numbers.
But the situation is a bit more subtle here, because not all representations are completely reducible.\footnote{Representations which are not completely reducible are remarkably common: one can show for instance that the representation of the group $\R$ with addition given by $R:\R\to GL(2;\Cbb)$ $x\mapsto \begin{pmatrix} 1 & x \\ 0 & 1\end{pmatrix}$ is not totally reducible.}
This begs a question: can we find a class of representations which are completely reducible?
The good news is that yes, and in fact, these are exactly the types of representations we care about for quantum applications: unitary representations (Theorem \ref{thm:complete reducibility of unitary representations}).
For this next definition, once again recall the Taylor series expansion computational trick (Equation \ref{eqn:computational trick for Lie groups to Lie algebras}).

\begin{definition}\label{def:unitary representations}
Let $G$ be a Lie group.
A \emph{unitary representation} of $G$ is a representation $R: G\to GL(V)$ such that $R_g^\dagger R_g = \idty$ for all $g\in G$.

Let $\liea$ be a Lie algebra.
A \emph{skew-Hermitian} representation of $\liea$ is a representation $r: \liea \to \mathfrak{gl}(V)$ with $r_X^\dagger= - r_X$ for all $X\in \liea$.
\end{definition}

\begin{tensor_boxed_example}[label={box:reducibility of tensor rep of SU(2)}]{Tensor rep of SU(2): reducibility}
We have established that the tensor representation of $SU(2)$ on two qubits decomposes as $V = \text{Sym}^2(\Cbb^2) \oplus \text{Alt}^2(\Cbb^2)$.
We now know that this is a unitary representation and by Theorem \ref{thm:complete reducibility of unitary representations}, it is completely reducible.
$\text{Alt}^2(\Cbb^2)$ is obviously irreducible since it is 1 dimensional, but what about $\text{Sym}^2(\Cbb^2)$?
Here, the power of passing to the Lie algebra representation of $\su(2)$ becomes evident.
Consider the representative $Z\otimes \idty + \idty \otimes Z$ acting on the basis $\ket{11}, \ket{01} + \ket{10}, \ket{00}$. 
Notice that it is diagonal in this basis.
Now, let us think about the ``creation'' and ``annihilation'' operators $\tilde{\sigma}^\pm = \sigma^\pm\otimes \idty + \idty \otimes \sigma^\pm$ where
\[
    \sigma^\pm = \frac{1}{2}\paran{X \mp iY} .
\] Clearly, $\tilde{\sigma}^\pm$ are in the image of the complex representation $\su(2)$, since they are complex linear combination of $X$ and $Y$ representatives. 
Observe their action on the basis $\text{Sym}^2(\Cbb^2)$:
\begin{align*}
    \ket{11} \xmapsto{\tilde{\sigma}^- \cdot } \paran{\ket{10}+\ket{01}} \xmapsto{\tilde{\sigma}^- \cdot } 2\ket{00}    \\
    \ket{00} \xmapsto{\tilde{\sigma}^+ \cdot } \paran{\ket{10}+\ket{01}} \xmapsto{\tilde{\sigma}^+ \cdot } 2\ket{11}\,.
\end{align*} In other words, there are no invariant subspaces of $\text{Sym}^2(\Cbb^2)$ since we have representatives from $\su_2$ permuting all subspaces.
\vspace{0.25cm}

This powerful strategy pervades representation theory: simultaneously diagonalize a set of commuting operators in $\liea$ (here, just a single operator $Z$, but in general, a so-called \textit{Cartan subalgebra}) to get an eigenbasis, then use other operators and their commutation relations to ``move around'' in this eigenbasis.
\end{tensor_boxed_example}

Again, using Theorem \ref{thm:lie group reps yield lie alg reps} and Corollary \ref{cor:lie alg reps locally induce Lie group reps}, a unitary representation $R$ of $G$ induces a skew-Hermitian representation of $\liea$, and a skew-Hermitian representation of $\liea$ induces (at least locally) a unitary representation $R$ of $G$.

\begin{theorem}[Complete reducibility of unitary representations]\label{thm:complete reducibility of unitary representations}
Any finite dimensional unitary representation is completely reducible. 
\end{theorem}

\begin{purple_boxed_example1}[label={box:equivalence of adjoint and tensor reps}]{Equivalence of the $SU(2)$ reps $U_g(\cdot)U_g^\dagger$ and $U_g\otimes U_g$}
We have learned quite a bit about the reps $U_g(\cdot)U_g^\dagger$ and $U_g\otimes U_g$: the first decomposes as $M_2(\Cbb) \cong \su(2)\oplus \Cbb\idty$, and the second as $(\Cbb^2)^{\otimes 2} \cong \text{Sym}^2(\Cbb^2)\oplus \text{Alt}^2(\Cbb^2)$. 
As it turns out, these are equivalent representations in the sense of Definition~\ref{def:equivariance}.
We begin by noticing that it is not too hard to see that as $SU(2)$ representations, $\Cbb\idty \cong \text{Alt}^2(\Cbb^2)$, i.e., both are trivial representations. 
That leaves showing that the adjoint representation $\su(2)$ is equivalent to $\text{Sym}^2(\Cbb^2)$, so we need to find an equivariant map between the two.

Let us look at the computational (diagonal) basis of $\text{Sym}^2(\Cbb^2)$ from Example~\ref{box:reducibility of tensor rep of SU(2)}
and compute eigenvalues for the $Z$ representative:
\begin{align*}
    (Z\otimes \idty + \idty \otimes Z)\ket{11} &= 2\ket{11} \\
    (Z\otimes \idty + \idty \otimes Z)(\ket{10}+\ket{01}) &= 0 \\
    (Z\otimes \idty + \idty \otimes Z)\ket{00} &= -2\ket{00}.
\end{align*} Now, let us diagonalize the $Z$ representative for the adjoint representation (Definition \ref{def:adjoint rep}):
\begin{align*}
    ad_Z(\sigma^-) &= [Z,\sigma^-] = 2\sigma^- \\
    ad_Z(Z) &= [Z,Z] = 0 \\
    ad_Z(\sigma^+) &= [Z,\sigma^+] = -2\sigma^+.
\end{align*} One can then check that the linear map $\phi:\text{Sym}^2(\Cbb^2) \to \su(2)$ which sends 
\begin{align*}
    \ket{11}&\xmapsto{\phi} \sigma^{-}\\ \ket{10}+\ket{01}&\xmapsto{\phi} Z \\ \ket{00}&\xmapsto{\phi} \sigma^{+} ,
\end{align*} is indeed an equivariant map by checking on the representatives of the Lie algebra basis $\{\sigma^-,Z,\sigma^+\}$ of $\su(2)$.
By exponentiating, we see that $\phi$ intertwines the Lie Group representations, i.e. $\phi\circ U^{\otimes 2} = Ad_{U}\circ\phi$, or more explicitly, 
\[
    \phi \circ U_g\otimes U_g (\cdot) = U_g(\phi(\cdot))U_g^\dagger \qquad \text{for all } g\in SU(2).
\]

\end{purple_boxed_example1}

This result is enormously powerful: as we will see in Theorem \ref{thm:Weyl's unitary trick}, every compact Lie group (which includes every finite discrete group) representation is equivalent to a unitary representation. To state this theorem however, we need an appropriate definition of equivalence of representations.

\begin{definition}\label{def:equivariance}
Let $R,S$ (resp. $r,s$) be representations of a Lie group $G$ (resp. $\liea$) on vector spaces $V_R,V_S$.
Then we call a linear map $\phi:V_R\to V_S$ an \emph{equivariant map} if 
\[
    \phi \circ R_g = S_g \circ \phi \qquad \text{for all } g\in G
\]
If further $\phi$ is invertible, the representations $R,S$ are called \emph{equivalent} or \emph{isomorphic}.
We also say $V_R$ and $V_S$ are \emph{equivalent} as representations.
\end{definition} Note that equivariant maps are also commonly called \textit{intertwiners}.

In Example~\ref{box:equivalence of adjoint and tensor reps} we return to the case of $SU(2)$ and show that its adjoint and tensor representations are isomorphic. Here, there are two important conceptual comments to be made. Firstly, representations are in general not uniquely defined by their dimension (although this happens to be true for $SU(2)$): there are counterexamples for all $SU(d)$ with $d\geq 3$, wherein the same dimensional vector space may have several inequivalent representations.
Secondly, we see the powerful strategy for detecting equivalent representations (partially aforementioned after  Example~\ref{box:reducibility of tensor rep of SU(2)}), at play:\textit{ diagonalize the representatives of $Z$ in both representations} (here, the maps $Z\otimes \idty + \idty\otimes Z$ on $\text{Sym}^2(\Cbb^2)$ and $ad_Z(\cdot)$ on $\su(2)$), \textit{then use other operators and commutation relations to move between these eigenbases}. 

Let us see how this strategy played out. We began the tensor example by diagonalizing the $Z$ representative in the tensor rep by choosing the basis $\ket{11}, \ket{10}+\ket{01}, \ket{00}$. In this basis, the $Z$ representative has respective eigenvalues 2, 0, and -2. Then, the raising and lowering operators $\sigma^\pm$ map between these eigenspaces, which looks like ``raising'' or ``lowering'' the eigenvalue. It becomes clear that we can map any of these eigenspaces to each other by raising and/or lowering operators: thus, we have an irreducible subspace. When we got to the adjoint representation and saw that $ad_Z(\cdot)$ had the same eigenvalues, this suggested a similar strategy should work. 
There is nothing inherently special about the $Z$ operator here--it just so happens to be fairly easy to diagonalize in both cases. This generalizes to more complicated Lie algebras by simultaneously diagonalizing a commuting collection of operators (called a \textit{Cartan subalgebra}) and cleverly tracking their eigen-information using so-called ``weights''. In practice, diagonalization in a given representation may be highly nontrivial. Checking the reducibility of any acquired invariant subspace may also be a difficult task, and often requires either problem specific insight or any variety of representation theoretic tools.

It cannot however be understated how crucial this strategy is. It undergirds the classification of semisimple Lie algebras. It also elegantly frames the analysis of the quantum harmonic oscillator: there, we diagonalize the self-adjoint number operator $a^\dagger a$ to obtain eigenstates, and then use the canonical commutation relations $[x,p]=i\hbar$ combined with positivity of $a^\dagger a$ to move between eigenstates via ``raising and lowering'' ladder operators. This is well trodden and exposited ground, and every reference we list contains this computation. In Hall~\cite{hall2013lie}, this computation appears in the section ``Representations of $\mathfrak{sl}_2(\Cbb)$'' (recall that as complexified Lie algebras, $\su(2) \cong \mathfrak{sl}_2(\Cbb)$, and so they have the same representation theory).

\begin{theorem}[Weyl's unitary trick for complete reducibility] \label{thm:Weyl's unitary trick}
Let $G$ be a compact Lie group. 
Then every finite dimensional representation of $G$ is equivalent to a unitary representation, and by Theorem \ref{thm:complete reducibility of unitary representations}, completely reducible.
\end{theorem}
While we omit the (brief) proof here, there are two important messages from it to be mentioned:
\begin{itemize}
    \item The equivalence is established by performing a Haar averaging procedure\footnote{This averaging procedure is called ``Weyl's unitary trick'', hence the name for this theorem.} to the dot product $\langle \cdot, \cdot \rangle$ on $V$, much akin to that of the twirl operator (we will visit these in Sections~\ref{sec:haar integration},~\ref{sec:twirling}).
    These sorts of Haar averaging methods are of great importance across quantum computation and representation theory alike, since the Haar measure is the most ``natural'' probability distribution on Lie groups.
    \item Compactness is critical here: without it, the Haar measure of our Lie group $\abs{G}$ may not be finite and the averaging procedure is nonsensical.
\end{itemize}

The final result we mention is perhaps the most critical weapon any researcher needs in their arsenal before using representation theory to fully tackle QML tasks. 
Recall that we mentioned earlier that it is common abuse of language to call the vector space $V$ a representation of $G$ or $\liea$ without explicitly describing the homomorphism $R$ or $r$.
\begin{theorem}[Schur's lemma] \label{lemma:Schur's lemma}
Let $V$ and $W$ be irreducible complex representations of a group $G$ or Lie algebra $\liea$, and let $\phi:V\to W$ be an equivariant map. 
Then either $\phi = 0$, or $\phi$ is an isomorphism and $V\cong W$.

If $\phi:V\to W$ is an isomorphism, it must be that $\phi = \lambda \idty$ for some $\lambda \in \Cbb$.
\end{theorem}

Schur's lemma plays a central role throughout representation theory and its applications. \textit{A priori}, there are many linear maps between vector spaces $V$ and $W$, and the list of suspects for equivariant maps could be intimidatingly large. But Schur's lemma tells us that the equivariance condition, which can be interpreted as ``the group action commutes with $\phi$'', i.e.
$\phi(g\cdot v ) = g\cdot \phi(v)$ for all $g\in G$, tightly constrains the space of equivariant maps: for inequivalent irreducibles, only the 0 map is equivariant, and for equivalent irreducibles, the space of equivariant maps is the 1 dimensional space spanned by the identity map $\idty$. Further, since Theorem \ref{thm:complete reducibility of unitary representations} guarantees every unitary representation can be decomposed into irreducibles, Schur's lemma in turn tightly constrains the space of equivariant maps between \textit{any} unitary representations. This is leveraged heavily in e.g. the commutant structure Theorem~\ref{thm:commutant structure}.

\begin{section}{Some representation theory-rich constructions in QML}\label{sec:some representation theory-rich constructions in QML}
Representation theory is all about block-diagonalization into irreps. 
When faced with symmetric tasks in QML, a powerful strategy is to understand how the task ``acts'' upon irreps.
Keep this core idea in mind as we review some commonly encountered objects and operations within quantum computing.
As a bird's-eye-view of where we are headed, we begin with the key tool of Haar integration, which affords us a version of integration that plays well with group structure. 
Just as Riemann integration grants access to a variety of useful averaging procedures, Haar integration yields ``group averaging'' procedures, which will lead us naturally to the twirling operator.
We then pivot to describing the commutant, the set of all operators commuting with a representation. 
This is intimately tied to block-diagonalization and can be thought of as the natural extension of a fundamental statement in quantum mechanics: ``two operators are simultaneously diagonalizable if and only if they commute''.
This leads to the Schur-Weyl duality, presented via example, which captures the interplay of a representation and its commutant.
We wrap up this section with a brief discussion of some highlight applications of this theory within GQML.

\begin{subsection}{Haar integration} \label{sec:haar integration}
Haar integration is ubiquitous in quantum science, because it is in some sense the most ``natural'' probability measure to place on unitary groups.
Depending on the level of generality, there are several well-exposited ways of defining the Haar measure on a compact group $G$.\footnote{In Hall~\cite{hall2013lie}, this is accomplished for matrix Lie groups via right-invariant differential forms furnished by the associated Lie algebra.}
But we will take a more pragmatic approach and just explain its key features.

Firstly, the Haar measure $d\mu$ of a compact Lie group defines, after proper rescaling, a probability distribution on a compact Lie group. 
This means we can describe sizes of ``nice'' (i.e., Borel) subsets $H\subseteq G$ by integrating $\frac{1}{\abs{G}}\int_{H\subseteq G} d\mu$, where $\abs{G} = \int_{G} d\mu$ (this is effectively a partition function in probability theory), and it means we can now integrate functions $f:G\to \R$, like the trace norm $\Tr \abs{\cdot}$, over these regions.
For instance, if we had a nice open region $H$ surrounding the identity $\idty$, perhaps describing a ball of $\varepsilon$ perturbations of $\idty$, we could now meaningfully integrate over that.

But what makes the Haar measure special is the so-called \textit{left-invariance} (or \textit{right-invariance}, depending on convention).\footnote{For a wide class of groups, called ``unimodular'', left invariance and right invariance are the same. This includes compact and semisimple Lie groups.} This implies the following relationship for any integrable function $f:G\to \R$:
\[
    \int_G f(g\cdot x) \, d\mu(x) = \int_G f(x) \, d\mu(x) , \qquad \text{for any } g\in G.
\] One fruitful interpretation of this definition is that this gives a \textit{change of variables} formula for integration over arbitrary Lie groups. 
To illustrate this, let us just think about calculus.
The real numbers with the addition $G = (\R,+)$ form a matrix Lie group, and using the standard Lebesgue measure on $\R$\footnote{...which is not a probability measure since $\mu(\R)=\infty$, but that is not the point here.} we can integrate real functions.
Then, writing $g\cdot x = g+x$ and $d\mu(x) = dx$, we have 
\[
    \int_{G= \R} f(g+x) \, dx = \int_{G=\R} f(x) \, dx .
\] The Haar measure is essentially the unique probability measure with this change of variable property, and given how useful change of variables is in any application of calculus, it is a short stretch of the imagination to see why this should be critical to any analysis on Lie groups. 
Finding closed form expressions for Haar measures of matrix Lie groups in terms of the more familiar Lebesgue measure from calculus boils down to combining standard methods for integration on manifolds using Jacobians with the ``left-invariant vector field'' interpretation of the Lie algebra~\cite{hall2013lie}.

We mention briefly that for any finite group $G$, the Haar measure is just the counting measure:
\[
    \int_G f(x) d\mu(x) = \sum_{x\in G} f(x) .
\]
The good news going forward is that we generally will be integrating over the whole group, and we only really need these formal properties, not the particular form of the measure.
In the following two sections, we describe two fundamental ideas which have close ties to Haar integrals: twirling and commutants.

\begin{subsection}{Twirling} \label{sec:twirling}
Twirling is a Haar-averaging technique that appears in a variety of disparate contexts, including entanglement-theory and quantum error correction~\cite{bennett1996mixed}, randomized benchmarking~\cite{knill2008randomized,magesan2011scalable}, quantum process tomography~\cite{emerson2007symmetrized,lu2015experimental} classical shadows \cite{huang2020predicting,elben2022randomized},  and barren plateau analysis~\cite{mcclean2018barren,cerezo2020cost,sharma2020trainability,pesah2020absence,holmes2021connecting,larocca2021diagnosing}, to say nothing of the numerous applications within representation theory. The following discussion is a rephrasing of the standard material presented in \cite{fulton1991representation}.

Every form of twirling has the following structure in common. Take a representation $R$ of a group $G$ over a vector space $V$. 
Suppose that we want to find the subspace of all trivial invariant subspaces $V^G$:
\[
    V^G := \{ v \vert R_g\cdot v = v\}.
\] 
Then the twirl operator $\TC_G: V\to V$ is the linear map
\[
    \TC_G (v) = \frac{1}{\abs{G}} \int_G R_g\cdot v \; d\mu(g) .
\] The twirl operator is in fact a projection onto $V^G$, so $\TC_G^2 = \TC_G$ and $\rm{Im}(\TC_G) = V^G$.
At a more abstract level, the twirl operator projects onto the 1D trivial invariant subspaces (trivial irreps) of the representation $V$.

Let us think about some special representations to shed light on the utility of this operator. 
\begin{enumerate}
    \item \textit{Equivariant unitary quantum neural networks:}
 \end{enumerate}   
 Let $U$ be a unitary representation of a group $G$ on a $d$-dimensional Hilbert space $\HC$, and let $W(\theta)=e^{-i \theta H}$ be a unitary quantum neural network with generator $H$. As discussed in Eqs.~\eqref{eq:equiv-qnn}--\eqref{eq:equiv-generator}, the quantum neural network will be equivariant if $[H,U_g]=0$ for all $g\in G$. Using the twirling formula, we can take any $H$ in $\su(d)$ and project it into $\su(d)^G$ (the set of all twirled, or projected, Pauli operators), so that the ensuing $e^{-i\theta\TC_G(H)}$ will generate a $G$-equivariant unitary. 
 
 For instance, consider again the QML task of Fig.~\ref{fig:fig3} where the symmetry acts through the adjoint representation of $G_{\SWAP}=\{\id,\SWAP\}$. Then, the twirl of  $H=(X\otimes \id)$ is
 \begin{align}
     \TC_G(H)&=\frac{1}{|G|}\sum_{g\in G} U_g (X\otimes \id) U_g\ad\nonumber\\
     &=\frac{1}{2}\left(\id (X\otimes \id)\id+\SWAP (X\otimes \id)\SWAP\right)\nonumber\\
     &=\frac{1}{2}(X\otimes \id+\id\otimes X)\,.
 \end{align}
 As shown in Fig.~\ref{fig:fig8}, $\TC_G(H)$ is precisely one of the generators one can use in a $G_{\SWAP}$-equivariant quantum neural network.

\begin{enumerate}
\setcounter{enumi}{2}
    \item \textit{Equivariant channels:}
 \end{enumerate} 
    The previous result can be generalized to obtain equivariant channels.   Consider two unitary representations $U, W$ of a group $G$ on Hilbert spaces $\mathcal{H}_1, \mathcal{H}_2$, and let $V$ be the space $V = \{\textit{linear maps }\mathcal{B}(\mathcal{H}_1) \to \mathcal{B}(\mathcal{H}_2) \}$.
    Consider a channel $\Phi\in V$.
    Per the discussion underneath Definition \ref{def:dual representation}, the natural representation on $\mathcal{B}(\mathcal{H}_1)$ is given by $O\mapsto U_g O U_g^\dagger$, and likewise for $\mathcal{B}(\mathcal{H}_2)$ with $V_g(\cdot)V_g^\dagger$.
    Then $\Phi$ is equivariant (Definition \ref{def:equivariance}) if for all $O\in \mathcal{B}(\mathcal{H}_1)$, 
    \[
        V_g^{\dagger} \Phi(U_g O U_g^\dagger)V_g  =\Phi(O),
    \] Then, the twirl operator $\TC_G: V\to V^G$ projects onto equivariant maps, so any equivariant map is the image of 
    \[
        \TC_G(\Phi) = \frac{1}{\abs{G}}\int_G V_g^{\dagger} \Phi(U_g (\cdot) U_g^\dagger)V_g \, d\mu(g).
    \] 
Notably, the twirl of a complete positive and trace preserving channel, it also completely positive and trace preserving, meaning that the twirl of a physical map, leads to a physical map.

For instance, consider a channel that changes between the representations $\{\id,X^{\otimes 2}\}$ and  $\{\id,\SWAP\}$ of $\mathbb{Z}_2$. We start with the trivial channel $\Phi(\rho)=\rho$, whose twirl is
\begin{equation}
    \TC_G(\Phi)=\frac{1}{2}\left(\rho + \SWAP X^{\otimes 2} \rho X^{\otimes 2}\SWAP \right)\,.
\end{equation}
There are a variety of methods to compute channel twirls in practice. See for instance a generalized teleportation procedure in Proposition 26 of~\cite{kaur2017amortized}, and a comparison of three methods in Table II of ~\cite{nguyen2022atheory}.

\begin{enumerate}    
\setcounter{enumi}{2}
    \item \textit{Unitary k-designs and tensor representations of $U(d)$}:
\end{enumerate}
    Let $G=U(d)$ and let $V = \mathcal{B}(\mathcal{H}^{\otimes k})$ with $\mathcal{H} = \Cbb^d$.
    Consider the fundamental representation $U: G\to GL(\mathcal{H})$ given by $U_g = g$.
    Combining the tensor representation (Definition~\ref{def:tensor rep}) and dual representation (Definition~\ref{def:dual representation}) constructions, the natural representation on an operator $O\in V$ is given by
    \[
        g\cdot O = (U_g)^{\otimes k} O (U_g^\dagger)^{\otimes k}.
    \] Then, using $dU := d\mu$ to denote the Haar measure on $U(d)$, the twirling operator $\TC_G^{(k)}$ here is given by
    \begin{equation} \label{eqn:unitary k design twirling}
        \TC_G^{(k)}(O) = \frac{1}{\abs{U(d)}} \int_{U(d)} (U_g)^{\otimes k} O (U_g^\dagger)^{\otimes k} \; dU(g) .
    \end{equation} Unitary $k$-designs are ensembles, i.e., ``nice'' (Borel) subsets $\mathcal{E}\subseteq U(d)$ equipped with a probability measure $d\mathcal{E}$ (often just the restriction of the Haar measure $dU$) which have that $\TC_\mathcal{E}^{(k)} = \tau^{(k)}$, where the modified twirl $\tau_\mathcal{E}^{(k)}$ operator is given by integrating the same representation over the ensemble $\mathcal{E}$:
    \[
        \TC^{(k)}_\mathcal{E} (O) = \frac{1}{\abs{\mathcal{E}}} \int_{\mathcal{E}} (U_g)^{\otimes k} O (U_g^\dagger)^{\otimes k} \; d\mathcal{E}(g) .
    \] Unitary $k$-designs are important because they characterize when an ensemble $\mathcal{E}$ captures the first $k$ moments of the Haar distribution on $U(d)$, which, raising $k$, allows for more precise approximation.
    The key is that these ensembles $\mathcal{E}$ can have much lower complexity  than the full $U(d)$. For instance, $\mathcal{E}$ can be discrete, such as the Clifford group which forms a $3$-design~\cite{webb2016clifford,kueng2015qubit}, and
    $k$-designs can be approximated with polynomially (in the number of qubits and in $k$) deep quantum circuits~\cite{brandao2016local,dankert2009exact,harrow2018approximate}.

\end{subsection}

\begin{subsection}{Commutants}
Commutants are closely related to the pervasive theme of block diagonalization. 
For this discussion, we heavily reference Simon \cite{simon1996representations}. 
Recall that via Theorem \ref{thm:Weyl's unitary trick}, every finite dimensional representation of a finite/compact group $G$ is unitary, and here we will only discuss these representations.
To motivate the definition of the commutant, it is good to recall one of the most fundamental facts of quantum science: two operators commute $[A,B] = 0$ if and only if they are simultaneously diagonalizable. 
We saw earlier that complete reducibility (Definition \ref{def:complete reducibility}), which holds for any unitary representation $U$ of $G$ (Theorem \ref{thm:complete reducibility of unitary representations}), amounts to a form of simultaneous block diagonalization of the set of representatives $\{U_g \vert g\in G\}$, where blocks correspond to irreducible invariant subspaces.
\textit{As such, the commutant of a representation will necessarily respect this invariant subspace decomposition.}\footnote{The more precise version of this is that the commutant of a unitary representation $U$ decomposes into minimal central projections, and these are in one-to-one correspondence with irreducible representations of $U$. 
This is reflected in Theorem \ref{thm:commutant structure}.}

Note for the following definition that unitary representations $U$ of $G$ generate algebras $\mathcal{A}$ by allowing linear combinations of the representatives $\{U_g : g\in G\}$ (this is often called a \textit{group algebra}).
\begin{definition}[Commutant]
Let $\mathcal{A}$ be a matrix algebra. Its commutant is $\mathcal{C}(\mathcal{A}):=\{B : [A,B]=0, \forall A \in \mathcal{A}\}$.
\end{definition}

We first establish a connection to a special case of the twirling map, $\TC_G^{(k)}:V\to V^G$ of Eq.~\eqref{eqn:unitary k design twirling} which acts on operators $O\in V = \mathcal{B}(\mathcal{H}^{\otimes k})$.
This connection is commonly implicitly leveraged using Weingarten calculus to analytically perform Haar integration over unitary groups in a variety of applications~\cite{nguyen2022atheory}.
Unraveling the definition of $V^G$, we have that $O\in V^G$ if $(U_g)^{\otimes k} O (U_g^\dagger)^{\otimes k} = O$, or equivalently, when
\[
    [U_g^{\otimes k}, O] = 0.
\] So here, the image $V^G$ of the twirling operator $\TC_G^{(k)}$ coincides with the commutant of the representation $U_g^{\otimes k}$.  In other words, \textit{twirled operators belong to the commutant of the group}. 
It can be shown via Schur-Weyl duality (Section~\ref{sec:Schur-Weyl Duality}) that the commutant of the $G=U(d)$ tensor representation $U_g^{\otimes k}$ is given by the index permutation representation of the symmetric group $S_k$ (see Example~\ref{box:Permutation_rep_of_S3_on_3_qubits}): this is exactly why when performing Weingarten calculus for the unitary group, $\TC_G^{(k)}$ projects onto operators that permute qubits, like $\idty$ and $SWAP_{i,j}$.
This is also why Weingarten calculus can be performed (albeit in more complicated ways) for other twirl operators with representations of other classical groups like $O(d)$.

The commutant of a representation is intimately connected to the representation's decomposition into irreducibles: in some sense, \textit{the commutant can be thought of as the operators that respect the block-diagonalization into irreducibles. } A generalized version of the ``two operators commute if and only if they are simultaneously diagonalizable'', if you will.
We will first provide a theorem formalizing this insight, and then exemplify it with a classic example of a representation and its commutant in Schur-Weyl duality.\footnote{It would be irresponsible to state this theorem without saying Schur's lemma is crucial to the proof: it guarantees that inequivalent representations do not ``mix''.}
Note that the expression $\idty_{m_k}\otimes U_k$ compactly encodes ``the irrep $U_k$ with multiplicity $m_k$'', for instance  $\begin{pmatrix}1& 0\\ 0 & 1\end{pmatrix}\otimes U=\begin{pmatrix}U& 0\\ 0 & U\end{pmatrix}$, means that $U$ appears two times, a statement which  which is reflected in the matrix forms of the following families of operators. 
\begin{theorem}[Commutant structure \cite{simon1996representations}]
Let $U$ be a unitary representation of a compact group $G$ on the Hilbert space $\mathcal{H}$ and its decomposition into irreps be 
\begin{equation}
    U = \bigoplus_{k=1}^{K} \idty_{m_k} \otimes U_k , \nonumber
\end{equation} 
where $m_k$is known as the multiplicity of the irrep $U_k$. 
Then the elements of its commutant are of the following form
\begin{equation}
    \mathcal{C}(U)= \bigoplus_{k=1}^{K} \mathcal{B}_{m_k} \otimes \idty_{\dim(U_k)},
\label{eq:comm-structure}
\end{equation}
where $\mathcal{B}_{m_k}$ denotes bounded operators in a $m_k$-dimensional Hilbert space.
\label{thm:commutant structure}
\end{theorem}
To connect this with our earlier block-diagonalization statement while discussing complete reducibility (Definition \ref{def:complete reducibility}), we observe that we are saying the matrix representatives for $U$ and the elements of the commutant can be simultaneously block diagonalized as
\small
\begin{widetext}
\begin{equation}\label{eq:block-uni}
U = \begin{pmatrix}
\, \umatxone &  &  &  & {\makebox(0,0){\text{\huge0}}} &   \\ 
 &  \ddots &  &   &   \\ 
 &  &  \underbrace{\umatxk}_{m_k \cdot \dim(U_k)} &  &   \\ 
 &  &  &  \ddots&  \\ 
  {\makebox(0,0){\text{\huge0}}}&  &  &  & \umatxK \\ 
\end{pmatrix}\,.
\end{equation}
\begin{equation}\label{eq:block-comm}
C(U) = \begin{pmatrix}
\, \bmatxone &  &  &  & {\makebox(0,0){\text{\huge0}}} &   \\ 
 &  \ddots &  &   &   \\ 
 &  &  \underbrace{\bmatxk}_{m_k\cdot\dim(U_k)} &  &   \\ 
 &  &  &  \ddots&  \\ 
  {\makebox(0,0){\text{\huge0}}}&  &  &  & \bmatxK \\ 
\end{pmatrix}\,.
\end{equation}
\end{widetext}
\normalsize
Morally, \textit{this means that the elements of the commutant can really only act upon irreps among multiplicities.}
The proof of this theorem rests fundamentally upon Schur's lemma (Theorem~\ref{lemma:Schur's lemma}), which ensures that equivariant maps between inequivalent irreps are zero maps, and equivariant maps between equivalent reps are scalar multiples of~$\idty$.

\end{subsection}

Let us now show how the block diagonal structure of the problem can significantly help us understand the way in which GQML models with equivariant quantum neural networks and  measurements process information. In particular, let us consider the schematic diagram of Fig.~\ref{fig:fig10}. We analyze a QML task where the dataset is $\SC=\{\rho_i,y_i\}_{i=1}^N$, the GQML model is of the form $h_{\vec{\theta}}(\rho_i)=\Tr[W(\vec{\theta})\rho_i W\ad(\vec{\theta}) M_i]$ as in Eq.~\eqref{eq:QML-model-qnn}, and the symmetry group that leaves the labels invariant is $G$. We know that any representation $U_g$ will admit some block diagonal structure as in Eqs.~\eqref{eq:block-uni}, and that both the equivariant quantum neural network $W(\vec{\theta})$ and the measurement operator $M$ will also have a block-diagonal structure as in Eq.~\eqref{eq:block-comm} (as they both belong to the commutant, see Eqs.~\eqref{eq:equiv} and~\eqref{eq:inv}). These block structures are schematically pictured in Fig.~\ref{fig:fig10}. 

\begin{figure*}[t]
    \centering
    \includegraphics[width=1\linewidth]{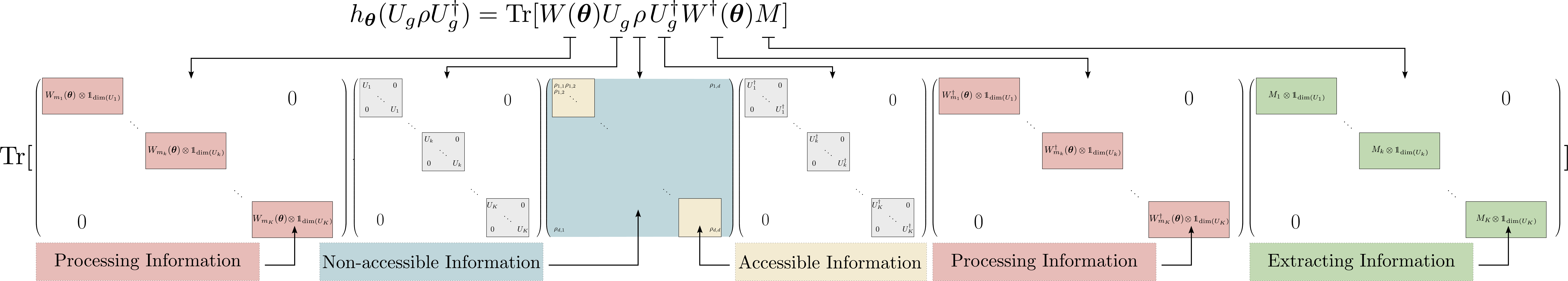}
    \caption{\textbf{QML and block-diagonal structure.} Here we consider a GQML model of the form $h_{\vec{\theta}}(U_g \rho_iU_g\ad)=\Tr[W(\vec{\theta})U_g\rho_iU_g\ad W\ad(\vec{\theta}) M_i]$, where the representation of the elements $g\in G$ have a block-diagonal structure as in Eqs.~\eqref{eq:block-uni}, while the equivariant neural network $W(\vec{\theta})$ and the equivariant measurement operator $M$ have a block-diagonal structure as in Eq.~\eqref{eq:block-comm}. When the states $\rho_i$ are not invariant under the action of $G$, the model can only access the information of $\rho_i$ encoded in the sub-blocks of the state associated to the irreps. On the other hand, there is information in the off-diagonal elements of $\rho_i$ that the QML model cannot access.
}
    \label{fig:fig10}
\end{figure*}

Then, we note that, as previously mentioned, the states need not be preserved under the action of a representation $G$. In this case, $U_g\rho_i U_g\ad\neq\rho_i$, and hence, $\rho_i$ cannot be simultaneously diagonalized as in Eqs.~\eqref{eq:block-uni} and~\eqref{eq:block-comm}. That is, $\rho_i$ does not necessarily admit a block diagonal structure in the same bases where $U_g$, $W(\vec{\theta})$, and $M$ are block diagonal. As shown in as Fig.~\ref{fig:fig10}, the previous implies that there is information in the quantum state $\rho_i$ that the model $h_{\vec{\theta}}(\rho_i)$ cannot access (i.e., the off-diagonal information outside of the sub-blocks of $\rho_i$ associated to the different irreps). Hence, the QML model can only process, extract, and learn from the information associated to the irrep sub-blocks. Crucially, if one wants to retrieve information outside of those sub-blocks (like we did in Fig.~\ref{fig:fig7}) one needs to \textit{change the representation}. This can be achieved, for instance, by allowing the QML model to act on multiple copies of each $\rho_i$ (i.e., going to a larger tensor representation of $G$) or in some cases by tracing out qubits (i.e., going to a smaller representation of $G$). We refer the reader to Ref.~\cite{nguyen2022atheory} for additional details on how GQML models can process information in the irreps. We mention in closing that this theme of block diagonalization imposing constraints on information processing appears in the study of quantum noise channels in the form of the ``noise commutant''~\cite{junge2005universal}, and more generally that commutants are a central object of study across the theory of operator algebras.

Taken together, the results in Fig.~\ref{fig:fig10} perfectly showcase how a fundamental understanding on the representation of the symmetry group allows us to better understand how information is being processed in GQML models, and as such, how we can construct better models to solve a given task.

\begin{subsection}{Schur-Weyl duality}\label{sec:Schur-Weyl Duality}
Here, we give a concrete (and very useful) example of the interplay of a representation and its commutant described in the commutant structure of Theorem~\ref{thm:commutant structure}: Schur-Weyl duality. 
There are several forms of Schur-Weyl duality, but we will restrict our attention to the classic case between tensor representations of the unitary group $U(d)$\footnote{The representation theory over $\Cbb$ of $U(d)$ and $SU(d)$ differ only by a phase, which is irrelevant for complex representations.} and index permuting representations of the symmetric group $S_n$.

Consider the space $\mathcal{H}=(\Cbb^d)^{\otimes n}$. 
A $g\in U(d)$ acts on $\mathcal{H}$ via the tensor representation $U$ (see Example~\ref{box:tensor rep of group SU(2)})
\[
    U_g \cdot (v_1\otimes \dots \otimes v_n) := U_g v_1 \otimes \dots \otimes U_g v_n ,
\] and a $\pi\in S_n$ acts on $\mathcal{H}$ via the index permutation representation $P$ (see Example~\ref{box:Permutation_rep_of_S3_on_3_qubits})
\[
    P_\pi \cdot (v_1\otimes \dots \otimes v_n) := v_{\pi^{-1}(1)} \otimes \dots \otimes v_{\pi^{-1}(n)} .
\] 
The key to notice is that these representations are mutual commutants of each other.
When this happens, Schur-Weyl duality states that the decomposition of $\mathcal{H}$ into $U(d)$ irreducibles $U_k$ immediately determines the decomposition of $\mathcal{H}$ into $S_n$ irreducibles, and vice versa.
In other words, if there are $K$ distinct irreducible invariant subspaces $U_k$ of $\mathcal{H}$, then there is a one-to-one correspondence to $S_n$ irreducible invariant subspaces $Q:\{U_k\}\to \{S_n \textit{ irreps}\}$ and we can write
\begin{equation}\begin{split} \label{eqn:schur-weyl duality}
    \mathcal{H} = \Cbb^d\otimes \dots \otimes  \Cbb^d &\cong \sum_{k=1}^K Q(U_k)\otimes U_k.
\end{split}\end{equation}
To illustrate this, if we set $n=2$, then (as we have seen similarly in Example~\ref{box:tensor rep of group SU(2)}), as representations of $U(2)$,
\[
    \mathcal{H} = \Cbb^2 \otimes \Cbb^2 \cong \text{Sym}^2(\Cbb^2) \oplus \text{Alt}^2(\Cbb^2).
\] We have seen several times by now that in the basis $\{\ket{11},\ket{01}+\ket{10},\ket{00}, \ket{10}-\ket{01}\}$, this is a block diagonalization statement. 
But let us think about these as representations of $S_2$. 
Observe that on $\text{Sym}^2(\Cbb^2)$, the index permutation representation $P$ of $S_2 = \{\idty, \text{SWAP}\}$ acts as the trivial representation $\texttt{1}$, since $P_{\text{SWAP}} = \idty$:
\begin{align*}
    P_{\text{SWAP}} \ket{11} &= \ket{11} \\
    P_{\text{SWAP}} \paran{\ket{10}+\ket{01}} &= \ket{10}+\ket{01} \\
    P_{\text{SWAP}} \ket{00} &= \ket{00}.
\end{align*} But on $\text{Alt}^2(\Cbb^2)$, the rep $P$ of $S_2$ acts as the sign representation $\texttt{sgn}$, since $P_\idty = \idty$ and $P_{\text{SWAP}} = -\idty$:
\[
    P_{\text{SWAP}}\paran{\ket{10}-\ket{01}} = -\paran{\ket{10}-\ket{01}}.
\]  In other words, the representation on $\mathcal{H}$ of the group $S_2\times U(2)$ given by $(\pi,g)\mapsto P_\pi (U_g\otimes U_g) $ decomposes into
\[
    \mathcal{H} = \Cbb^2 \otimes \Cbb^2 \cong \paran{\texttt{1}\otimes \text{Sym}^2(\Cbb^2)} \oplus \paran{\texttt{sgn}\otimes \text{Alt}^2(\Cbb^2)}.
\] Here, the connection to the expression in the commutant structure, Theorem \ref{thm:commutant structure}, becomes very explicit: we have that representatives $U_g\otimes U_g\in U(2)$ can be written as
\[
    U_g\otimes U_g = \begin{pmatrix} 
    & & & 0 \\
    & (U_g\otimes U_g)_1 & & 0 \\
    & & & 0 \\
    0 & 0 & 0 & (U_g\otimes U_g)_2
    \end{pmatrix}, 
\] while elements in the commutant $C(U)$, the permutation representatives $P_\idty,P_{\text{SWAP}}$, can be written as $P_\idty = \idty_4$ and 
\[
    P_{\text{SWAP}} = \begin{pmatrix} 
    & & & 0 \\
    & \idty_{3} & & 0 \\
    & & & 0 \\
    0 & 0 & 0 & -\idty_{1}
    \end{pmatrix} .
\] This exactly matches the block diagonal form we gave earlier!
\end{subsection}

\end{subsection}
\begin{subsection}{Transforms, convolutions, and equivariant nets}
Having spent the last few sections awash in block-diagonalization rhetoric, you may have a nagging thought in the back of your head: how do you actually perform this block-diagonalization in practice?
Like for any diagonalization procedure, the answer in general boils down to a series of linear algebra computations.
But by now it is apparent that not all representations are created equal: some representations have earned special privilege and their own transforms for block-diagonalization. 
The quantum Schur transform and the group Fourier transform are exactly examples of this.
In both cases, these transforms are well exposited in the quantum literature: we recommend Harrow's thesis ~\cite{harrow2005applications} for the former and Childs' and van Dam's review~\cite{childs2010quantum} for the latter.

\begin{subsubsection}{Quantum Schur transform}
The quantum Schur transform implements a change of basis from the computational basis to a basis block-diagonalizing the representations of $U(d)^{\otimes n}$ and $S_n$ described in the Schur-Weyl duality section. 
There exists several approaches to implement the Schur transform in a quantum circuit (see for instance ~\cite{bacon2007quantum,kirby2017practical,krovi2019efficient}), and it has in turn given rise to a wide variety of applications across quantum algorithms and quantum information.
This includes quantum teleportation~\cite{studzinski2017port,christandl2021asymptotic}, quantum compression~\cite{knill1997theory,kempe2001theory,yang2016efficient}, distortion-free entanglement concentration~\cite{matsumoto2007universal,beckey2021computable}, estimating the spectrum of density operators~\cite{keyl2005estimating}, and a litany of others ~\cite{harrow2005applications}.
Notably, it has recently breached into the world of QML~\cite{zheng2022super}, where $SU(d)$-equivariant quantum convolutional neural network training algorithms were shown to enjoy super-exponential speedup over classical algorithms. 
Central to the result is block diagonalization via the Schur transform, combined with an $SU(d)$-equivariant architecture. Similarly, it has been recently shown that $S_n$-equivariant quantum neural networks exhibit several favorable properties such as not exhibiting barren plateaus, being able to efficiently reach the overparametrization regime~\cite{larocca2021theory}, and being able to generalize from few training data points~\cite{schatzki2022theoretical}. 
Here, the power of working on irreps instead of the full space is evident.
\end{subsubsection}

\begin{subsubsection}{Group Fourier transform and convolutions}
Of course, there are many other interesting representations. 
Perhaps one of the most important is the so-called ``left regular representation'', which contains \textit{every} irreducible representation as an invariant subspace.
In some sense, this allows us to study all of the representations of a group simultaneously. 
Much as the Schur transform block-diagonalized the representations $U(d)^{\otimes n}$, the group Fourier transform will block-diagonalize the left regular representation.
Indeed, the usual Fourier transform is effectively a special case\footnote{This is slightly misleading, as the theory is much richer when the groups at play are not compact.
See \cite{folland2016course}.} of this, where the group is the additive group of real numbers $(\R,+)$ acting on functions $L^2(\R)$ by translation $g\cdot f(x) = f(x-g)$.
For ease of discussion, we restrict to the case of a finite group $G$, mentioning that via the Peter-Weyl theorem, this can be readily extended to compact Lie groups $G$ acting appropriately upon infinite dimensional Hilbert spaces~\cite{folland2016course}.

\begin{definition}[Left regular representation]
Let $V$ be a complex vector space of dimension $\abs{G}$ with orthonormal basis $\{\ket{g} : g\in G\}$ labeled by elements of $G$.
Then the \emph{left regular representation} $L$ of $G$ is given by 
\[
    L_h\ket{g} = \ket{hg}.
\]
\end{definition}

We need to note that there is another notion of the left regular representation, which effectively subsumes the above definition: if $V = L^2(G)$, the square-integrable functions on $G$ (which recall is a manifold as well as a group), then the left regular representation $L$ acts on functions $f\in V$ by
\begin{equation} \label{def:left regular rep, function version}
    L_g \cdot f(x) = f(g^{-1} x) , \qquad x\in G.
\end{equation} This should be somewhat reminiscent of the permutation representation (Example~\ref{box:Permutation_rep_of_S3_on_3_qubits}).
Indeed, when $G$ a compact Lie group, this is the representation which contains every finite dimensional irrep via the Peter-Weyl theorem~\cite{folland2016course}.

As promised, the block diagonalization result.
\begin{proposition}(Ref.~\cite{childs2010quantum})
Let $\hat{G}$ denote the set of irreps of $G$.
Then as representations,
\[
    L_g \cong \bigoplus_{\rho\in \hat{G}} \rho \otimes \idty_{\dim(\rho)}.
\] Moreover, $L$ is block-diagonalized by the group Fourier transform $F_G$:
\[
    L_g = F_G^\dagger \paran{\bigoplus_{\rho\in \hat{G}} \rho \otimes \idty_{\dim(\rho)}} F_G ,
\] where $F_G$ is a unitary matrix.
\end{proposition}
A great many authors, in and out of the quantum community, have written well about the group Fourier transform---see e.g. ~\cite{childs2010quantum} for a quantum perspective and ~\cite{folland2016course} for a harmonic analysis perspective---so we will not say much about its intriguing properties.
Perhaps most notably, the group Fourier transform plays a central role in quantum algorithms attacking hidden subgroup problems, which include as famous special cases prime number factorization and discrete logarithm~\cite{shor1994algorithms,nielsen2000quantum}.
For a review of these and other fundamental applications to quantum algorithms, see~\cite{childs2010quantum}.

For QML, the power lies in that when we have \textit{any} irreducible representation, it can be realized as an invariant subspace of the regular representation (or several copies thereof), and when this is realized, the group Fourier transform block diagonalizes it.
Encouragingly, the regular representation is surprisingly common within classical machine learning, thanks to a privileged relationship with \textit{group convolution}.\footnote{The connection is rich and beyond our scope, but one key to all of it is that for functions $u,v\in L^2(G)$, the left regular representative $L$ commutes with group convolution: $L_g (u\star v) = (L_g u )\star v$ for all $g\in G$.} And beyond being a crucial primitive, one can show that every classical feedforward neural network is equivariant to a compact $G$ (meaning the layers of the network are equivariant maps in the sense of Definition \ref{def:equivariance}) if and only if it is a $G$-convolutional neural network (Theorem 1~\cite{kondor2018generalization}).
This correspondence between equivariant layers and convolutions is remarkable, given the dramatic success of convolutional neural networks in classical machine learning.
The intimate connections between the group Fourier transform, group convolutions, the regular representation, and equivariance is further explored in the contexts of quantum algorithms in~\cite{castelazo2021quantum} and QML in~\cite{nguyen2022atheory,castelazo2021quantum}.

\end{subsubsection}
\end{subsection}
\end{section}

\section{Symmetries in the wild}\label{sec:symmetries in the wild}

Quite a bit of abstract theory has been presented thus far.
We hope that the examples provided earlier and in Boxes~\ref{box:1} and~\ref{box:2} leave the reader with some intuition for first noticing a symmetry, identifying the group/algebra representation that formalizes the symmetry, and applying the powerful tools of representation theory to decompose it into irreducible representations which can then be further analyzed.
In practice, the first step is often the hardest and requires genuine insight into the invariants of the problem at hand: finding and pinning down the symmetry generally requires physical information, geometric information, computational patterns, combinatorial witchcraft, or divine intervention.
It also typically requires some familiarity with ``usual suspect'' symmetry groups, and so we have included a (highly noncomprehensive) list of groups that just have a way of sneaking all over the place in Boxes~\ref{box:3} and~\ref{box:4}.
In this section, we briefly describe a few strategies for finding and capturing wild symmetries which have enjoyed a variety of successes across disciplines.
The distinctions between these types of ``symmetry traps'' is somewhat artificial, and in reality there is a high degree of overlap among them.

\begin{subsection}{Tripwire nets: physical invariants}
The history of symmetries by way of physical invariants and conserved quantities is a rich one. To understate it, it is useful to start at the beginning: at its inception Lie theory was constructed to study symmetries of differential equations, many of which arose from physical models. 
One of the most celebrated instances of this appears while studying radial potentials in both classical and quantum mechanics: in some sense, the ``radial'' assumption of a potential $f:\R^3\to \R$, meaning the 3D coordinates depend only on the distance to a fixed origin $f(\mathbf{x}) = f(\abs{\mathbf{x}})$, can be rephrased by saying that $f$ is invariant under the 3D rotation group $SO(3)$, where the representation is on the (infinite dimensional) space $V = L^2(\R^3)$ and is given by $g\cdot f(\mathbf{x}) = f(g^{-1} \mathbf{x})$.\footnote{This shows up quite prominently in the analysis of the Hydrogen atom: for a treatment which is more Lie-theoretic in flavor, see Hall's other book \cite{hall2013quantum}.}

Let us sketch a common (but not universal) program for identifying these physical invariants when one suspects a given symmetry group $G$ is at work.
We will fundamentally use the following proposition and of course, the correspondence between Lie group and Lie algebra representations (Theorems~\ref{thm:lie group reps yield lie alg reps} and~\ref{thm:lie alg reps lift to simple lie group reps}).
First, one takes a Hamiltonian for a system $H$ and considers a putative unitary representation $U$ of a symmetry group $G$ on the Hilbert space of states $V=\mathcal{H}$. 
This immediately induces a natural representation on operators on this Hilbert space, as discussed underneath Definition~\ref{def:dual representation}.
So, we can check whether the Hamiltonian is invariant under the symmetry either at the level of the Lie group representation $U$ or the Lie algebra representation $u$: so, checking either $U_g H U_g^\dagger = H$ for all $g\in G$, or (taking derivatives as in Example~\ref{box:adjoint rep of alg su(2)}) $[u(X), H] = 0$ for all $X\in \liea$. 
Commonly it is easier to check the latter condition, because $\liea$ is a vector space and so we need only check on a basis (unlike the manifold $G$).
Once this is accomplished, we apply the following proposition:

\begin{proposition}\label{prop:symmetries, Hamiltonians, and eigenstates}
Let $U:G\to GL(\mathcal{H})$ be a representation of a group $G$, and let $H$ be a Hermitian operator such that $[U_g,H] = 0$ for all $g\in G$. 
Then, for any eigenvector $\ket{\psi}$ of $H$ with eigenvalue $\lambda$, $U_g\ket{\psi}$ is also an eigenvector of $H$ of eigenvalue $\lambda$.
\end{proposition}\begin{proof}
Observe that $H U_g \ket{\psi} = U_g H\ket{\psi} = \lambda U_{g}\ket{\psi}.$
\end{proof}
\begin{observation}
\emph{Warning: }As discussed earlier, this does not mean that $\ket{\psi} = U_g\ket{\psi}$! 
An illustrative counterexample is to take $H = \idty$ on $\mathbb{C}^2$ and to consider the representation $U:\mathbb{Z}_2\to GL(\mathbb{C}^2)$ given by $U_1 = \idty$ and $U_\sigma = X$. 
Then certainly $\ket{0}, \ket{1}$ are eigenvectors of $H$, but $U_\sigma \ket{0} = \ket{1}$.
In general, this means that representations can (but do not have to) permute within eigenspaces of $H$, but they cannot map between eigenspaces.
In other words, since $[U_g,H]=0$, $H$ is simultaneously block diagonalized with $U_g$, and distinct eigenspaces necessarily inhabit different blocks.
\end{observation}

So, what does this do for us?
This means that the eigenspaces of $H$ (often interpreted as energy eigenspaces) are invariant subspaces of the representation $U$, and we can use the powerful methods of representation theory to further analyze them.
Take a look at two examples, one discrete and one continuous, of this approach in Box~\ref{box:1} and Box~\ref{box:2}. 
In the former, we see that translation invariant Hamiltonians have energy eigenspaces which are invariant under translations. 
Heeding the warning, this does not necessarily mean that every eigenstate is translation invariant: for instance, in a process known as dimerization~\cite{boette2016pair}, the Majumdar-Ghosh model~\cite{majumdar1969next} has a translation-invariant Hamiltonian, but its 2 ground states are only 2-periodic. 
In this case, translation by 1 site maps between the two ground states (and so of course preserves the energy).
To translate to the language of representation theory, we would say the ground state space forms an irreducible representation for this lattice translation symmetry.

While a simple case, the theme is quite deep: when an eigenspace forms an irreducible representation, we need only find a single eigenvector---the rest can be generated by acting on this vector with the representation (in the above example, acting via the translation operator).
Alternatively, when the eigenspace forms a reducible representation, restricting to an invariant subspace allows us to continue extracting structure via further decomposition.
As a special case of this lemma, note that if one possesses a nondegenerate eigenstate, i.e. one whose eigenspace is one-dimensional, that state automatically inherits the symmetries of its Hamiltonian. 
\end{subsection}

\begin{subsection}{Deadfalls: geometric invariants}
Representations have close ties to geometry---indeed, many of the physical symmetries we described a moment ago are highly geometric in nature. 
An obvious application to QML is when the data itself is invariant under the symmetry group~\cite{larocca2022group,verdon2019quantumgraph,sauvage2022building}.
But often, geometric invariants have a way of hiding in plain sight: in fact, even the definitions of the classical Lie groups can be phrased as ``the transformations which leave a tensor invariant''.
Just as in the case of the inner product, other tensors (like bilinear forms and the determinant) carry rich geometric information, and many of the labels we consider in quantum machine learning are ultimately tensorial in nature.
For example:
\begin{itemize}
    \item The unitary group $U(d)$ consists of linear transformations $U$ are those which leave the Hermitian inner product $\langle \cdot, \cdot \rangle$ invariant:
    \[
        \langle Uv,Uw \rangle = \langle v,w \rangle.
    \]
    \item The special linear group $SL(d)$ consists of linear transformations which preserve volumes, i.e. invertible maps with $\det = 1$.
    \item The orthogonal group $O(d; \R)$ can be defined as the set of transformations $O$ preserving a symmetric positive definite bilinear form $Q(v,w) = v_1w_1 + v_2w_2 + \dots + v_dw_d$:
    \[
        Q(Ov,Ow) = Q(v,w) .
    \] By adjusting the bilinear form, we can get other interesting groups: for instance, the Lorentz group $O(1,3)$, which is the symmetry group of Minkowski spacetime in relativity, is the set of linear transformations $\Lambda$ preserving the bilinear form
    $Q(v,w) = v_t w_t - v_xw_x - v_yw_y - v_zw_z$
    \[
        Q(\Lambda v, \Lambda w) = Q(v,w) .
    \]
\end{itemize}
Despite their relative simplicity, these symmetry groups already appear in a wide range of QML problems~\cite{larocca2022group} such as classifying datasets based on purity~\cite{garcia2013swap,cincio2018learning,huang2021quantum}, time-reversal dynamics~\cite{sachs1987physics,huang2021quantum,aharonov2022quantum,chen2021exponential},
multipartite entanglement~\cite{schatzki2021entangled,beckey2021computable}, and graph isomorphism~\cite{verdon2019quantumgraph}. We could continue \textit{ad infinitum}, but the takeaway is to pay close attention to tensors which track geometric information: they may have several naturally associated symmetry groups, and if you are lucky, your data labels may respect their representations.

In the Review~\cite{jaeger2005entanglement}, Jaeger details more concrete connections between geometry and quantum information: geometric invariants under local unitary transformations (LUTs) and stochastic local operations and classical communication (SLOCCs) provide rich insight into a variety of entanglement and mixedness measures in multi-qubit systems, with special attention given to invariants of the Lorentz group $O(1,3)$.
\end{subsection}

\begin{subsection}{Snares: algebraic invariants}
Of course, the geometric tensorial invariants we described could be recast as algebraic invariants---they are, after all, given by polynomials and equations. 
But let us go another direction and think about recasting some well-known identities within quantum information in a more symmetry-driven light (we have done this already in our examples in Box~\ref{box:2}, but why not a few more):
\begin{itemize}
    \item The spectrum of an operator/mixed state is invariant under change of basis, i.e. under the symmetry group $GL(d)$.
    We use this fact constantly, for instance using the identity $\Tr[A] = \Tr[gAg^{-1}]$.
    A case could be made that this is the most fundamental symmetry in quantum computing.
    \item Two of the most important rank $n$ tensor invariants are full symmetry and antisymmetry under index permutations, aka the trivial and sign representations of $S_n$, respectively.
    These representations are ubiquitous in quantum information and computation: for starters, we recommend taking a look at Harrow's~\cite{harrow2013church}. 
    As a particularly important example, we note that the determinant (or the Levi-Cevita tensor) is antisymmetric. Leveraging the symmetric subspace, one can find optimal channels for state estimation/cloning and prove a quantum de Finetti theorem~\cite{harrow2013church}.
    \item Rank $n$ tensor invariants arise as representations of the symmetric group $S_n$ (and by Schur-Weyl duality, as representations of $U(d)^{\otimes n}$ acting on all indices). 
    But the symmetric and antisymmetric representations are just the tip of the iceberg: tensor symmetries play a central role in the analysis of tensor networks~\cite{biamonte2019lectures}, including those of matrix product states (MPS) like the AKLT chain and more generally projected entangled pair states (PEPS)~\cite{affleck1988valence, cirac2020matrix}. 
    It should not be surprising that representation theory plays a critical role in the classification of symmetry-protected topological phases of matter. 

    \item A large class of optimization problems in quantum information can be phrased as semidefinite programs with a $U^{\otimes p}\otimes (U^*)^{\otimes q}$ symmetry. Leveraging this symmetry allows one to prove significantly better time scaling than standard semidefinite programming guarantees~\cite{grinko2022linear}.
    \item Representations of the symmetric group in turn are closely tied to symmetric polynomials (aka Schur polynomials). While it is not obvious to us how these will appear in QML contexts, it is still worth keeping them in mind, as their ubiquity in representation theory makes use cases feel inevitable.
    
    \item Classifying states under so-called "stochastic local operations and classical communication" (SLOCC) is a major goal of entanglement theory. Two states are equivalent under SLOCC if one can convert between them with some non-zero probability. In fact, these classes can be formalized through the orbits of the special linear group $SL(d)$. One can classify and measure entanglement through SL-invariant polynomials (SLIPs)~\cite{dur2000three,verstraete2003normal,wootters1998entanglement,coffman2000distributed,leifer2004measuring,osterloh2005constructing}.
\end{itemize}
\end{subsection}

\begin{subsection}{Variable rate loans: other dangerous traps}

Say that we want to check if a given problem, defined in terms of a Hamiltonian $H$ exhibits symmetry to a given group $G$. This can be cast into the task of checking whether $H$ commutes with every representation $U_g$ of the elements $g$ of $G$. Despite its innocuous look, this task can quickly become intractable if $H$ is very large (for $n$-qubit systems, the dimension of $H$ grows as $2^n$). However, one can tackle this problem from the optics of quantum property testing, which quite literally studies the task of ``given a large object that either has a certain property or is somehow far from having that property, a tester should efficiently distinguish between these two cases''. We refer the reader to Ref.~\cite{montanaro2013survey} for a nice review on quantum property testing. 

Recent efforts have been put forward towards developing symmetry-testing algorithms. For instance, the work in~\cite{laborde2022quantum} proposes a quantum algorithm to check whether a Hamiltonian exhibits symmetry with respect to a group, and the work in~\cite{laborde2021testing} sets forth a variety of quantum algorithms with corresponding numerics which test symmetries of states and channels. These point towards an exciting new era of the symmetry program in physics and computer science, wherein we may algorithmically test for symmetries without requiring analytic expressions of models and data.

\end{subsection}

\section{Outlook}\label{sec:outlook}

Representation theory is one of the most powerful tools that any quantum computing or quantum information scientist can possess under their belt. For the particular context of QML, representation theory allows for the manipulation and understanding of symmetries in the data, as well as to study how physical processes can be built to respect those symmetries. Despite its tremendous importance, many key results and insights in representation theory are hidden behind mathematical and algebraic notations that may seem insurmountable for non-experts. However, we hope that after reading this article, the reader will find themselves comfortable enough with the notational and conceptual basics to dig deeper into the rich literature of representation theory (we again cannot recommend enough Refs.~\cite{hall2013lie,simon1996representations,fulton1991representation,serre1977linear}) and its fundamental importance for the future of QML.

Motivating this entire article is the simple idea that building models respecting the symmetries of a dataset should improve their performance. While this claim has been investigated and verified in the classical literature for some years now~\cite{bronstein2021geometric,cohen2016group,maron2018invariant,kondor2018clebsch,kondor2018generalization,bekkers2018roto,anderson2019cormorant,cohen2019gauge,cohen2019general,elesedy2021provably,wang2022approximately}, the same cannot be said for the quantum realm. GQML is a nascent field with many promising results~\cite{larocca2022group,skolik2022equivariant,meyer2022exploiting,zheng2021speeding,sauvage2022building},  but far more work is needed to elucidate problems wherein models with symmetry may be able to outperform symmetry-agnostic ones, possibly granting quantum advantage. Conversely, it might also be interesting to analyze if using slightly-symmetric breaking models has any benefit.  Navigation of this complex landscape of models will inevitably demand intrepid researchers armed with the fundamentals of representation theory.

Let us briefly highlight our journey. First, the GQML program is straightforward: given a task a hand, we need to \textit{identity the relevant symmetries underlying the data}. This is not always an easy task. Here, there is no better advice than ``learn the examples and trust your gut''. To this end we have presented several ``candidate'' symmetry group examples that should occupy a privileged place in one's mind, as well as standard methods that can facilitate their discovery: see Box~\ref{box:1} and Box~\ref{box:2} for a field guide on some common symmetry groups.

Once the symmetry group has been identified we can plan to build GQML models respecting these symmetries (e.g., \textit{equivariant} quantum channels, and  measurement operators). In particular, when finding the channels that respect the symmetry group it is highly advisable to follow the trick: \textit{When given a problem with Lie group symmetry, pass to the Lie algebra, analyze it, and return to the Lie group.} As we have seen, working at the Lie algebra level allows us to harness the full power of linear algebra.

The next fundamental step is understanding how the representation of the symmetry group induces a block diagonal structure in the problem. We cannot overemphasize that \textit{GQML is all about block-diagonal structures}. Understanding the decomposition into irreps of the group representation allows us to identify the accessible and processable information of a quantum state for an equivariant quantum neural network or an equivariant measurement operator. Crucially, here we can leverage the full power of changes of representation (e.g. acting on multiple copies of the data, or tracing out qubits), as these change the block-diagonal structure and concomitantly, the information that the model can ``see''. This realization paves the way towards a more detailed understanding of how information gets embedded into the different irreps and how one can use this to solve a given QML task.

We hope that this work will serve the reader as a starting point towards the exciting world of representation theory. But be warned, as pursuing representation theory is a dangerous business. If one does not keep their feet, there is no knowing to where one may be swept off.

\section*{Acknowledgements}

We thank Bruno Nachtergaele, Eugene Gorsky, Timo Eckstein, and Laura Gentini for helpful discussions and comments on our work.
This work was partly supported by the U.S. Department of Energy (DOE) through a quantum computing program sponsored by the Los Alamos National Laboratory (LANL) Information Science \& Technology Institute.
M.R. was partially supported  by the National Science Foundation through DMS-1813149 and DMS-2108390.
L.S. was partially supported by the NSF Quantum Leap Challenge Institute for Hybrid Quantum Architectures and Networks (NSF Award 2016136).
P.J.C. and M.L. were initially supported by the U.S. DOE, Office of Science, Office of Advanced Scientific Computing Research, under the Accelerated Research in Quantum Computing (ARQC) program. P.J.C. and L.S. were also supported by the LANL ASC Beyond Moore's Law project. 
F.S. was supported by the Laboratory Directed Research and Development (LDRD) program of LANL under project number 20220745ER. 
M.L. was also supported by the Center for Nonlinear Studies at LANL. 
M.C. acknowledges support by the LDRD program of LANL under project numbers 20210116DR and 20230049DR.

\bibliography{quantum.bib}

\clearpage
\newpage

\begin{pabox}[float*, floatplacement=h!t,label={box:1},width=2\columnwidth]{Discrete symmetries}
\textbf{Bit parity with bit-flip symmetry:} Let $n$ be an even number of qubits. 
Consider a problem of classifying computational basis product states according to their parity. Here, $\rho_i=\dya{\psi_i}$, with  $\ket{\psi_i} = \ket{z_{i_1}z_{i_2}\dots z_{i_n}}$, and where $z_{i_k}\in \{0,1\}$ the parity of $\ket{\psi_i}$ is defined as $y_i = f(\ket{\psi_i}) = \sum_{k=1}^n z_{i_k} \mod 2$.
Defining the spin-flip operator $P=\bigotimes_{j=1}^n X_j$, where $X_j$ is the Pauli-$x$ operator acting on the $j$-th qubit.
One can readily see that while $P\ket{\psi_i}\neq \ket{\psi_i}$, the parity is invariant under $P$, i.e. $f(P\ket{\psi_i}) = f(\ket{\psi_i})$. 
For a concrete example, $f(P\ket{01}) = 1 = f(\ket{10})$.
In other words, the states are not invariant under the symmetry, but the labels are.
\begin{itemize}
    \item \textit{States:} Bitstring product states $\ket{\psi_i} = \ket{z_{i_1}z_{i_2}\dots z_{i_n}} \in (\mathbb{C}^{2})^{\otimes n}$, where $z_{i_k}\in \{0,1\}$ and $n$ even
    \item \textit{Labels:} Parity $y_i = f(\ket{\psi_i}) = \sum_{k=1}^n z_{i_k} \mod 2$
    \item \textit{Group:} $\mathbb{Z}_2 = \{1, p\}$
    \item \textit{Representation: } $R:G\mapsto GL((\mathbb{C}^{2})^{\otimes n})$, where $1\cdot \ket{\psi_i} = \ket{\psi_i}$,     $\sigma\cdot \ket{\psi_i} = P\ket{\psi_i}$
\end{itemize}
\textbf{Qubit reflection parity:} Consider a problem of classifying states according to their qubit-reflection parity. Defining the qubit-reflection operator $R:=R_{1,n}R_{2,n-1}\dots R_{\lfloor{n/2}\rfloor, \lfloor{n/2}\rfloor+1}$, where $R_{j,j'}$ swaps qubits $j$ and $j'$, and writing $\rho_i=\dya{\psi_i}$, the states will have label $y_i=0$ ($y_i=1$) if $\rho_i$ is an eigenstate of $R$ with eigenvalue $1$ ($-1$). Here, one can readily verify that $R\rho_i R\ad=\rho_i$. \\
\textbf{Qubit permutations:} Learning problems with permutation symmetries abound. Examples include learning over sets of elements, modeling relations between pairs (graphs) or multiplets (hypergraphs) of entities, problems defined on grids (such as condensed matter systems), molecular systems, evaluating genuine multipartite entanglement, or working with distributed quantum sensors. Consider for instance a problem where an $n$-qubit $\rho$ is a graph state  encoding  the topology of an underlying graph. One can create such state by starting with the state $\ket{+}^{\otimes n}$, and applying a unitary $U^{(a,b)}$ for each edge $(a,b)$ in the graph. Here  $U^{(a,b)} = e^{-i\gamma ((\dya{0})^a \otimes \id^b + (\dya{1})^a\otimes Z^b)}$ is an Ising-type interaction. By conjugating
the state with an element of $S_n$, one obtains a new quantum
state whose interaction graph  is isomorphic to the original one. 
\begin{itemize}
    \item \textit{States:} Quantum states on qubits, where the qubit labeling index do not matter.
    \item \textit{Labels: (Here, any label will work, since the states themselves are invariant).}
    \item \textit{Group:} $G = S_n$, the symmetric group on $n$ letters
    \item \textit{Representation: } $R:G\mapsto GL((\mathbb{C}^{2})^{\otimes n})$, where the 2-cycle $(j,j')\cdot \ket{\psi_i} = \SWAP_{j,j'} \ket{\psi_i}$. Note that since any permutation in $S_n$ can be expressed as a product of swaps, this defines our representation on all permutations.
\end{itemize}
\textbf{Translation invariance:} Let $H$ a Hamiltonian and consider the problem of classifying energies $y_i$ of a set of eigenstates $\ket{\psi_i}$.
Suppose $H = \sum_{j=1}^{n} h_{j,j+1}$, where $h_{j,j+1}$ is a nearest-neighbor interaction and we impose periodic boundary conditions so that $n+1 \equiv 1$. 
Then $H$ commutes with the translation operator $\tau_g: (\mathbb{C}^2)^{\otimes n} \to (\mathbb{C}^2)^{\otimes n}$, which translates the state $\ket{\psi}$ to the right by $g$ sites (e.g. $\tau_0 = \idty$ and $\tau_2\ket{01101} = \ket{01011}$).
We can then use Prop. \ref{prop:symmetries, Hamiltonians, and eigenstates} to argue that the energy label $y_i$ is invariant under the group of translations, so $f(\ket{\psi_i}) = f(\tau_g\ket{\psi_i})$ for all translations $\tau_g$.
\begin{itemize}
    \item \textit{States:} Eigenstates $\ket{\psi_i}$ of a Hamiltonian $H$ on a ring of $n$ qubits
    \item \textit{Labels:} Eigenenergies $y_i$, i.e. $H\ket{\psi_i} = y_i \ket{\psi_i}$
    \item \textit{Group:} $G = \mathbb{Z}_n$, the cyclic group of order $n$.
    \item \textit{Representation: } $\tau:G\mapsto GL((\mathbb{C}^{2})^{\otimes n})$, where $\tau_g$ translates to the right by $g$ sites
\end{itemize}
\end{pabox}

\begin{pabox}[float*, floatplacement=h!t,label={box:2},width=2\columnwidth]{Continuous symmetries}
\textbf{Unitary transformations and purity:} Consider a problem of classifying pure states from mixed states. 
The dataset here is composed of states with label $y_i=0$ ($y_i \neq 0$) if $\rho_i$ is pure (mixed).  
Since the purity is a spectral property, then the labels in $\SC$ are invariant under the action of any unitary. Note that here  $f(U \rho_i U\ad)=f(\rho_i)$, but in general $U \rho_i U\ad\neq \rho_i$.
\begin{itemize}
    \item \textit{States:} States $\rho_i\in \DC(\mathcal{H})$.
    \item \textit{Labels:} Pure $y_i = 0$ and mixed $y_i\neq 0$.
    \item \textit{Group:} $G = U(d)$, the unitary group on $\mathcal{H}$. 
    \item \textit{Representation: } $U:G\mapsto U(d)$, where $g\cdot \rho_i = U_g \rho_i U_g^\dagger$.
\end{itemize}
\textbf{Orthogonal transformations:} Consider a problem of classifying orthogonal (real-valued) states from Haar-random states.  The dataset here is composed of states with label $y_i=0$ ($y_i \neq 0$) if $\rho_i$ is a real-valued state (a Haar random state).  
Here, the labels $y_i=0$ are invariant under the action of any orthogonal unitary, as conjugated a real-valued state by a real-valued unitary yields a real-valued state. Note that here  $f(U \rho_i U\ad)=f(\rho_i)$, but in general $U \rho_i U\ad\neq \rho_i$.
\begin{itemize}
    \item \textit{States:} States $\rho_i\in \DC(\mathcal{H})$.
    \item \textit{Labels:} Orthogonal $y_i = 0$ and mixed $y_i\neq 0$.
    \item \textit{Group:} $G = O(d)$, the orthogonal group on $\mathcal{H}$
    \item \textit{Representation: } $U:G\mapsto O(d)$, where $g\cdot \rho_i = U_g \rho_i U_g^\dagger$
\end{itemize}
\textbf{Local unitary transformations and the XXX model:} Consider the problem of classifying ground states of the Heisenberg $XXX$ model $H=J\sum_{j=1}^n(X_j X_{j+1}+Y_j Y_{j+1}+Z_j Z_{j+1})$. Here, $y_i=0$ ($y_i=1$) if $\rho_i$ is a ferromagnetic (antiferromagnetic) ground state of $H$ with $J<0$ ($J>0$). Since  $H$ commutes with the total magnetization operators $S_x=\sum_{j=1}^nX_j$, $S_y=\sum_{j=1}^nY_j$, $S_z=\sum_{j=1}^n Z_j$, then the labels are invariant under the action of the same local unitary acting on all qubits. That is, $f((\bigotimes_i^n U) \rho_i (\bigotimes_i^n U\ad))=f(\rho_i)$ for any local unitary $U$. 
\begin{itemize}
    \item \textit{States:} Ground states of the $XXX$ chain $\rho_i\in \DC(\mathcal{H})$
    \item \textit{Labels:} Ferromagnetic $y_i = 0$ and antiferromagnetic $y_i = 1$
    \item \textit{Group:} $G = U(2)$
    \item \textit{Representation: } $U:G\mapsto U(d)$, where $g \cdot \rho_i = (U_{g} \otimes \dots \otimes U_{g}) \rho_i (U_{g} \otimes \dots \otimes U_{g})\ad $
\end{itemize}
\textbf{Local unitary transformations and multipartite entanglement:} Consider the problem of classifying pure quantum states according to the amount of multipartite entanglement they posses.  Here, $y_i=1$ if the states posses a large amount of multipartite entanglement (according to some measure), while $y_i=0$ if the states are separable. Since local unitaries do not change the multipartite entanglement in a quantum state, then we have that $f((\bigotimes_j^n U_j) \rho_i (\bigotimes_j^n U_j\ad))=f(\rho_i)$ for any local unitary $U_j$  acting on the $j$-th qubit.
\begin{itemize}
    \item \textit{States:} Pure states $\rho_i\in \DC(\mathcal{H})$
    \item \textit{Labels:} $y_i\in [0,1]$, where 0 means separable and 1 means ``highly entangled''
    \item \textit{Group:} $G = U(2) \times \dots \times U(2)$, ($n$ times)
    \item \textit{Representation: } $U:G\mapsto U(d)$, where $(g_1,\dots,g_n)\cdot \rho_i = (U_{g_1} \otimes \dots \otimes U_{g_n}) \rho_i (U_{g_1} \otimes \dots \otimes U_{g_n})\ad $
\end{itemize}

\end{pabox}

\begin{pabox}[float*, floatplacement=h!t,label={box:3},width=2\columnwidth]{The usual suspects: commonly appearing discrete groups}
While a veritable zoo of group symmetries can (and do!) appear in the wild, any QML practicioner should be familiar with some especially common species. 
Consider this a beginner's field guide to some frequently appearing groups and a few classic applied locations they have been observed.

Note that discrete group representations often appear in disguise as subgroups of continuous groups (which is inevitable, because as we described earlier, we care about unitary representations on Hilbert spaces). For example, a $\pi/2$ pulse instantiates a rotation by $\pi/2$ on the Bloch sphere for a single qubit---the possible actions of strings of $\pi/2$ pulses generate the group of rotations $\mathbb{Z}_4$. 
\\ \\
\textbf{Discrete groups}
\begin{itemize}
    \item $\mathbb{Z}_n$, the cyclic group of integers modulo $n$. 
    \begin{itemize}
        \item \textit{Type, size, \# irreps:} Abelian, $\abs{\mathbb{Z}_n} = n$, n irreps (given by roots of unity)
        \item \textit{Where you might find them:} translations on periodic lattices, rotations of 2D regular polygons, roots of unity, parity transformations
        \item \textit{Useful fact:} Every finite abelian group is a direct product of cyclic groups, so finite abelian symmetries can often be studied by restricting to cyclic groups.
    \end{itemize}
    \item $S_n$, the symmetric group on $n$ letters (aka the group of all permutations on a set of size $n$).
        \begin{itemize}
        \item \textit{Type, size, \# irreps:}  nonabelian for $n\geq 3$, $\abs{S_n} = n!$, integer partitions $\lambda$ of $n$ 
        \begin{itemize}
            \item  A tuple of positive integers $\lambda=(\lambda_1,\dots,\lambda_k)$ is called a \textit{partition} if $\lambda_1 + \dots + \lambda_k = n$ and $\lambda_1\geq \dots \lambda_k > 0$. 
            Partitions, and thus irreps of $S_n$, are labeled by Young diagrams.
        \end{itemize}
        \item \textit{Where you     might find them:} Symmetric and antisymmetric vectors, determinants, more general permutations on tensor indices, combinatorial identities
        \item \textit{Useful fact:} Schur-Weyl duality tells us that the tensor representations $U^{\otimes n}$ of $U(d)$ on $(\mathbb{C}^d)^{\otimes n}$ can be decomposed into representations of $U(d)$ and $S_n$, where $U(d)$ acts on one $\mathbb{C}^d$ and $S_n$ permutes tensor indices.
        \end{itemize}
    \item $D_n$, the dihedral group of symmetries of the regular $n$-gon.
    \begin{itemize}
        \item \textit{Type, size, \# irreps:} nonabelian for $n\geq 3$, $\abs{D_n} = 2n$, $(n+3)/2$ if $n$ odd, $(n+6)/2$ if $n$ even
        \item \textit{Where you     might find them:} Molecular symmetry, finite subgroups of $O(2)$
        \item \textit{Useful fact:} $D_n$ is generated by rotations of angle $2\pi/n$ and reflections, and so on the complex plane is commonly thought of as the group generated by multiplication by $e^{2\pi i/n}$ and complex conjugation.
    \end{itemize}
    \item $\mathbb{Z}^n=\mathbb{Z}\times \dots \times \mathbb{Z}$, the additive group of integer vectors $(k_1,\dots,k_n)$
    \begin{itemize}
        \item \textit{Type, size:} abelian, $\abs{\mathbb{Z}^n} = \infty$
        \item \textit{Where you might find them:} Fourier series, translations on infinite lattices
        \item \textit{Useful fact:} The cyclic group $\mathbb{Z}_\ell$ can be thought of a quotient group of the group of integers $\mathbb{Z}$, or more geometrically, $\mathbb{Z}_\ell$ is the group of translations on a ring with $\ell$ sites (``periodic boundary conditions''). Likewise, if we quotient every $\mathbb{Z}$ in the translation group $\mathbb{Z}^n$, we get tori $\mathbb{Z}_{\ell_1}\times \mathbb{Z}_{\ell_2} \times \dots \times \mathbb{Z}_{\ell_n}$.
    \end{itemize}   
\end{itemize}

\end{pabox} 

\begin{pabox}[float*, floatplacement=h!t,label={box:4},width=2\columnwidth]{The usual suspects: commonly appearing continuous groups}
The field guide continues with some commonly appearing continuous groups. 
These species hold privileged positions in physics as symmetries enjoyed by a variety of differential equations, and they are often detected at the level of their Lie algebra as ``infinitesimal symmetries''.
While certainly the unitary groups are most important for QML, we would not put it past these other groups to sneakily appear in a variety of tasks.
Perhaps you will tell us where you have caught them!
\\ \\
\textbf{Continuous groups}
\begin{itemize}
    \item $GL(d; \Cbb)$ and $SL(d; \Cbb)$, the (complex) general and special linear groups
    \begin{itemize}
        \item \textit{Topological info:} Not compact, connected, simply connected
        \item \textit{Lie algebra:} $\mathfrak{gl}(d) = M_d(\Cbb)$, the $d\times d$ complex matrices, and $\mathfrak{sl}(d) = \{X\in \mathfrak{gl}(d) : \Tr[X] = 0\}$, the traceless $d\times d$ complex matrices
        \item \textit{Where you might find them:} Change of bases, so basically everywhere.
        \item \textit{Useful fact:} The complexified Lie algebras $ \mathfrak{u}(d) \otimes \Cbb \cong \mathfrak{gl}(d)\otimes \Cbb \cong \mathfrak{gl}(d)$ and $ \su(d) \otimes \Cbb \cong \mathfrak{sl}(d)\otimes \Cbb \cong \mathfrak{sl}(d)$, which means their (complex) representation theory is the same, even though $\mathfrak{gl}(d)\neq \mathfrak{u}(d)$ and $\mathfrak{sl}(d) \neq \mathfrak{su}(d)$ ($\mathfrak{u}(d),\mathfrak{su}(d)$ are real vector spaces but not complex vector spaces).
    \end{itemize}   
    \item $U(d)$ and $SU(d)$, the unitary and special unitary groups
    \begin{itemize}
        \item \textit{Topological info (for $SU(d)$):} Compact, connected, simply connected
        \begin{itemize}
            \item $U(d)/SU(d) \cong U(1)$, the circle group. 
            The isomorphism follows immediately from $\det:U(d)\to U(1)$.
            $U(d)$ is not simply connected.
        \end{itemize}
        \item \textit{Lie algebra:} $\mathfrak{u}(d) = \{X \in M_d(\Cbb) : X = -X^\dagger\}$, the $d\times d$ skew-hermitian matrices, and $\mathfrak{su}(d) = \{X\in \mathfrak{u}(d) : \Tr[X] = 0\}$, the traceless $d\times d$ skew-hermitian matrices. Note that $\mathfrak{u}(d), \mathfrak{su}(d)$ are real vector spaces: e.g. if $X\in \su(d)$, then $iX\not\in \su(d)$.
        \item \textit{Where you might find them:} Literally everywhere in quantum.
        \item \textit{Useful fact:} Wigner's theorem assures us that every symmetry transformation on physical states which preserves the Hermitian inner product is either a unitary or antiunitary transformation. 
        This in large part motivates the focus on unitary representations from a physical standpoint.
    \end{itemize}   
    \item $O(d)$ and $SO(d)$, the orthogonal and special orthogonal groups
    \begin{itemize}
        \item \textit{Topological info (for $SO(d)$):} Compact, connected, not simply connected
        \begin{itemize}
            \item Note that the orthogonal group $O(d)$ consists of two disconnected copies of $SO(d)$: matrices $\{R: R\in SO(d)\}$, and matrices $\{SR : R\in SO(d) \text{ and } \det(S)=-1\}$. 
            We can think of $O(d)$ as being generated by ``rotations'' ($SO(d)$) and ``hyperplane reflections'', like the dihedral group.
        \end{itemize}
        \item \textit{Lie algebra:} $\mathfrak{o}(d) = \mathfrak{so}(d) = \{X \in M_d(\R) : X = -X^T\}$, the skew-symmetric $d\times d$ matrices. 
        \item \textit{Where you might find them:} Spin systems like AKLT chain, physical rotations and radial potentials, angular momentum, Clifford algebras
        \item \textit{Useful fact:} Even and odd orthogonal groups have significantly different structures and representation theory. 
        This is closely related to the presence of subgroups which control rotations in planes, i.e. subgroups isomorphic to $SO(2)$.
    \end{itemize}   
    \item $Sp(2d,\R)$, the symplectic group
    \begin{itemize}
        \item \textit{Topological info:} Noncompact, connected, not simply connected
        \item \textit{Lie algebra:} $\mathfrak{sp}(2d) = \{X \in M_d(\R) : X^T J + J X = 0\}$, where $J = \begin{pmatrix} 0 & \idty_d \\ -\idty_d & 0 \end{pmatrix}$ is a symplectic matrix, i.e. a linear transformation preserving the symplectic form $\omega(v,w) = \omega(Jv,Jw)$ $\forall v,w \in \R^{2d}$.
        \item \textit{Where you might find them:} Classical Hamiltonian Systems and their quantizations, quantum optics, gaussian states, quasi-free states
        \item \textit{Useful fact:} The symplectic group describes canonical transformations in classical mechanics, i.e. linear maps which preserve Hamilton's equations. 
    \end{itemize}   
\end{itemize}

\end{pabox} 

\end{document}